\documentclass[screen]{lmcs}
\pdfoutput=1

\usepackage{lastpage}
\lmcsdoi{17}{4}{10}
\lmcsheading{}{\pageref{LastPage}}{}{}%
{Nov.~30,~2020}{Nov.~18,~2021}{}

\makeatletter
\renewcommand\paragraph{\@startsection{paragraph}{4}{\z@}%
  {2.25ex \@plus 1ex \@minus .2ex}%
  {-0.75em}%
  {\normalfont\normalsize\bfseries}}
\makeatother

\usepackage{etoolbox}
\newtoggle{arxiv}
\newtoggle{revise}
\toggletrue{revise}

\usepackage[utf8]{inputenc}
\usepackage[utf8]{inputenc}
\usepackage{color}

\usepackage{mathbbol} % for bb-style greek letters

\definecolor{BlueViolet}{rgb}{0, 0, 0.55}
\definecolor{RubineRed}{rgb}{0.88, 0.07, 0.37}
\definecolor{ForestGreen}{rgb}{0.13, 0.55, 0.13}
\definecolor{NavyBlue}{rgb}{0.0, 0.0, 0.5}
\definecolor{Black}{rgb}{0.02, 0.02, 0.02}
\definecolor{MidnightBlue}{rgb}{0.0, 0.2, 0.4}
\definecolor{Gray}{rgb}{0.41, 0.41, 0.41}
\definecolor{TealBlue}{rgb}{0.212,0.459,0.533}
\definecolor{Plum}{rgb}{0.6,0.25,0.6}

\usepackage{adjustbox}

\usepackage{eucal}
\usepackage{centernot}

\usepackage{thm-restate}

\usepackage{siunitx}
\usepackage{textcomp}
\usepackage{caption}

\usepackage{pst-func}
\usepackage{amssymb}
\usepackage{graphicx}

\usepackage{environ}
\usepackage{pgfplots}
\usepackage{tikz}
\usetikzlibrary{shadows,arrows,shapes,snakes,automata,backgrounds,petri}

\usepackage{centernot} % for negated arrows

\usepackage{mathtools}
\usepackage{array,multirow}

\usepackage{breakurl}

\usepackage{pgfplotstable}
\usepackage{mdframed}
\usepackage{threeparttable}

\usepackage{etoolbox}
\usepackage{arydshln}

\usepackage{xspace} % Smart spaces after commands

\usepackage{dirtytalk} % for \say command
\usepackage{stmaryrd} % for llbracket and rrbracket\textbf{\textbf{}}

\usepackage[inline,shortlabels]{enumitem}
\newlist{inlinelist}{enumerate*}{1}
\setlist*[inlinelist,1]{%
  label=(\roman*),
}

\usepackage{xifthen}  % Extended conditional commands
\usepackage{ifthen}

\usepackage{nicefrac}

\usepackage{hyperref}
\usepackage{cleveref}
\usepackage{bigints}

\usepackage{caption,subcaption}     % Support for subfigures

\usepackage{footmisc}

\usetikzlibrary{pgfplots.dateplot}
\usetikzlibrary{matrix,chains,positioning,decorations.pathreplacing}
\usetikzlibrary{patterns,calc,fit,arrows}
\usetikzlibrary{shapes.geometric}

\usepackage[final,nomargin,inline,index]{fixme} % Simplified management of FIXME's
\fxusetheme{color}

\FXRegisterAuthor{bart}{anbart}{\color{magenta} {\underline{bart}}}
\FXRegisterAuthor{lillo}{anlillo}{\color{orange} {\underline{lillo}}}
\FXRegisterAuthor{MM}{anMM}{\color{green} {\underline{Maurizio}}}

\hypersetup{
  breaklinks   = true,
  colorlinks   = true, %Colours links instead of ugly boxes
  urlcolor     = blue, %Colour for external hyperlinks
  linkcolor    = blue, %Colour of internal links
  citecolor    = red   %Colour of citations
}
\hypersetup{final} % Needed to generate hyperlinks even in draft mode (ERROR IN MAC)

\usepackage{listings}
\definecolor{listingBG}{HTML}{FFFFCB}%
\definecolor{listingFrame}{HTML}{BBBB98}%
\definecolor{listingLineno}{rgb}{0.5,0.5,1.0}%
\definecolor{LightGrey}{rgb}{0.975,0.975,0.975}
\lstdefinelanguage{tins}{
	commentstyle=\color{Gray},
	morecomment=[l]{//},
	morecomment=[s]{/*}{*/},
	classoffset=0,
	morekeywords={skip,throw,if,then,else,while,do},
	keywordstyle=\color{Black}\bfseries,
	classoffset=1,
	morekeywords={sender,value,balance,undef},
	keywordstyle=\valColor{}\bfseries\itshape,
	classoffset=2,
	morekeywords={contract},
	keywordstyle=\funColor{}\bfseries,
        extendedchars=true,
        literate={\$}{{\textbf{{\dollar}}}}1
}

\definecolor{LightGrey}{rgb}{0.975,0.975,0.975}
\lstdefinelanguage{solidity}{
	commentstyle=\color{Gray},
	morecomment=[l]{//},
	morecomment=[s]{/*}{*/},
	classoffset=0,
        escapechar=\$,
	morekeywords={struct,mapping,function,this,public,private,static,final,class,extends,switch,case,break,finally,try,catch,return,if,else,new},
	keywordstyle=\color{NavyBlue}\bfseries,
	classoffset=1,
	morekeywords={unit,int,string,bool,address,uint,uint256},
	keywordstyle=\color{TealBlue},
	classoffset=2,
	morekeywords={ether,wei,finney,contract,send,throw,msg,sender,value},
	keywordstyle=\color{Plum}\bfseries,
}

\usepackage{tikz-cd}
\usetikzlibrary{cd,decorations.markings}

\newcommand{\ifempty}[3]{%
  \ifthenelse{\isempty{#1}}{#2}{#3}%
}

\newcommand{\ifzero}[3]{%
  \ifthenelse{\equal{#1}{0}}{#2}{#3}%
}

\newcommand{\ifdots}[3]{%
  \ifthenelse{\equal{#1}{...}}{#2}{#3}%
}

\newcommand{\hidden}[1]{}

\newenvironment{nscenter}
 {\parskip=0pt\par\nopagebreak\centering}
 {\parskip=2pt\par\noindent} % \ignorespacesafterend

\newcommand{\eqdef}{\triangleq}

\renewcommand{\vec}[1]{\boldsymbol{#1}}
\newcommand{\relR}{\mathcal{R}}

\newcommand{\Real}[1]{\mathrm{Real}}

\newcommand{\codefont}{\fontsize{10}{11}\selectfont}
\newcommand{\code}[1]{{\tt\codefont {#1}}}
\newcommand{\codeAddr}[1]{\ensuremath{\code{\aColor{#1}}}}
\newcommand{\codeVal}[1]{\ensuremath{\code{\valColor{#1}}}}
\newcommand{\codeFun}[1]{\ensuremath{\code{\funColor{#1}}}}

\newcommand{\dollar}{{\textup{\texttt{\symbol{`\$}}}}}
\newcommand{\codeand}{\textup{\texttt{\symbol{`\&}\symbol{`\&}}}}

\def\etc{etc.\@\xspace}

\newcommand{\eg}{e.g.\@\xspace}
\newcommand{\ie}{i.e.\@\xspace}
\newcommand{\wrt}{w.r.t.\@\xspace}

\newcommand{\emptyseq}{\varepsilon}

\newcommand{\utxo}[1][]{\mathit{UTXO}_{#1}}

\newenvironment{proofof}[2]{%
  \subsection*{Proof of {#1}~\ref{#2}}
  \label{#2-proof}
  }%
  {\qed}

\newcommand{\sig}[3][]{\mathit{sig}^{#1}_{#2}\ifempty{#3}{}{({#3})}}

\newcommand{\BTC}{\textup{%
  \leavevmode
  \vtop{\offinterlineskip %\bfseries
    \setbox0=\hbox{B}%
    \setbox2=\hbox to\wd0{\hfil\hskip-.03em
    \vrule height .3ex width .15ex\hskip .08em
    \vrule height .3ex width .15ex\hfil}
    \vbox{\copy2\box0}\box2}}\xspace}

\newcommand{\ether}{\textit{ether}\xspace}

\newcommand{\true}{\code{true}\xspace}

\def\aColor{\color{ForestGreen}}
\def\pColor{\color{ForestGreen}}
\def\cColor{\color{ForestGreen}}

\newcommand{\Addr}{{\aColor{\textup{\textbf{Addr}}}}}
\newcommand{\Val}{{\valColor{\mathbb{V}}}} % universe of all values
\newcommand{\Obs}{\mathbb{O}} % universe of all observables
\newcommand{\QmvU}[1][]{\Obs}
\newcommand{\Const}{{\valColor{\textup{\textbf{Const}}}}}
\newcommand{\Tx}{{\txColor{\mathbb{T}}}}
\newcommand{\BcSt}{{\mathbb{\Sigma}}}

\newcommand{\addrToContr}[1]{{\Gamma}\ifempty{#1}{}{({#1})}}

\newcommand{\aFmt}[1]{{\aColor{\mathcal{#1}}}}
\newcommand{\pFmt}[1]{{\pColor{\mathcal{#1}}}}
\newcommand{\cFmt}[1]{{\cColor{\mathcal{#1}}}}

\newcommand{\amv}[2][]{\aFmt{#2}_{\aColor{#1}}\xspace}
\newcommand{\amvA}[1][]{\amv[{#1}]{X}}
\newcommand{\amvB}[1][]{\amv[{#1}]{Y}}

\newcommand{\pmv}[2][]{\pFmt{#2}_{\pColor{#1}}\xspace}
\newcommand{\pmvA}[1][]{\pmv[{#1}]{A}}
\newcommand{\pmvB}[1][]{\pmv[{#1}]{B}}

\newcommand{\cmv}[2][]{\cFmt{#2}_{\cColor{#1}}\xspace}
\newcommand{\cmvA}[1][]{\cmv[{#1}]{C}}
\newcommand{\cmvB}[1][]{\cmv[{#1}]{D}}

\newcommand{\qmv}[2][]{{#2}_{#1}}
\newcommand{\qmvA}[1][]{\qmv[{#1}]{p}}

\newcommand{\qmvB}[1][]{\qmv[{#1}]{q}}

\newcommand{\QmvA}[1][]{P_{#1}}
\newcommand{\QmvAi}[1][]{P'_{#1}}
\newcommand{\QmvB}[1][]{Q_{#1}}

\newcommand{\QmvC}[1][]{R_{#1}}

\def\txColor{\color{MidnightBlue}}
\newcommand{\txFmt}[1]{{\txColor{\sf #1}}}

\newcommand{\tx}[2][]{\txFmt{#2}_{\txColor{#1}}}
\newcommand{\txT}[1][]{\tx[#1]{T}}
\newcommand{\txTi}[1][]{\txFmt{T'_{\txColor{{#1}}}}}

\DeclareMathAlphabet{\mathbfsf}{\encodingdefault}{\sfdefault}{bx}{n}

\newcommand{\bcEmpty}{{\mathbfsf{\txColor{\emptyseq}}}}
\newcommand{\bcB}[1][]{{\boldsymbol{\mathsf{\txColor{B}}}}_{\txColor{#1}}}
\newcommand{\bcBi}[1][]{{\boldsymbol{\txColor{\sf B'_{\txColor{\mathrm{\textup{#1}}}}}}}}

\newcommand{\valid}{\rhd}
\newcommand{\nvalid}{\centernot{\rhd}}

\newcommand{\bcSt}[1][]{\sigma_{#1}}
\newcommand{\bcSti}[1][]{\sigma_{#1}'}
\newcommand{\bcStii}[1][]{\sigma_{#1}''}
\newcommand{\bcStInit}{\sigma_{0}}

\newcommand{\bcEnv}[1][]{\rho_{#1}}

\newcommand{\balanceInit}[1]{\valN^{0}_{#1}}

\newcommand{\mrg}{\oplus}
\newcommand{\mSubst}[1][]{\pi_{#1}}
\newcommand{\mSubsti}[1][]{\pi'_{#1}}
\newcommand{\WR}[1]{\ifempty{#1}{\Pi}{\Pi(#1)}}
\newcommand{\WRi}[1]{\ifempty{#1}{\Pi'}{\Pi'(#1)}}
\newcommand{\WRmin}[1]{\ifempty{#1}{\Pi^{\star}}{\Pi^{\star}(#1)}}

\def\fieldColor{\color{Plum}}
\def\scriptColor{\color{Black}}
\newcommand{\script}[2][]{{\scriptColor{{\it #2}_{#1}}}}
\newcommand{\expe}[1][]{\script[{#1}]{e}}
\newcommand{\expei}[1][]{\script[{#1}]{e'}}

\newcommand{\versigName}{{\sf versig}}
\newcommand{\versig}[2]{\versigName({#1},{#2})}
\newcommand{\rtx}{{\sf rtx}}
\newcommand{\rtxWit}{\rtx.\txWit[]{}}
\newcommand{\hashE}[1]{{\sf H}(#1)}

\newcommand{\sizeE}[1]{\ensuremath{| #1 |}}
\newcommand{\ifE}[3]{\mathsf{if}~{#1}~\mathsf{then}~{#2}~\mathsf{else}~{#3}}

\newcommand{\const}[2][]{#2_{#1}} % constants

\newcommand{\constPK}[1][]{\const[{#1}]{pk}}
\newcommand{\constSK}[1][]{\const[{#1}]{sk}}

\newcommand{\txTag}[3][]{{\fieldColor\sf #3}\ifempty{#1}{\ifempty{#2}{}{: {#2}}}{({#1})\ifempty{#2}{}{: {#2}}}}

\newcommand{\txIn}[2][]{\txTag[{#1}]{#2}{in}}
\newcommand{\txWit}[2][]{\txTag[{#1}]{#2}{wit}}
\newcommand{\txOut}[2][]{\txTag[{#1}]{#2}{out}}

\newcommand{\txf}{\txTag{}{f}} % transaction field

\newcommand{\txscript}{\txTag{}{scr}}
\newcommand{\txval}{\txTag{}{val}}

\newcommand{\trans}[1]{\xrightarrow{#1}}
\newcommand{\nottrans}[1]{\centernot{\xrightarrow{#1}}}

\newcommand{\pre}[1]{{}^{\bullet}{#1}}

\newcommand{\post}[1]{{#1}{{}^{\bullet}}}

\newcommand{\irule}[2]{\dfrac{#1}{#2}}

\newcommand{\mapstopart}{\rightharpoonup}

\newcommand{\bnfdef}{::=}
\newcommand{\bnfmid}{\;|\;}

\newcommand{\nrule}[1]{{\scriptsize \textsc{#1}}}

\newcommand{\sem}[2][]{\mbox{\ensuremath{\llbracket{#2}\rrbracket_{#1}}}}

\newcommand{\msem}[3]{\mbox{\ensuremath{\llbracket{#3}\rrbracket^{#1}_{#2}}}}

\newcommand{\semCmd}[3]{\mbox{\ensuremath{\llbracket{#1}\rrbracket_{#2}^{#3}}}}
\newcommand{\semTx}[2]{\sem[{#2}]{#1}}
\newcommand{\semBc}[2]{\sem[{#2}]{#1}}

\newcommand{\equivStSeq}[1][]{\simeq_{#1}}
\newcommand{\equivSeq}{\simeq}

\newcommand{\dom}[1]{\operatorname{dom} {#1}}
\newcommand{\keys}[1]{\dom{#1}}

\newcommand{\Nat}{\mathbb{N}}

\newcommand{\bind}[2]{\nicefrac{#2}{#1}}
\newcommand{\setenum}[1]{\{#1\}}
\newcommand{\setcomp}[2]{\left\{{#1} \,\middle|\, {#2}\right\}}
\newcommand{\emptymset}{[]}
\newcommand{\msetenum}[1]{\lbrack{#1}\rbrack}

\newcommand{\card}[1]{|#1|}
\newcommand{\seqat}[2]{{#1}.{#2}}

\crefname{lem}{lemma}{lemmas}
\Crefname{lem}{Lemma}{Lemmas}
\crefname{thm}{theorem}{theorems}
\Crefname{thm}{Theorem}{Theorems}

\crefname{appendix}{appendix}{appendices}
\Crefname{appendix}{Appendix}{Appendices}
\crefname{notation}{notation}{notations}
\Crefname{notation}{Notation}{Notations}

\definecolor{LightGrey}{rgb}{0.95,0.95,0.95}
\definecolor{keyword}{HTML}{7F0055}

\newlength\replength
\newcommand\repfrac{.1}

\newcommand\rulewidth{.6pt}
\newcommand\tdashfill[1][\repfrac]{\cleaders\hbox to \replength{%
  \smash{\rule[\arraystretch\ht\strutbox]{\repfrac\replength}{\rulewidth}}}\hfill}

\newcommand\tdotfill[1][\repfrac]{\cleaders\hbox to \replength{%
  \smash{\raisebox{\arraystretch\dimexpr\ht\strutbox-.1ex\relax}{.}}}\hfill}

\newcommand{\var}[2][]{#2_{#1}} % variables
\newcommand{\varX}[1][]{\var[#1]{x}} 
 
\newcommand{\varY}[1][]{\var[#1]{y}}

\def\funColor{\color{RubineRed}}
\newcommand{\funFmt}[1]{{\funColor{\mathtt{#1}}}}

\newcommand{\funF}[1][]{\funFmt{f}_{\funColor{#1}}} 
 
\newcommand{\funG}[1][]{\funFmt{g}_{\funColor{#1}}} 
 
\newcommand{\funH}[1][]{\funFmt{h}_{\funColor{#1}}}

\def\valColor{\color{magenta}}
\newcommand{\val}[2][]{{\valColor{#2_{#1}}}} % values
\newcommand{\valV}[1][]{\val[#1]{v}}
\newcommand{\valVi}[1][]{\val[#1]{v'}}
\newcommand{\valN}[1][]{\val[#1]{n}}
\newcommand{\valNi}[1][]{\val[#1]{n'}}
\newcommand{\valK}[1][]{\val[#1]{k}}
\newcommand{\valKi}[1][]{\val[#1]{k'}}
\newcommand{\valX}[1][]{\val[#1]{x}}
\newcommand{\valY}[1][]{\val[#1]{y}}

\newcommand{\valZ}[1][]{\val[#1]{z}}

\newcommand{\obsA}[1][]{\qmvA[#1]}
\newcommand{\obsB}[1][]{\qmvB[#1]}

\def\cmdColor{\color{RubineRed}}
\newcommand{\cmdFmt}[1]{{\cmdColor{\mathit{#1}}}}
\newcommand{\cmdC}[1][]{\mathord{\cmdFmt{S}_{\cmdColor{#1}}}}

\newcommand{\cmdSkip}{\ensuremath{\code{skip}}}
\newcommand{\cmdAss}[2]{{#1} \code{:=} {#2}}
\newcommand{\cmdIfTE}[3]{\code{if}\, {#1}\, \code{then} \, {#2} \, \code{else} \, {#3}}
\newcommand{\cmdIfT}[2]{\code{if} \, {#1} \, \code{then} \, {#2}}

\newcommand{\cmdCall}[4][]{{#2}\ifempty{#3}{}{:{#3}({#4})}\ifempty{#1}{}{\dollar {#1}}}
\newcommand{\cmdSend}[2]{{#2}.\texttt{transfer}({#1}){}{}}
\newcommand{\cmdThrow}{\code{throw}\xspace}

\newcommand{\expGet}[2]{{#1}.{#2}} %%% deprecated

\newcommand{\expLookup}[1]{{#1}}

\newcommand{\lineno}[1]{{\tt\codefont {\textcolor{ForestGreen}{#1}}}}

\definecolor{LightGrey}{rgb}{0.95,0.95,0.95}
\definecolor{keyword}{HTML}{7F0055}

\newmdenv[linewidth=0pt]{mdNoFramed}
\makeatletter

\newcommand*{\tabminted@finalstrut}[1]{%
  \ifdim\prevdepth>0pt
    \ifdim\dp#1>\prevdepth
      \vskip\dimexpr(\dp#1)-\prevdepth\relax
    \fi
  \else
    \vskip\dimexpr(\dp#1)\relax
  \fi
}
\newcommand*{\@tabmintedend}{%
  \let\@finalstrut\tabminted@finalstrut
}
\makeatother

\newcommand{\netN}{\sf{N}}
\newcommand{\Places}{{\sf P}}

\newcommand{\Transitions}{{\sf Tr}}
\newcommand{\Arcs}{{\sf{F}}}
\newcommand{\markM}[1][]{{\sf{m}_{#1}}}
\newcommand{\markMi}[1][]{{\sf{m}'_{#1}}}
\newcommand{\markMii}[1][]{{\sf{m}''_{#1}}}
\newcommand{\markMiii}[1][]{{\sf{m}}'''_{#1}}
\newcommand{\trT}[1][]{\mathsf{t}_{#1}}
\newcommand{\trTi}[1][]{{{\sf{t}}_{#1}'}}
\newcommand{\trTii}[1][]{{{\sf{t}}_{#1}''}}
\newcommand{\plP}[1][]{{\sf{p_{#1}}}}

\newcommand{\placeP}[1][]{{\sf{p_{#1}}}}
\newcommand{\placePi}[1][]{{\sf{p'_{#1}}}}
\newcommand{\PNet}[2]{\ifempty{#2}{{\sf N}_{{#1}}}{{\sf N}_{#1}(#2)}}
\newcommand{\trSU}[1][]{{\sf U}_{#1}}
\newcommand{\trSUi}[1][]{{\sf U'}_{#1}}
\newcommand{\trSUii}[1][]{{\sf U''}_{#1}}

\newcommand{\varBalance}{\ensuremath{\codeVal{balance}}\xspace}
\newcommand{\varSender}{\ensuremath{\codeVal{sender}}\xspace}

\newcommand{\varValue}{\ensuremath{\codeVal{value}}\xspace}

\newcommand{\ethtx}[5]{{#2}\xrightarrow{#1} {#3}:{#4}({#5})}

\newcommand{\contrFun}[3]{{#1}({#2}) \{ {#3} \}}
\newcommand{\contrFunSig}[2]{{#1}({#2})}

\newcommand{\swap}{\rightleftarrows}
\newcommand{\nswap}{\mbox{\ensuremath{\,\not\rightleftarrows\,}}}

\newcommand{\safeapprox}[3][]{{#2} \models^{#1} {#3}}
\newcommand{\wapprox}[2]{\safeapprox[w]{#1}{#2}}
\newcommand{\rapprox}[2]{\safeapprox[r]{#1}{#2}}

\newcommand{\pswap}[2]{\#^{#1}_{#2}}
\newcommand{\pswapWR}{\pswap{\wset{}}{\rset{}}}

\newcommand{\rset}[1]{\mathtt{R}\ifempty{#1}{}{({#1})}}

\newcommand{\wset}[1]{\mathtt{W}\ifempty{#1}{}{({#1})}}

\newcommand{\txsA}[1][]{{\txColor{\mathbb{T}}_{\txColor{#1}}}}
\newcommand{\txsAi}[1][]{{\txColor{\mathbb{T}'}_{\txColor{#1}}}}

\newcommand{\bcsA}[1][]{{\txColor{\mathbb{B}}_{\txColor{#1}}}}

\newcommand{\seqn}{\vartriangleleft}

\newcommand{\independent}{\;\mathrm{I}\;}

\newcommand{\txOfTr}[1]{\alpha\ifempty{#1}{}{({#1})}}
\newcommand{\trOfTrSU}[1]{{\it tr}\ifempty{#1}{}{({#1})}}

\newcommand{\mytitle}{A theory of transaction parallelism in blockchains}

\begin{document}

\title{\mytitle}

\author[M.~Bartoletti]{Massimo Bartoletti}
\address{University of Cagliari, Italy}
\email{bart@unica.it}
%\thanks{thanks 1, optional.}	%optional

\author[L.~Galletta]{Letterio Galletta}
\address{IMT School for Advanced Studies, Lucca, Italy}	%optional
\email{letterio.galletta@imtlucca.it}

\author[M. Murgia]{Maurizio Murgia}
\address{University of Trento, Italy}
\email{maurizio.murgia@unitn.it}

\maketitle

\begin{abstract}
  Decentralized blockchain platforms have enabled 
  the secure exchange of crypto-assets without the intermediation
  of trusted authorities.
  To this purpose, these platforms rely on a peer-to-peer network 
  of byzantine nodes, which collaboratively maintain
  an append-only ledger of transactions, called \emph{blockchain}.
  Transactions represent the actions required by users,
  \eg the transfer of some units of crypto-currency to another user,
  or the execution of a smart contract which distributes crypto-assets 
  according to its internal logic.
  Part of the nodes of the peer-to-peer network compete to
  append transactions to the blockchain.
  To do so, they group the transactions sent by users into \emph{blocks},
  and update their view of the blockchain state by 
  executing these transactions in the chosen order.
  Once a block of transactions is appended to the blockchain, 
  the other nodes validate it, 
  re-executing the transactions in the same order.
  The serial execution of transactions does not take advantage of 
  the multi-core architecture of modern processors, 
  so contributing to limit the throughput.
  In this paper we develop a theory of transaction parallelism for blockchains,
  which is based on static analysis of transactions and smart contracts.
  We illustrate how blockchain nodes can use our theory
  to parallelize the execution of transactions.
  Initial experiments on Ethereum show that our technique
  can improve the performance of nodes.
\end{abstract}

\section{Introduction}
\label{sec:intro}

Decentralized blockchain platforms like Bitcoin and Ethereum 
allow mutually untrusted users to create and exchange crypto-assets, 
without resorting to trusted intermediaries.
These exchanges can be either 
simple transfers of an asset from one user to another one,
or they can be the result of executing complex protocols, 
called \emph{smart contracts}.
All the actions performed by users
are recorded on a public data structure, called \emph{blockchain},
from which everyone can infer the amount of crypto-assets owned by each user.
The disintermediation stems from the fact that maintaining the blockchain
does not depend on trusted authorities: rather, this task 
is collaboratively performed by a peer-to-peer network, 
following a complex consensus protocol
which guarantees the consistency of the blockchain
also in the presence of (a minority of) adversaries in the network.

Users interact with the blockchain by sending \emph{transactions},
which may request direct transfers of crypto-assets, 
or invoke smart contracts which in turn trigger transfers
according to the programmed logic.
The sequence of transactions on the blockchain determines,
besides the balance of each user, the state of each smart contract.
The nodes of the peer-to-peer network process the transactions sent by users,
playing either the role of \emph{miner} or that of \emph{validator}.
Miners group transactions into blocks, 
execute them \emph{serially} to determine the new blockchain state,
and append blocks to the blockchain.
Validators read blocks, and re-execute their transactions 
to update their local view of the blockchain state.
To do this, validators process transactions exactly in the same order 
in which they occur in the block, 
since choosing a different order could potentially
result in inconsistencies between the nodes.

Executing transactions in a purely sequential fashion
is quite effective to ensure the consistency of the blockchain state, 
but in the age of multi-core processors 
it fails to properly exploit the computational capabilities of nodes.
By enabling miners and validators to concurrently execute transactions, 
it would be  possible to improve the efficiency and the throughput 
of the blockchain.
Although there exist a few works that address this problem
(we discuss them in~\Cref{sec:related} below),
their approach is eminently empirical,
and they are focussed only on Ethereum.
A comprehensive study of the theoretical foundations of transaction
parallelism in blockchains would improve the understanding of
these optimizations, and it would allow to extend them to other
blockchains beyond Ethereum.

\subsection{Contributions}

This paper exploits techniques from concurrency theory to 
provide a formal backbone for parallel execution of transactions
in blockchains.
More specifically, our main contributions can be summarised as follows:
\begin{itemize}

\item We introduce a general model of blockchain platforms, 
  parameterized over the observables and the semantics of transactions
  (\Cref{sec:transactions}).
  Building upon it, we define the semantics of a blockchain 
  by iterating the semantics of its transactions:
  this reflects the standard implementation of nodes,
  where transactions are evaluated in sequence, without any concurrency.
  We show that the two most widespread blockchain platforms,
  \ie Bitcoin and Ethereum, can be expressed as an instance 
  of this general model.

\item We introduce two notions of \emph{swappability} of transactions
  (\Cref{sec:txswap}).
  The first one is extensional: two adjacent transactions can be swapped
  if this preserves the blockchain state.
  The second notion --- \emph{strong} swappability --- is intensional:
  two adjacent transactions can be swapped is the static approximations 
  of their read/written observables satisfy a simple condition,
  inspired by Bernstein's conditions for the parallel execution of processes.
  Basically, these conditions require that the observables written
  by a transaction are not read or written by the other transaction.
  \Cref{th:pswap-implies-swap} shows that the strong swappability relation
  is included in the extensional relation.
  \Cref{th:pswapWR:mazurkiewicz} shows that, if we
  repeatedly exchange adjacent strongly swappable transactions,
  the resulting blockchain is observationally equivalent to the original one.

\item For Bitcoin, we show
  that the static approximations checked by the strong swappability 
  condition can be easily inferred by transactions:
  the least approximations of the written observables 
  are the transaction inputs and outputs,
  while those of the read observables are the transaction inputs
  (Lemma~\ref{lem:safeapprox:btc}).
  For Ethereum obtaining precise approximations is more complex, 
  because of its Turing-complete contract language.
  We discuss in~\Cref{sec:txswap:eth} a few tricky cases,
  and we report in~\Cref{sec:validation} our experience 
  with a novel tool to statically detect swappable Ethereum transactions.
  We further show that, for both Bitcoin and Ethereum,
  strong swappability is stricter then swappability
  (Examples~\ref{cex:pswap-implies-swap:btc} and~\ref{cex:pswap-implies-swap:eth}).

\item Building upon strong swappability, we devise
  a true concurrent model of transaction execution
  (\Cref{sec:txpar}).
  To this purpose, we transform a block of transactions into an
  \emph{occurrence net}, 
  describing exactly the partial order induced by the swappability relation.
  We model the concurrent executions of a blockchain
  in terms of the \emph{step firing sequences} 
  (\ie finite sequences of \emph{sets} of transitions)
  of the associated occurrence net.
  In Theorem~\ref{th:bc-to-pnet} we establish that 
  the concurrent executions are semantically equivalent to the serial one.

\item Finally, we describe how miners and validators can use our results to
  parallelize transactions, exploiting their multi-core architecture
  (\Cref{sec:validation}). 
  An initial experimental validation of our technique on Ethereum, 
  which exploits a novel static analyser of Ethereum bytecode,
  shows that there are margins to make it applicable in practice.

\end{itemize}

  \subsection{Overview of the approach: ERC-721 tokens}
\label{sec:erc721}

We illustrate the main elements of our theory by considering an archetypal Ethereum smart contract, which implements a ``non-fungible token''.
A non-fungible token represents a digital version of real-world assets, \eg access keys, pieces of arts, and 
serves as verifiable proof of authenticity and ownership within a blockchain network. 
This kind of contracts are quite relevant: 
currently, token transfers involve $\sim{50}\%$ of the 
transactions on the Ethereum blockchain~\cite{tokens},
with larger peaks due to popular contracts like Cryptokitties~\cite{Young17cointelegraph}.

We sketch the implementation of the \code{Token} contract (the full code is in the Appendix), 
using Solidity, the main high-level smart contract language in Ethereum.
This contract follows the standard ERC-721 interface~\cite{ERC721,Frowis19fc} and defines 
functions to transfer tokens between users, and to delegate their trade to other users.

In Ethereum, a smart contract is similar to an object in an object-oriented
language: it has an internal state, and a set of functions to manipulate it.
Users and contracts are identified by their \emph{addresses}.

The state of the contract \code{Token} is defined by the following mappings:
\begin{lstlisting}[language=solidity]
mapping(uint256 => address) owner;
mapping(uint256 => bool) exists;
mapping(address => uint256) balance;
mapping(address => mapping (address => bool)) operatorApprovals;
\end{lstlisting}

Tokens are uniquely identified by an integer value (of type \code{uint256}),
while users are identified by an \code{address}
(the address $0$ denotes a dummy owner).
The mapping \code{owner} associates tokens to their owners' addresses,
\code{exists} tells whether a token has been created or not,
and \code{balance} gives the number of tokens owned by each user.
The mapping \code{operatorApprovals} allows users to delegate 
the transfer of their tokens to third parties.

The following function \code{transferFrom} transfers a token from the owner to another user:
\begin{lstlisting}[language=solidity,numbersep=10pt,numbers=left]
function transferFrom(address from, address to, uint256 id) external {
  require (exists[id] && from==owner[id] && from!=to && to!=address(0));      
  if (from==msg.sender || operatorApprovals[from][msg.sender]) {
    owner[id] = to;
    balance[from] -= 1;
    balance[to] += 1;
  }
}
\end{lstlisting}

The \code{require} assertion at line~\lineno{2} rules out some undesirable cases,
\eg, if the token does not exist, 
or it is not owned by the \code{from} user,
or the user attempts to transfer the token to himself.
Once all these checks are passed, the transfer succeeds if 
the \code{sender} of the transaction owns the token, or if he has been delegated by the owner (line~\lineno{3}).
The mappings \code{owner} and \code{balance} are updated as expected
(lines~\lineno{4-6}).

The function \code{setApprovalForAll} delegates the transfers of all the tokens of the \code{sender} to the \code{operator} 
when the boolean \code{isApproved} is true, 
otherwise it revokes the delegation:
\begin{lstlisting}[language=solidity,numbersep=10pt]
function setApprovalForAll(address operator, bool isApproved) external {
  operatorApprovals[msg.sender][operator] = isApproved;
}
\end{lstlisting}

Users interact with contracts by sending transactions to the blockchain.
Transactions involve the execution of smart contract functions that may trigger contracts updates and transfer 
of crypto-currency from the caller to the callee.
For example, consider a user (with address) $\pmvA$, 
which owns two tokens identified by the integers $1$ and $2$,
and consider the following transactions:
\begin{align*}
  & \txT[1] = \ethtx{}{\pmvA}{\codeAddr{Token}}
    {\codeFun{transferFrom}}{\pmvA,\pmv{P},1} 
  \\
  & \txT[2] = \ethtx{}{\pmvA}{\codeAddr{Token}}
    {\codeFun{setApprovalForAll}}{\pmvB,\true} 
  \\
  & \txT[3] = \ethtx{}{\pmvB}{\codeAddr{Token}}
    {\codeFun{transferFrom}}{\pmvA,\pmv{Q},2}
  \\
  & \txT[4] = \ethtx{}{\pmv{P}}{\codeAddr{Token}}
    {\codeFun{transferFrom}}{\pmv{P},\pmvB,1}
\end{align*}

Intuitively, transaction $\txT[1]$ means that $\pmvA$ (the sender) 
calls the function $\codeFun{transferFrom}$ of the $\codeAddr{Token}$ contract
to transfer the ownership of token $1$ to user $\pmv{P}$.
Transaction $\txT[2]$ delegates user $\pmv{P}$ to manage $\pmvA$'s tokens.
Transaction $\txT[3]$ says that $\pmvB$ transfers token $2$ from $\pmvA$ to $\pmv{Q}$; 
$\txT[4]$ means that user $\pmv{P}$ transfers token $1$ to $\pmvB$.

Since each transaction modifies the internal state 
of the contract $\codeAddr{Token}$, the order in which a miner 
executes them is relevant.
For example, executing the sequence of transactions 
$\bcB = \txT[1] \txT[2] \txT[3] \txT[4]$ 
results in a state where $\pmvB$ owns token $1$, and $\pmv{Q}$ owns token $2$.  
It is easy to see that $\txT[3]$ can only succeed 
if executed after $\txT[2]$, because it depends on the fact that 
$\pmvB$ is delegated by $\pmvA$, 
\ie $\codeVal{operatorApprovals}[\pmvA][\pmvB]$ is $\codeVal{true}$.   
Therefore, to run in parallel the transactions of $\bcB$, 
a miner would need to find an execution schedule 
that does not affect the resulting state.
Our notion of \emph{swappability} formalizes this intuition: 
two transactions $\txT$ and $\txT'$ are swappable 
if they result in the same state, independently of their order
(Definition~\ref{def:swap}). 
For example, consider the transactions $\txT[1]$ and $\txT[2]$ above: 
regardless of whether $\txT[1]$ is executed before or after $\txT[2]$, after their execution we obtain a state where token $1$ is owned by $\pmv{P}$, 
and $\pmvB$ can act as delegate of $\pmvA$.
   
Clearly, the notion of swappability outlined above is undecidable 
whenever the contract language is Turing-equivalent, 
like in the case of Ethereum.
Therefore, swappability cannot be directly used by a miner 
to determine a parallel execution schedule.
We overcome this issue by resorting to a static analysis of the smart contract. 
The underlying idea is to derive a syntactic approximation of swappability, 
called \emph{strong swappability} (Definition~\ref{def:pswap}).
This captures the fact that two transactions $\txT$ and $\txT'$ depend and affect different portions of a contract state.
Thus, such transactions can be run in any order.
For example, the transactions $\txT[1]$ and $\txT[2]$ above depend on and modify different parts of the state of the $\codeAddr{Token}$ contract:
therefore, they are strongly swappable.

To detect if two transactions $\txT$ and $\txT'$ are strongly swappable (in symbols $\txT \pswap{}{} \txT'$), one needs to statically over-approximate 
the state variables that may be read and written during the execution 
of the called functions.
They are strongly swappable if the set of variables written by $\txT$ is disjoint from those written and read by $\txT'$ and vice versa. 
This ensures that their executions are not interfering with each other.

From the code of \code{transferFrom}, we see that
$\txT[1]$ updates the $\codeVal{owner}$ of token $1$
and the $\codeVal{balance}$ of addresses $\pmvA$ and $\pmv{P}$
(lines~\lineno{4-6}). 
The variables read by $\txT[1]$ are 
$\codeVal{exists[1]}$, $\codeVal{owner[1]}$ (line~\lineno{2}), 
$\codeVal{operatorApprovals[\pmvA][\pmvA]}$ (line~\lineno{3}), 
$\codeVal{balance[\pmvA]}$ (line~\lineno{5}), and
$\codeVal{balance[\pmv{P}]}$ (line~\lineno{6}).
Transaction $\txT[2]$ updates $\codeVal{operatorApprovals[\pmvA][\pmvB]}$.
Using the same reasoning for $\txT[3]$ and $\txT[4]$,
we obtain the following over-approximations of the 
state variables written/read by $\txT[1]$--$\txT[4]$
(we denote with $W^i$ and $R^i$ the variables written and read by $\txT[i]$, respectively):
\[
\begin{array}{rcll}
    R^1
    & = & \setenum{\codeVal{exists[1]}, \codeVal{owner[1]}, \codeVal{balance[\pmvA]}, \codeVal{balance[\pmv{P}]},\codeVal{operatorApprovals[\pmvA][\pmvA]}}
    \\
    W^1
    & = & \setenum{\codeVal{owner[1]}, \codeVal{balance[\pmvA]}, \codeVal{balance[\pmv{P}]}} 
    \\[4pt]
    R^2
    & = & \emptyset
    \\
    W^2
    & = & \setenum{\codeVal{operatorApprovals[\pmvA][\pmvB]}} 
    \\[4pt]
    R^3
    & = & \setenum{\codeVal{exists[2]}, \codeVal{owner[2]}, \codeVal{balance[\pmvA]}, \codeVal{balance[\pmv{Q}]}, \codeVal{operatorApprovals[\pmvA][\pmvB]}}
    \\
    W^3
    & = & \setenum{\codeVal{owner[2]}, \codeVal{balance[\pmvA]}, \codeVal{balance[\pmv{Q}]}} 
    \\[4pt]
    R^4
    & = & \setenum{\codeVal{exists[1]}, \codeVal{owner[1]}, \codeVal{balance[\pmv{P}]}, \codeVal{balance[\pmvB]}, 
    \codeVal{operatorApprovals[\pmv{P}][\pmv{P}]}}
    \\
    W^4
    & = & \setenum{\codeVal{owner[1]}, \codeVal{balance[\pmv{P}]}, \codeVal{balance[\pmv{B}]}} 
    \\
\end{array}
\]

By the approximations above, we have that $\txT[1] \pswap{}{} \txT[2]$, because $(R^1 \cup W^1) \cap W^2 = \emptyset = (R^2 \cup W^2)\cap W^1$. 
Similarly, it is straightforward to see that
$\txT[2] \pswap{}{} \txT[4]$ and $\txT[3] \pswap{}{} \txT[4]$,
while the other combinations are \emph{not} strongly swappable.

The strong swappability relation induces a partial order between transactions: 
this can be exploited by a blockchain node to choose a parallel execution schedule.
To do that, from a given sequence of transactions, 
we build an \emph{occurrence net}~\cite{Best87tcs}, 
a special kind of Petri net with no cycles and where places can hold 
at most 1 mark. 
This net encodes the partial order induced by the swappability relation and formalizes the concurrent semantics of transactions.
One of our main results is that 
any concurrent execution in the occurrence net is equivalent to the serial one
(see item~\ref{th:bc-to-pnet:maximal} of \Cref{th:bc-to-pnet}).

Consider again the the sequence of transactions $\bcB = \txT[1] \txT[2] \txT[3] \txT[4]$ above. 
The associated occurrence net is displayed in \Cref{fig:erc721:petri}.
Intuitively, each transaction in $\bcB$ corresponds, in the net, 
to a transition (rendered as a box), 
linked to two places (rendered as circles).
A transition can fire when all its incoming places contain the mark.
If two transactions are not strongly swappable, the corresponding 
transitions are linked through a place. 
In \Cref{fig:erc721:petri}, since $\txT[1]$ and $\txT[3]$ are not strongly swappable, the place between $\trT[1]$ and $\trT[3]$ ensures that $\trT[3]$ can be executed only after $\trT[1]$, so rendering the dependency implicitly defined in $\bcB$.
The same holds for $\txT[2]$ and $\txT[3]$, and for $\txT[1]$ and $\txT[4]$.  
Instead, transitions corresponding to strongly swappable transactions 
can be fired concurrently.
In our example, this is the case for $\trT[1]$ and $\trT[2]$ 
(since $\txT[1] \pswap{}{} \txT[2]$), as well as for 
$\trT[3]$ and $\trT[4]$ (since $\txT[3] \pswap{}{} \txT[4]$).

Although in our example we have considered the tricky case where
the sender and the receiver of tokens overlap, 
in practice this is a marginal case:
in Ethereum, the large majority of transactions in a block 
either involve distinct users, or invoke distinct ERC-721 interfaces.%
\footnote{%
  Although we are not aware of any work to support this claim,
  some empirical evidence can be obtained by inspecting
  the token-related transactions in \url{https://etherscan.io/tokentxns},
  which shows that this overlapping is a rare event in practice.}
Therefore, we expect that in practice 
the degree of concurrency of \code{transferFrom} transactions 
is higher than shown above.

\begin{figure}[t]
  \centering
  \begin{tikzpicture}[>=stealth',scale=0.65] 
    \tikzstyle{place}=[circle,thick,draw=black!75,minimum size=3mm]
    \tikzstyle{transition}=[rectangle,thick,draw=black!75,minimum size=5mm]
    \tikzstyle{edge}=[->,thick,draw=black!75]
    % \tikzstyle{edgered}=[->,thick,draw=red!75]
    % \tikzstyle{edgeblu}=[->,thick,draw=blue!75]
    \node[place] (p*1) at (0,3) [label = above:{}] {};
    \node[place] (p*2) at (0,5) [label = above:{}] {};
    \node[place] (p*3) at (6,6.5) [label = above:{}] {};
    \node[place] (p*4) at (6,1.5) [label = above:{}] {};
    \node[place] (p1*) at (2,1.5) [label = above:{}] {};
    \node[place] (p2*) at (2,6.5) [label = above:{}] {};
    \node[place] (p3*) at (8,5) [label = above:{}] {};
    \node[place] (p4*) at (8,3) [label = above:{}] {};
    \node[place] (p13) at (4,4) [label = above:{}] {};
    \node[place] (p14) at (4,3) [label = above:{}] {};
    \node[place] (p23) at (4,5) [label = above:{}] {};
    \node[place,token = 1] (tokp*1) at (p*1) {};
    \node[place,token = 1] (tokp*1) at (p*2) {};
    \node[place,token = 1] (tokp*1) at (p*3) {};
    \node[place,token = 1] (tokp*1) at (p*4) {};
    \node[transition] (t1) at (2,3) {$\trT[1]$}
    edge[pre] node {} (p*1)
    edge[post] node {} (p1*)
    edge[post] node {} (p13)
    edge[post] node {} (p14);
    \node[transition] (t2) at (2,5) {$\trT[2]$}
    edge[pre] node {} (p*2)
    edge[post] node {} (p2*)
    edge[post] node {} (p23);
    \node[transition] (t3) at (6,5) {$\trT[3]$}
    edge[pre] node {} (p*3)
    edge[post] node {} (p3*)
    edge[pre] node {} (p13)
    edge[pre] node {} (p23);
    \node[transition] (t4) at (6,3) {$\trT[4]$}
    edge[pre] node {} (p*4)
    edge[post] node {} (p4*)
    edge[pre] node {} (p14);
  \end{tikzpicture}
  \caption{Occurrence net for $\bcB = \txT[1] \txT[2] \txT[3] \txT[4]$ of the ERC-721 token.}
  \label{fig:erc721:petri}
\end{figure}

  \subsection{Related work}
\label{sec:related}

A few works study how to optimize the execution of transactions
on Ethereum, using dynamic techniques adopted from software transactional 
memory.
In \cite{Dickerson17podc,Dickerson18eatcs}, miners execute
a set of transactions speculatively in parallel,
using abstract locks and inverse logs to dynamically discover conflicts
and to recover from inconsistent states.
The obtained execution is guaranteed to be equivalent to a serial
execution of the same set of transactions.
The work \cite{Anjana19pdp} proposes a conceptually similar technique, but
based on optimistic software transactional memory.
The work~\cite{Saraph19tokenomics} studies the effectiveness of 
speculatively executing smart contracts in Ethereum.
After sampling past blocks of transactions (from July 2016 to December 2017), 
the authors replay them by using a speculative execution engine, 
and measure the speedup obtained by parallel execution. 
The results show that simple speculative strategies 
are enough to obtain non-negligible speed-ups.
Another observation of~\cite{Saraph19tokenomics} 
is that many of the data conflicts 
(\ie concurrent read/write accesses to the same state location)
arise in periods of high traffic, and they are caused by a small number 
of popular contracts, like \eg ERC-20 and ERC-721 tokens.
The experiments in~\cite{Dickerson17podc} suggest that parallelizing
transaction execution may lead to a significant improvement 
of the performance of nodes:
the benchmarks on a selection of representative contracts
show an overall speedup of 1.33x for miners and 1.69x for validators,
using only three cores.

A main difference between these works and ours is that
they study empirical aspects of transaction parallelism
(\eg, the speedup obtained on a given benchmark),
while ours is more focussed on the theoretical counterpart.
Still, our theory is not intended to serve as a justification of 
the correctness of the above-mentioned approaches.
Actually, we follow a different path to transaction parallelism,
based on static analysis of transactions, rather than on speculative execution.
The reason for this divergence lies in the fact
that optimizations based on speculative execution of transactions
are not fully compatible with current blockchain platforms.
Indeed, since speculative execution is non-deterministic,
miners need to communicate the chosen schedule of transactions
to validators, which otherwise cannot correctly validate the block.
This schedule must be embedded in the mined block:
since current blockchains do not support this kind of block metadata,
implementing in practice these approaches would require 
a ``soft-fork'' of the blockchain.
Instead of performing dynamic checks, our approach relies
on a static analysis to detect potential conflicts.
Miners can use any static analysis to execute transactions in parallel;
once they have appended a block, validators just need to execute its
transactions, possibly exploiting another static analysis to parallelize
execution while preserving the semantics of the block.
In this way, our approach is compatible with any blockchain platform,
without requiring a soft-fork.
% Our technique is compatible with 
% the current implementation of the Ethereum blockchain,
% while the other existing approaches to parallelize transactions execution
% would require a soft-fork.

Our approach is based on static analyses of the variables read and written
by transactions.
Although the literature describes various static analyses of smart contracts, 
most of them are focussed on finding security vulnerabilities~\cite{Miller18isola},
and they do not produce the approximations needed for our purposes.
A few papers propose static analyses of read/written variables,
but they are not specifically targeted to Ethereum bytecode contracts.
The recent work~\cite{Pirlea21pldi} implements a static analysis  
that approximates the portion of a contract state affected by 
the execution of a transaction. 
This analysis is then exploited to evaluate the parallel execution 
of transactions over multiple \emph{shards}~\cite{Luu16sharding}.
Although the commutativity relation inferred by the 
analysis of~\cite{Pirlea21pldi} is similar of our swappability relation,
it is not directly usable on arbitrary Ethereum contracts,
since the analysis of~\cite{Pirlea21pldi} is targeted to contracts written in
the functional contract language Scilla~\cite{Sergey19pacmpl}.
Actually, the vast majority of transactions in Ethereum are sent to contracts
written in Solidity 
(see footnote~\ref{most-contracts-from-solidity} in Section~\ref{sec:validation}),
hence this assumption could undermine the applicability of the analysis 
of~\cite{Pirlea21pldi} in the wild.
The work~\cite{Dias11} describes an analysis based on separation logic,
and applies it to resolve conflicts in the setting of
snapshot isolation for transactional memory in Java.
When a conflict is detected, the read/write sets are used to determine how the
code can be modified to resolve it.
The work~\cite{CheremCG08} presents a static analysis 
of read and written locations in a C-like language with atomic sections,
and uses it to translate atomic sections into standard lock operations.
Designing precise static analyses for Solidity could perhaps take inspiration 
from these works.

In the permissioned setting, 
Hyperledger Fabric \cite{Androulaki18eurosys} 
natively supports transaction parallelism.
It follows the ``execute first and then order'' paradigm:
transactions are executed speculatively, 
and then their ordering is checked for correctness~\cite{Fabric19rw}.
In this paradigm, appending a transaction requires a few steps.
First, a client proposes a transaction to a set of ``endorsing'' peers, 
which simulate the transaction without updating the blockchain.
The output of the simulation includes 
the state updates of the transaction execution,
and the sets of read/written keys.
These sets are then signed by the endorsing peers, and returned to the client,
which submits them to the ``ordering'' peers. 
Ordering peers group transactions into blocks, 
and send them to the ``committing'' peers, which validate them. 
A block $\txT[1] \cdots \txT[n]$ is valid when
the keys read by transaction $\txT[i]$
are not written by a transaction $\txT[j]$ with $j<i$.
Finally, validated blocks are appended to the blockchain.

A preliminary version of this work was presented at 
COORDINATION 2020~\cite{BGM20coordination}.
The current version substantially extends it, 
generalising the theory to arbitrary blockchain platforms
(while~\cite{BGM20coordination} is focussed only on Ethereum).
Besides making it possible to extend our results to blockchains beyond Ethereum,
this generalization has allowed us to refine some of our results
to UTXO-based blockchains like Bitcoin.
The current work also contains the complete technical machinery,
including proofs, of our theory,
and the experimental validation of our optimization technique on Ethereum.

\section{An abstract model of blockchains}
\label{sec:transactions}

In this~\namecref{sec:transactions} we introduce a general model
of blockchain platforms, abstracting from the actual form of transactions,
from the language to write smart contracts,
and from the fact that transactions are grouped into blocks.
We then show how to instantiate this model
to the two most widespread blockchain platforms, \ie Bitcoin and Ethereum.

\begin{defi}[{Blockchain platform}]
  A \emph{blockchain platform} is a tuple 
  \mbox{$(\Tx,\Obs,\Val,\BcSt,\bcStInit,\semBc{\cdot}{})$}, where:
  \begin{itemize}
    \item $\Tx$ is a set of \emph{transactions}
      (ranged over by $\txT, \txTi, \ldots$);
    \item $\Obs$ is a set of \emph{observables}
      (ranged over by $\obsA, \obsB, \ldots$);
    \item $\Val$ is a set of \emph{values}
      (ranged over by $\valV, \valVi, \ldots$);
    \item $\BcSt \subseteq \Obs \rightharpoonup \Val$ is a set of \emph{valid blockchain states}
      (ranged over by $\bcSt, \bcSti, \ldots$);
    \item $\bcStInit \in \BcSt$ is the \emph{initial state};
    \item $\semBc{\cdot}{} \in  \BcSt \times \Tx \rightarrow \BcSt$ is the \emph{state transition function} (we write $\semBc{\txT}{\bcSt}$ for $\semBc{(\bcSt,\txT)}{}$).
  \end{itemize}
\end{defi}

The set of all finite sequences of transactions is denoted by $\Tx^*$,
and the empty sequence is denoted by $\bcEmpty$.
The set $\Tx^*$ with the concatenation operator 
and the neutral element $\bcEmpty$ is a monoid, 
referred to as the \emph{free monoid over $\Tx$}.
A \emph{blockchain} $\bcB$ is an element of $\Tx^*$.
The semantics of a blockchain $\bcB$ starting from a state $\bcSt$,
denoted as $\semBc{\bcB}{\bcSt}$, is obtained by iterating 
the semantics of its transactions:
\[
\semBc{\bcEmpty}{\bcSt} = \bcSt
\qquad\qquad
\semBc{\txT \bcB}{\bcSt} = \semBc{\bcB}{\bcSti}
\quad \text{where } \bcSti = \semTx{\txT}{\bcSt}
\]
We write $\semBc{\bcB}{}$ for $\semBc{\bcB}{\bcStInit}$,
where $\bcStInit$ is the initial state.
We say that a blockchain state $\bcSt$ is \emph{reachable}
if $\bcSt = \semBc{\bcB}{}$ for some $\bcB$.

A \emph{state update} $\mSubst : \Obs \mapstopart \Val$ is a function
which defines how values associated with observables are modified.
We denote with $\setenum{\bind{\qmvA}{\valV}}$
the state update which maps the observable $\qmvA$ to the value~$\valV$.
% We define $\keys{\mSubst}$ as the set of qualified keys $\amvA.\valK$ such that
% $\amvA \in \dom{\mSubst}$ and $\valK \in \dom{\mSubst \amvA}$.
Given a blockchain state $\bcSt$ and a state update $\mSubst$,
applying $\mSubst$ to $\bcSt$ results in a blockchain state $\bcSt \mSubst$
such that, for all observables $\qmvA$:
\[
(\bcSt \mSubst) \qmvA 
\; = \;
\begin{cases}
  \mSubst \qmvA & \text{if $\qmvA \in \dom{\mSubst}$} \\
  \bcSt \qmvA & \text{otherwise}
\end{cases}
\]
We use $\QmvA,\QmvB,\hdots$ to range over sets of observables.

  \subsection{Bitcoin}
\label{sec:blockchain:btc}

Bitcoin~\cite{bitcoin} is the first crypto-currency based on 
a decentralized ledger. 
Its mechanism to transfer currency (the bitcoin, $\BTC$) 
is based on the \emph{Unspent Transaction Output} (\emph{UTXO}) model.
This means that each transaction spends the outputs generated 
by one or more previous transactions, 
and it creates new outputs, that can be spent by later transactions
according to programmable redeem conditions.
This model contrasts with the so-called \emph{account-based} model,
implemented, \eg, by Ethereum,
where transactions update a global state, 
recording the amount of crypto-currency in each account,
and updating the state of smart contracts.
By contrast, in Bitcoin the state is given by the 
unspent transactions outputs, which
represent either $\BTC$ deposits redeemable by users
or the state of smart contracts~\cite{bitcoinsok}.
Although the language for specifying redeem conditions is quite basic,
complex smart contracts can be crafted by suitably chaining transactions
\cite{BZ18bitml}.

We now formalise the basic functionality of Bitcoin
within our general blockchain model,
simplifying or omitting the parts that are irrelevant for our subsequent technical development.

\paragraph*{Transactions}

Bitcoin transactions are records with the following fields:
\begin{itemize}
\item $\txOut{}$ is the list of \emph{outputs}.  
  Each output is a record of the form $\{ \txscript: \expe, \txval: \valV \}$, 
  where $\expe$ is a \emph{script}, and $\valV \geq 0$
  is the amount of bitcoins stored in the output.
  Intuitively, a later transaction can \emph{spend} the bitcoins 
  stored in a transaction output
  by providing a \emph{witness} which satisfies its script.
\item $\txIn{}$ is the list of \emph{inputs}. 
  Each input is a pair $(\txT,i)$, 
  meaning that the transaction wants to spend 
  the $i$-th output of the transaction $\txT$;
\item $\txWit{}$ is the list of \emph{witnesses},
  of the same length as $\txIn{}$.
  Intuitively, if the $j$-th input is $(\txT,i)$, 
  then the $j$-th witness must make the $i$-th script 
  of $\txT$ evaluate to true.%
  \footnote{Bitcoin transactions can also impose time constraints on when they can be appended to the blockchain, or when they can be spent. 
    Since our theory is applied to parallelize transactions within the \emph{same} block,
    hence satisfying exactly the same time constraints, we omit them.}
\end{itemize} 

We let $\txf$ range over transaction fields, 
and we denote with $\txT.\txf$ the content of field $\txf$ of transaction $\txT$.
We write $\txT.\txf(i)$ for the $i$-th element of the sequence $\txT.\txf$, when in range.
We interchangeably use the notation
$(\txT,i)$ and $\txT.\txOut[]{}(i)$ for \emph{transaction outputs}.
We use $\pmvA, \pmvB, \ldots$ to range over users,
and, we just write the name $\pmvA$ of a user in place
of her public/private keys, \eg we write $\versig{\pmvA}{\expe}$ for $\versig{\constPK[\pmvA]}{\expe}$,
and $\sig{\pmvA}{\txT}$ for $\sig{\constSK[\pmvA]}{\txT}$.

\label{sec:bitcoin-scripts}

Bitcoin scripts are small programs written in a non-Turing equivalent language.
Following~\cite{bitcointxm}, 
we model them as terms with the following syntax:
\begin{align*}
  \expe 
  \; \bnfdef \;
  & \valV  
  && \text{constant (integer or bitstring)}
  \\[-1pt]
  \bnfmid 
  & \expe \circ \expe 
  && \text{operators $(\circ \in \{+, -, =, <\})$}
  \\[-1pt]
  \bnfmid 
  & \ifE{\expe}{\expe}{\expe}  
  && \text{conditional}
  \\[-1pt]
  \bnfmid 
  & \seqat{\expe}{n}
  && \text{$n$-th element of sequence $\expe$ ($n \in \Nat$)}
  \\[-1pt]
  \bnfmid
  & \rtx.\txWit{}
  && \text{witnesses of the redeeming tx}
  \\[-1pt]
  \bnfmid
  & \sizeE{\expe} 
  && \text{size (number of bytes)}
  \\[-1pt]
  \bnfmid
  & \hashE{\expe} 
  && \text{hash}
  \\[-1pt]
  \bnfmid
  & \versig{\expe}{\expei}
  && \text{signature verification}
  % \\
  % \bnfmid
  % & \afterAbs{\expe}{\expe} 
  % && \text{absolute timelock}
  % \\
  % \bnfmid 
  % & \afterRel{\expe}{\expe}
  % && \text{relative timelock}
\end{align*}

Besides constants $\valV$, basic arithmetic/logical operators, and conditionals,  
scripts can access the elements of a sequence ($\seqat{\expe}{n}$),
and the sequence of witnesses of the redeeming transaction ($\rtx.\txWit{}$);
further, they can compute the size $\sizeE{\expe}$ of a bitstring
and its hash $\hashE{\expe}$.
The script $\versig{\expe}{\expei}$
evaluates to $1$ if the signature resulting from the evaluation of $\expei$
is verified against the public key resulting from the evaluation of~$\expe$,
and $0$ otherwise.
For all signatures, the signed message is the redeeming transaction (except its witnesses).%

The evaluation of scripts is defined as a function $\sem[\txT,i]{\cdot}$,
which takes two additional parameters (used for signature verification):
the redeeming transaction $\txT$, and
the index $i$ of the redeeming input/witness.
The result of the semantics can be an integer or a bitstring.
The rules of the semantics are standard:
we refer to~\cite{bitcointxm} for a formalization.

\begin{exa}
  \label{ex:blockchain-btc:1}
  Consider a transaction of the form:
  \begin{nscenter}
    \small
    \begin{tabular}[t]{|l|}
      \hline
      \multicolumn{1}{|c|}{$\txT[0]$} \\
      \hline
      \txIn[1]{$\emptyseq$} \\
      \txWit[1]{$\emptyseq$} \\
      \txOut[1]{$\{ \txscript: \versig{\pmvA}{\rtxWit}, \txval: 80 \BTC \}$} \\
      \txOut[2]{$\{ \txscript: \versig{\pmvB}{\rtxWit}, \txval: 20 \BTC \}$} \\
      \hline
    \end{tabular}
  \end{nscenter}

  The $\txIn{}$ and $\txWit{}$ fields are empty, 
  making $\txT[0]$ a \emph{coinbase} transaction
  (\ie, the first transaction in the blockchain).
  This transaction has two outputs:
  $(\txT[0],1)$ allows $\pmvA$ to redeem $80 \BTC$,
  while $(\txT[0],2)$ allows $\pmvB$ to redeem $20 \BTC$.
  Assume that $(\txT[0],1)$ is unspent, 
  and that $\pmvA$ wants to transfer $10 \BTC$ to $\pmvB$, 
  and keep the remaining $70 \BTC$.
  To do this, $\pmvA$ can append to the blockchain a new transaction, \eg:
  \begin{nscenter}
    \small
    \begin{tabular}[t]{|l|}
      \hline
      \multicolumn{1}{|c|}{$\txT[1]$} \\
      \hline
      \txIn[1]{$(\txT[0],1)$} \\
      \txWit[1]{$\sig{\pmvA}{\txT[1]}$} \\
      \txOut[1]{$\{ \txscript: \versig{\pmvA}{\rtxWit}, \txval: 70 \BTC \}$} \\
      \txOut[2]{$\{ \txscript: \versig{\pmvB}{\rtxWit}, \txval: 10 \BTC \}$} \\
      \hline
    \end{tabular}
  \end{nscenter}

  The $\txIn{}$ field points to the first output of $\txT[0]$,
  and the $\txWit{}$ field contains $\pmvA$'s signature on $\txT[1]$ 
  (but for the $\txWit{}$ field itself).
  This witness makes the script $(\txT[0],1).\txscript$ evaluate to true,
  hence the redemption succeeds, and the output $(\txT[0],1)$ is \emph{spent}.

  Assume now that the outputs $(\txT[0],2)$ and $(\txT[1],2)$ are unspent.
  Participant $\pmvB$ can spend both of them by appending a new transaction 
  $\txT[2]$ to the blockchain:
  \begin{nscenter}
    \small
    \begin{tabular}[t]{|l|}
      \hline
      \multicolumn{1}{|c|}{$\txT[2]$} \\
      \hline
      \txIn[1]{$(\txT[0],2)$} \hspace{50pt}
      \txIn[2]{$(\txT[1],2)$} \\
      \txWit[1]{$\sig{\pmvB}{\txT[2]}$} \hspace{35pt}
      \txWit[2]{$\sig{\pmvB}{\txT[2]}$} \\
      \txOut[1]{$\{ \txscript: \hashE{\rtxWit} = 51, \txval: 30\BTC \}$} \\
      \hline
    \end{tabular}
  \end{nscenter}

  In this case, the recipient of the $30 \BTC$ is not explicitly 
  specified by the script $(\txT[2],1).\txscript$: 
  actually, any transaction which provides as witness a preimage of 
  51 can spend that output.
\end{exa}

\paragraph*{Blockchain states}

We define observables as transaction outputs $(\txT,i)$, 
and the set of values as $\Val = \setenum{0,1}$.
In this way, blockchain states are partial functions 
$\bcSt \in \Obs \rightharpoonup \setenum{0,1}$,
modelling the set of \emph{unspent transaction outputs} ($\utxo{}$).
We denote with $U_{\bcSt}$ the set whose characteristic function is $\bcSt$,
\ie $U_{\bcSt} = \sigma^{-1}\setenum{1}$.
Hereafter, when not ambiguous we treat $\txT.\txIn{}$ as a set,
rather than as a sequence, and we write $\txT.\txOut{}$ for the set of pairs
$\setenum{(\txT,1), \ldots, (\txT,n)}$, where $n = |\txT.\txOut{}|$.
The initial blockchain state is the UTXO $\setenum{\txT[0].\txOut{}}$,
where $\txT[0]$ is a coinbase transaction (\ie, $\txT[0].\txIn{} = \emptyseq$).

\paragraph*{State transitions}

We start by defining when a transaction is valid in a blockchain state.

\begin{defi}[{Valid Bitcoin transactions}]
  \label{def:blockchain-btc:valid}
  We say that a transaction $\txT$ is \emph{valid} in a blockchain state $\bcSt$ 
  (in symbols, $\bcSt \valid \txT$) when the following conditions hold:
  \begin{enumerate}
  \item \label{consistent-update:unspent}
    $(\txTi,j) \in U_{\bcSt}$,
    for each $(\txTi,j)$ in $\txT.\txIn{}$
    \smallskip
  \item \label{consistent-update:output} 
    $\sem[{\txT},i]{(\txTi,j).\txscript} = \valV \neq 0$,
    for each $(\txTi,j)$ in $\txT.\txIn{}$
    \smallskip
  \item \label{consistent-update:value}
    \(
    \sum_{\obsA \in \txT.\txIn{}} \obsA.\txval
    \; \geq \;
    \sum_{\obsB \in \txT.\txOut{}} \obsB.\txval
    \)
  \end{enumerate}

  \smallskip\noindent
  We say that $\txT$ is \emph{consistent} when there exists some
  $\bcSt$ such that $\bcSt \valid \txT$.
\end{defi}

Condition~\eqref{consistent-update:unspent} requires that all the inputs
of $\txT$ are unspent in $\bcSt$;
condition~\eqref{consistent-update:output} asks that all the scripts
referred to by $\txT.\txIn{}$ evaluate to true, using the witnesses
in $\txT.\txWit{}$;
condition~\eqref{consistent-update:value} asks that
the value of the inputs of $\txT$ is greater or equal
to the value of its outputs.

We now define the state transition function of Bitcoin as 
$\sem[\bcSt]{\txT} = \bcSti$, where:
\[
U_{\bcSti} 
\; = \;
\begin{cases}
  (U_{\bcSt} \setminus \txT.\txIn{}) \cup \txT.\txOut{} & \text{if $\bcSt \valid \txT$} \\
  U_{\bcSt} & \text{otherwise}
\end{cases}
\]
We extend validity to blockchains by passing through their semantics, 
\ie we write $\bcB \valid \txT$ when $\sem{\bcB} \valid \txT$.
Further, we write $\bcSt \valid \txT[1] \cdots \txT[n]$ iff
$\sem[\bcSt]{\txT[1] \cdots \txT[i-1]} \valid \txT[i]$,
for all $i \leq n$.%
\footnote{The Bitcoin consensus protocol ensures that 
  each transaction $\txT[i]$ in the blockchain is valid
  with respect to the sequence of past transactions $\txT[0] \cdots \txT[i-1]$.
  Since our model requires the state transition function to be total, 
  we make the operation of appending invalid transactions idempotent.}

\begin{exa}
  \label{ex:blockchain-btc:2}
  Recall the transactions $\txT[0]$, $\txT[1]$, $\txT[2]$ 
  from Example~\ref{ex:blockchain-btc:1},
  and let $U_{\bcSt[0]} = \setenum{\txT[0].\txOut{}}$.
  We have that $\bcSt[0] \valid \txT[1] \txT[2]$,
  and $\semBc{\txT[1] \txT[2]}{\bcSt[0]}$
  is the UTXO $\setenum{(\txT[1],1), (\txT[2],1)}$.
\end{exa}

  \subsection{Ethereum}

Ethereum~\cite{ethereum} is one of the most used platforms for smart contracts: 
it actually implements a decentralized virtual machine 
that runs contracts written in a Turing-complete bytecode language,
called EVM~\cite{EVMdoc}.
Abstractly, an Ethereum  contract is similar to an object in an object-oriented
language: 
it has an internal state, and a set of functions to manipulate it.
Each contract controls an amount of crypto-currency (the \ether), 
that it can exchange with other users and contracts.
Transactions trigger contracts updates, 
which may possibly involve a transfer of 
crypto-currency from the caller to the callee.

Users and contracts are identified by their \emph{addresses}.
We use $\cmvA, \cmvB, \ldots$ to range over contract addresses,
and $\funF, \funG, \ldots$ for contract functions.
We denote with $\Addr$ the set of all \emph{addresses} $\amvA,\amvB,\ldots$,
including both user and contract addresses.

\paragraph{Transactions}

Ethereum transactions are terms of the form:
\[
\ethtx{\valN}{\pmvA}{\amvA}{\funF}{\vec{\valV}}
\]
where $\pmvA$ is the address of the caller,
$\amvA$ is the address of the called contract or user,
$\funF$ is the called function,
$n$ is the amount of $\ether$ transferred from $\pmvA$ to $\amvA$, and
$\vec{\valV}$ is the sequence of actual parameters.
A contract has a finite set of functions, \ie terms of the form
$\contrFun{\funF}{\vec{\varX}}{\cmdC}$,
where $\funF$ is a function name,
$\vec{\varX}$ is the sequence of formal parameters (omitted when empty),
and $\cmdC$ is the function body.
The functions in a contract have distinct names.
We denote with $\addrToContr{\amvA}$ the contract at address $\amvA$.
We abstract from the actual syntax of $\cmdC$,
and we just assume that the semantics of function bodies is defined
(see \eg \cite{BGM19cbt,Crafa19wtsc,Jiao20sp} 
for concrete instances of syntax and semantics of function bodies).
For uniformity, we assume that user addresses are associated with a 
contract having exactly one function, which just skips.
In this way, the statement $\cmdSend{n}{\pmvA}$,
which transfers $n$ currency units to user $\pmvA$,
can be rendered as a call to this function.

\paragraph{Blockchain states}

Each Ethereum contract has a key-value store, 
rendered as a partial function $\Val \mapstopart \Val$
from values to values.
The elements in the domain of this function are also called \emph{keys}.
The set of values $\Val$ includes basic types, \eg integers and strings.
An observable is a term of the form $\amvA.\valK$,
\ie a key in the key-value store at a given address.
The possible blockchain states are the partial functions 
$\bcSt \in \Obs \rightharpoonup \Val$ such that:
\begin{itemize}
\item for all addresses $\amvA$, $\bcSt \, \amvA.\varBalance$ is defined;
\item for all user addresses $\pmvA$,
  $\bcSt \, \pmvA.\valK$ is defined iff $\valK = \varBalance$.
\end{itemize}
The second constraint allows for a uniform treatment of users and contracts.
The initial state $\bcStInit$ maps
each address $\amvA$ to a balance $\balanceInit{\amvA} \geq 0$,
while all the other keys are unbound.

\paragraph{State transitions}

Let $\Const$ be a set  of \emph{constant names} $\varX, \varY, \ldots$.
We denote with $\semCmd{\cmdC}{\bcSt,\bcEnv}{\amvA}$ the semantics of 
the statement $\cmdC$.
This semantics is either a blockchain state $\bcSti$, 
or it is undefined (denoted by $\bot$).
The semantics is parameterised over a state $\bcSt$, an address $\amvA$
(the contract wherein $\cmdC$ is evaluated),
and an \emph{environment} $\bcEnv : \Const \mapstopart \Val$,
used to evaluate the formal parameters
and the special names $\varSender$ and $\varValue$.
These names represent, respectively, the caller of the function,
and the amount of \ether transferred along with the call.
We postulate that \varSender and \varValue 
are not used as formal parameters.

We define the auxiliary operators $+$ and $-$ on blockchain states as follows:
\[
\bcSt \circ (\amvA:\valN)
\; = \;
\bcSt\setenum{\bind{\amvA.\varBalance}{(\bcSt \amvA.\varBalance) \,\circ\, \valN}}
\tag*{($\circ \in \setenum{+,-}$)}
\]
\ie, $\bcSt + \amvA:\valN$ updates $\bcSt$ 
by increasing the $\varBalance$ of $\amvA$ of $\valN$ currency units.

\begin{defi}[{Valid Ethereum transactions}]
  \label{def:blockchain-eth:valid}
  A transaction $\txT = \ethtx{\valN}{\pmvA}{\amvA}{\funF}{\vec{\valV}}$ 
  is \emph{valid} in a blockchain state $\bcSt$ 
  (in symbols, $\bcSt \valid \txT$) when the following conditions hold:
  \begin{enumerate}
  \item \label{cond:blockchain-eth:valid:balance}
    $\bcSt \, \pmvA.\varBalance \geq \valN$
  \item \label{cond:blockchain-eth:valid:sem}
    if $\contrFun{\funF}{\vec{\varX}}{\cmdC} \in \addrToContr{\amvA}$, 
    then $\semCmd{\cmdC}{\bcSt-\pmvA:\valN+\amvA:\valN,\, \setenum{\bind{\varSender}{\pmvA},\bind{\varValue}{\valN},\bind{\vec{\varX}}{\vec{\valV}}}}{\amvA} \neq \bot$
  \end{enumerate}
  \smallskip\noindent
  We say that $\txT$ is \emph{consistent} when there exists
  $\bcSt$ such that $\bcSt \valid \txT$.
\end{defi}

Condition~\eqref{cond:blockchain-eth:valid:balance} requires that
$\pmvA$'s balance is sufficient to transfer $\valN$ \ether to $\amvA$;
condition~\eqref{cond:blockchain-eth:valid:sem} asks that
the function call terminates in a non-error state.

We define the semantics of a transaction 
in a blockchain state $\bcSt$ as follows:
\[
\semTx{\ethtx{\valN}{\pmvA}{\amvA}{\funF}{\vec{\valV}}}{\bcSt} 
= \begin{cases}
  \semCmd{\cmdC}{\bcSt-\pmvA:\valN+\amvA:\valN,\, \setenum{\bind{\varSender}{\pmvA},\bind{\varValue}{\valN},\bind{\vec{\varX}}{\vec{\valV}}}}{\amvA} 
  & \text{if $\bcSt \valid \txT$ and $\contrFun{\funF}{\vec{\varX}}{\cmdC} \in \addrToContr{\amvA}$} \\
  \bcSt & \text{otherwise}
\end{cases}
\]
If the transaction is valid, the updated state is the one resulting
from the execution of the function call.
Note that $\valN$ units of currency are transferred to $\amvA$
\emph{before} starting to execute $\funF$,
and that the names $\varSender$ and $\varValue$ are bound, respectively,
to $\pmvA$ and $\valN$.
If the transaction is not valid, 
\ie $\pmvA$'s balance is not enough or the execution of $\funF$ fails,
then the transaction does not alter the blockchain state.
Invalid transactions can actually occur in the Ethereum blockchain,
but they have no effect on the state of contracts: 
so, our semantics makes them identities \wrt the append operation.

\begin{exa}
  \label{ex:th-sem}
  Consider a contract at address $\cmvA$ which includes the following functions:
  \[
  \contrFun{\funF[0]}{}{\cmdAss{\valX}{1}}
  \qquad
  \contrFun{\funF[1]}{}{\cmdIfT{\valX=0}{\cmdSend{1}{\pmvB}}}
  \qquad
  % \contrFun{\funF[2]}{}{\cmdSend{1}{\pmvB}}
  \contrFun{\funF[2]}{y}{\cmdSend{\varValue}{y}}
  \]
  The first function only sets the value of the key $\valX$ to $1$; 
  the second one transfers a unit of \ether to address $\pmvB$ when 
  $\valX$ is $0$; 
  the last one always sends a unit of \ether to $y$. 
  Consider a blockchain $\bcB = \txT[0]\txT[1]\txT[2]$ where:
  \[
  \txT[0] = \ethtx{0}{\pmvA}{\cmvA}{\funF[0]}{}
  \hspace{30pt}
  \txT[1] = \ethtx{1}{\pmvA}{\cmvA}{\funF[1]}{}
  \hspace{30pt}
  \txT[2] = \ethtx{1}{\pmvA}{\cmvA}{\funF[2]}{\pmvB}
  \]
  Let $\bcStInit$ be a state such that $\bcStInit \, \pmvA.\varBalance \geq 2$.
  The semantics of $\bcB$ in $\bcStInit$ is: 
  \[
  \semBc{\bcB}{\bcStInit} = \bcStInit\setenum{\bind{\cmvA.\varX}{1}} - \pmvA:2 + \pmvB:1 + \cmvA:1
  \]
  where the semantics of the single transactions is:
  \begin{align*}
    \semTx{\txT[0]}{\bcStInit} 
    & = \semCmd{\cmdAss{\valX}{1}}{\bcStInit,\,
      \setenum{\bind{\varSender}{\pmvA},\bind{\varValue}{0}}}{\cmvA} = 
      \bcStInit\setenum{\bind{\cmvA.\varX}{1}} = \bcSt[1]
    \\
    \semTx{\txT[1]}{\bcSt[1]} 
    &  = \semCmd{\cmdIfT{\valX=0}{\cmdSend{1}{\pmvB}}}{\bcSt[1] - \pmvA:1 + \cmvA:1,\,
      \setenum{\bind{\varSender}{\pmvA},\bind{\varValue}{1}}}{\cmvA} 
      = \bcSt[1] - \pmvA:1 + \cmvA:1 = \bcSt[2]
    \\
    \semTx{\txT[2]}{\bcSt[2]} 
    &  = \semCmd{\cmdSend{1}{y}}{\bcSt[2] - \pmvA:1 + \cmvA:1,\,
      \setenum{\bind{y}{\pmvB},\bind{\varSender}{\pmvA},\bind{\varValue}{1}}}{\cmvA} = \bcSt[2] - \pmvA:1 + \pmvB:1
  \end{align*}
\end{exa}

\section{Swapping transactions}
\label{sec:txswap}

We define two blockchain states to be \emph{observationally equivalent}
when they agree on the values associated to all observables.
The actual definition of equivalence is a bit more general, 
allowing us to restrict the set $\QmvA$ of observables over which 
we require the agreement.

\begin{defi}[{Observational equivalence}]
  \label{def:sim}
  For all $\QmvA \subseteq \Obs$,
  we define 
  \mbox{$\bcSt \sim_{\QmvA} \bcSti$}
  iff 
  $\forall \qmvA \in \QmvA: \bcSt \qmvA = \bcSti \qmvA$.
  We say that $\bcSt$ and $\bcSti$ are \emph{observationally equivalent},
  in symbols $\bcSt \sim \bcSti$,
  when $\bcSt \sim_{\QmvA} \bcSti$ holds for all $\QmvA$.
\end{defi}

The following lemma ensures that our notion of observational equivalence 
is an equivalence relation, 
and that it is preserved when we restrict the set of observables: 

\begin{restatable}{lem}{lemsimequiv}
  \label{lem:sim:equiv}
  For all $\QmvA,\QmvB \subseteq \Obs$:
  \begin{inlinelist}
  \item \label{lem:sim:equiv:equiv}
    $\sim_{\QmvA}$ is an equivalence relation;
  \item \label{lem:sim:equiv:subseteq}
    if $\bcSt \sim_{\QmvA} \bcSti$ and $\QmvB \subseteq \QmvA$,
  then $\bcSt \sim_{\QmvB} \bcSti$;
  \item \label{lem:sim:equiv:sim}
    $\sim \, = \, \sim_{\Obs}$.
  \end{inlinelist}
\end{restatable}

We extend the relations above to blockchains, 
by passing through their semantics.
For all $\QmvA$, we define 
$\bcB \sim_{\QmvA} \bcBi$
iff
$\semBc{\bcB}{\bcSt} \sim_{\QmvA} \semBc{\bcBi}{\bcSt}$
holds for all reachable $\bcSt$
(note that all the definitions and results in this paper apply to reachable states, 
since the unreachable ones do not represent actual blockchain executions).
We write $\bcB \sim \bcBi$ when $\bcB \sim_{\QmvA} \bcBi$ holds for all $\QmvA$.

A relation $\relR \subseteq \Tx^* \times \Tx^*$ is a 
is a \emph{congruence} (with respect to concatenation) if:
\[
  \bcB \;\relR\; \bcBi 
  \implies 
  \forall \bcB[0],\bcB[1] \, : \, 
  \bcB[0] \bcB \bcB[1] \;\relR\; \bcB[0] \bcBi \bcB[1]
\] 
The following lemma states that $\sim$ is a congruence:
therefore, if $\bcB$ and $\bcBi$ are observationally equivalent,
then we can replace $\bcB$ with $\bcBi$ 
in a larger blockchain, preserving its semantics.

\begin{restatable}{lem}{lemsimappend}
  \label{lem:sim:append}
  $\sim$ is a congruence relation.
\end{restatable}

We say that two transactions are \emph{swappable} when
exchanging their order preserves observational equivalence.

\begin{defi}[{Swappability}]
  \label{def:swap}
  We say that two transactions $\txT \neq \txTi$ are \emph{swappable},
  in symbols $\txT \swap \txTi$, when $\txT \txTi \sim \txTi \txT$.
\end{defi}

The theory of trace languages originated from Mazurkiewicz's works~\cite{Mazurkiewicz88rex}
allows us to study observational equivalence under various swappability relations.
In general, given an alphabet $\Sigma$ and 
a symmetric and irreflexive relation $I \subseteq \Sigma \times \Sigma$
(which models independence between two elements in $\Sigma$),
the Mazurkiewicz's trace equivalence $\equivStSeq[I]$ is
a congruence between words on $\Sigma$.
Intuitively, all the words in the same equivalence class of $\equivStSeq[I]$ 
represent equivalent concurrent executions.
In our setting, $\Sigma$ is the set of transactions,
and $I$ will be instantiated with various swappability relations.
The fact that $\equivStSeq[I]$ is a congruence will allow us to
replace a sequence of transactions with an equivalent one 
within a blockchain.

\begin{defi}[{Mazurkiewicz equivalence}]
  \label{def:mazurkiewicz}
  Let $I$ be a symmetric and irreflexive relation on $\Tx$.
  The \emph{Mazurkiewicz equivalence} $\equivStSeq[I]$
  is the least congruence in the free monoid $\Tx^*$ such that:
  $\forall \txT, \txTi \in \Tx$:
  \(\;
  \txT \, I \, \txTi \implies \txT \txTi \equivStSeq[I] \txTi \txT
  \).
\end{defi}

To exemplify Definition~\ref{def:mazurkiewicz},
let $I = \setenum{(\txT[1],\txT[2]),(\txT[2],\txT[1])}$.
The equivalence class of the word $\txT[0] \txT[1] \txT[1] \txT[2] \txT[0]$ 
under the relation $\equivStSeq[I]$ is
\(
\setenum{
  \txT[0] \txT[1] \txT[1] \txT[2] \txT[0], \,
  \txT[0] \txT[1] \txT[2] \txT[1] \txT[0], \,
  \txT[0] \txT[2] \txT[1] \txT[1] \txT[0]
}
\).
Note that, starting from the word $\txT[0] \txT[1] \txT[1] \txT[2] \txT[0]$,
the other words in its equivalence class can be obtained by
swapping adjacent occurrences of $\txT[1]$ and $\txT[2]$.
This reflects the fact that $\txT[1]$ and $\txT[2]$ 
are assumed to be concurrent, as they are related by $I$.

\medskip
Intuitively, all the words in the same equivalence class 
(with respect to $\equivStSeq[I]$) represent equivalent executions.
This is made formal by Theorem~\ref{th:swap:mazurkiewicz} below, which
ensures that the Mazurkiewicz equivalence
constructed on the swappability relation $\swap$ is an observational equivalence.
Hence, we can transform a blockchain into an observationally equivalent one
by a finite number of swaps of adjacent swappable transactions.

\begin{restatable}{thm}{thswapmazurkiewicz}
  \label{th:swap:mazurkiewicz}
  $\equivStSeq[\swap] \;\; \subseteq \;\; \sim$.
\end{restatable}

Note that the converse of Theorem~\ref{th:swap:mazurkiewicz} does not hold:
indeed, $\bcB \equivStSeq[\swap] \bcBi$ requires that $\bcB$ and $\bcBi$
have the same length, while $\bcB \sim \bcBi$ may also hold for blockchains
of different lengths (\eg, $\bcBi = \bcB \txT$ where $\txT$ is a transaction
which does not alter the state).

\paragraph*{Safe approximations of read/written observables}

The relation $\swap$ is undecidable whenever the contract language 
is Turing-equivalent, \eg, in the case of Ethereum.
When $\swap$ is undecidable, to detect swappable transactions 
we can follow a static approach.
First, we over-approximate the set of observables read and written 
by transactions (Definition~\ref{def:safeapprox}).
We then check a simple condition on these approximations 
(Definition~\ref{def:pswapWR}) to detect if two transactions can be swapped.
Of course, the quality of the approximation is crucial to the
effectiveness of the approach.
In general, the coarser the approximation,
the stricter the induced swappability relation:
therefore, an overly coarse approximation would undermine the
parallelization of transactions.

In Definition~\ref{def:safeapprox} we state that $\QmvA$ safely
approximates the observables \emph{written} by $\txT$
(in symbols, $\wapprox{\QmvA}{\txT}$) 
when executing $\txT$ does not alter the state of the observables not in~$\QmvA$.
Defining the set of \emph{read} observables is a bit trickier: 
we require that executing $\txT$ in two states that agree 
on the values of the observables in the read set
results in two states that differ at most on the 
observables where they did not agree before the execution of $\txT$.

\begin{defi}[{Safe approximation of read/written observables}]
  \label{def:safeapprox}
  Given a set of observables $\QmvA$ and a transaction $\txT$, we define:
  \begin{align*}
    & \wapprox{\QmvA}{\txT}
    && \text{iff}
    && \forall \QmvB: \QmvB \cap \QmvA = \emptyset \implies \txT \sim_{\QmvB} \bcEmpty
    \\
    & \rapprox{\QmvA}{\txT}
    && \text{iff}
    && \forall \bcB,\bcBi,\QmvB: \bcB \sim_{\QmvA} \bcBi \land
       \bcB \sim_{\QmvB} \bcBi
       \implies \bcB\txT \sim_{\QmvB} \bcBi\txT
  \end{align*}
\end{defi}

\begin{exa}
  Recall from Example~\ref{ex:th-sem} the Ethereum transaction:
  \[
  \txT[2] = \ethtx{1}{\pmvA}{\cmvA}{\funF[2]}{\pmvB}
  \qquad\qquad
  \text{where } \;\contrFun{\funF[2]}{y}{\cmdSend{1}{y}}  
  \]
  The execution of $\txT[2]$ affects the $\varBalance$ of $\pmvA$, $\pmvB$ and $\cmvA$; however, $\cmvA.\varBalance$ is first incremented
  and then decremented, so its value is unchanged.
  Then, $\setenum{\expGet{\pmvA}{\varBalance},\expGet{\pmvB}{\varBalance}}$ is a safe approximation of the observables \emph{written} by $\txT[2]$, \ie
  $\wapprox{\setenum{\expGet{\pmvA}{\varBalance},\expGet{\pmvB}{\varBalance}}}{\txT[2]}$.
 
  A safe approximation of the observables \emph{read} by $\txT[2]$
  is $\QmvA = \setenum{\pmvA.\varBalance}$.
  To prove this, consider two blockchains $\bcB$ and $\bcBi$, 
  and a set of observables $\QmvB$ such that 
  $\bcB \sim_{\QmvA} \bcBi$ and $\bcB \sim_{\QmvB} \bcBi$.
  We have two cases:
  \begin{itemize}
    
  \item If $\semBc{\bcB}{} \pmvA.\varBalance < 1$, 
    then $\txT[2]$ is not valid in $\bcB$, and so
    $\semBc{\bcB\txT[2]}{} = \semBc{\bcB}{}$.
    Since $\bcB \sim_{\QmvA} \bcBi$, 
    then $\semBc{\bcBi}{} \pmvA.\varBalance < 1$,
    so $\txT[2]$ is not valid also in $\bcBi$,
    from which we have $\semBc{\bcBi\txT[2]}{} = \semBc{\bcBi}{}$.
    % Therefore, $\bcB\txT[2] \sim_{\QmvB} \bcBi \txT[2]$.
    
  \item If $\semBc{\bcB}{} \pmvA.\varBalance = n \geq 1$, 
    then $\txT[2]$ is valid in $\bcB$, and $\txT[2]$ 
    affects exactly $\pmvA.\varBalance$ and $\pmvB.\varBalance$,
    as it transfers $1 \, \ether$ from $\pmvA$ to $\pmvB$.
    Since $\bcB \sim_{\QmvB} \bcBi$, 
    the states of $\bcB \txT[2]$ and $\bcBi \txT[2]$ may only 
    differ on $\pmvA.\varBalance$ or $\pmvB.\varBalance$.
    However: 
    \begin{align*}
      & \semBc{\bcBi\txT[2]}{} \pmvA.\varBalance = n - 1 = \semBc{\bcB\txT[2]}{} \pmvA.\varBalance
      \\
      & \semBc{\bcBi\txT[2]}{} \pmvB.\varBalance = \semBc{\bcBi}{} \pmvB.\varBalance + 1 =
        \semBc{\bcB}{} \pmvB.\varBalance + 1 = \semBc{\bcB\txT[2]}{} \pmvB.\varBalance
    \end{align*}
  \end{itemize}
  
  \noindent
  Therefore, in both cases $\bcB\txT[2] \sim_{\QmvB} \bcBi \txT[2]$,
  and so we have proved that $\rapprox{\setenum{\pmvA.\varBalance}}{\txT[2]}$.
\end{exa}

Widening a safe approximation (either of read or written observables)
preserves its safety;
further, the intersection
of two write approximations is still safe.
From this, it follows that there exists a \emph{least} safe approximation
of the observables written by a transaction.

\begin{restatable}{lem}{lemsafeapprox}
  \label{lem:safeapprox}
  Let $\bullet \in \setenum{r,w}$. Then:
  \begin{enumerate}[(a)]
  \item \label{lem:safeapprox:subseteq}
    if $\safeapprox[\bullet]{\QmvA}{\txT}$ and $\QmvA \subseteq \QmvAi$,
    then $\safeapprox[\bullet]{\QmvAi}{\txT}$;
  \item \label{lem:safeapprox:cap}
    if $\safeapprox[w]{\QmvA}{\txT}$
    and $\safeapprox[w]{\QmvB}{\txT}$,
    then $\safeapprox[w]{\QmvA \cap \QmvB}{\txT}$.
  \end{enumerate}
\end{restatable}

The following example shows that, in general, part~\ref{lem:safeapprox:cap} 
of Lemma~\ref{lem:safeapprox}
does not hold for read approximations.

\begin{exa}
  \label{ex:safeapprox:read-not-cap}
  Let $\cmvA$ be an Ethereum contract with functions:
  \begin{align*}
    & \contrFunSig{\funF}{\varX}
      \; \{ \cmdAss{\valK}{x}; \cmdAss{\valKi}{x} \}
    && \contrFunSig{\funG}{}
       \; \{\cmdIfTE{\expLookup{\valK} \neq \pmvA}{\cmdSend{\expLookup{\varBalance}}{\pmvB}}{\cmdSkip} \}
  \end{align*}
  and let $\txT = \ethtx{0}{\pmvA}{\cmvA}{\funG}{}$.
  Note that, in any reachable state $\bcSt$,
  it must be $\bcSt \, \cmvA.\valK = \bcSt \, \cmvA.\valKi$.
  Let $\QmvB$ be such that $\bcB \sim_{\QmvB} \bcBi$,
  and let $\bcSt = \semBc{\bcB}{}$, $\bcSti = \semBc{\bcBi}{}$,
  let $\valN = \bcSt \, \cmvA.\varBalance$, and
  let $\valNi = \bcSti \, \cmvA.\varBalance$.
  Appending $\txT$ to $\bcB$ and $\bcBi$ will result in:
  \begin{align*}
    & \semTx{\txT}{\bcSt} = \begin{cases}
      \bcSt - \cmvA:\valN + \pmvB:\valN & \text{if $\bcSt \, \cmvA.\valK \neq \pmvA$} \\
      \bcSt & \text{otherwise}
    \end{cases}
    \quad
    && \semTx{\txT}{\bcSti} = \begin{cases}
      \bcSti - \cmvA:\valNi + \pmvB:\valNi & \text{if $\bcSti \, \cmvA.\valK \neq \pmvA$} \\
      \bcSti & \text{otherwise}
    \end{cases}
  \end{align*}
  If $\bcB \sim_{\setenum{\cmvA.\valK}} \bcBi$, then the conditions
  $\bcSt \, \cmvA.\valK \neq \pmvA$ and $\bcSti \, \cmvA.\valK \neq \pmvA$
  are equivalent.
  Therefore,
  $\sem{\bcB \txT} = \semTx{\txT}{\bcSt} \sim_{\QmvB} \semTx{\txT}{\bcSti} = \sem{\bcBi \txT}$,
  and so we have proved that $\rapprox{\setenum{\cmvA.\valK}}{\txT}$.
  Similarly, we obtain that $\rapprox{\setenum{\cmvA.\valKi}}{\txT}$,
  since $\valK$ and $\valKi$ are always bound to the same value.
  Note however that $\setenum{\cmvA.\valK} \cap \setenum{\cmvA.\valKi} = \emptyset$
  is \emph{not} a safe approximation of the observables read by $\txT$.
  For instance, if $\bcSt \, \cmvA.\valK = \pmvA \neq \bcSti \, \cmvA.\valKi$
  and $\bcSt \, \cmvA.\varBalance = \bcSti \, \cmvA.\varBalance$,
  then appending $\txT$ to $\bcB$ or to $\bcBi$ results in states
  which differ in the balance of $\cmvA$.
\end{exa}

\paragraph*{Strong swappability}

We use safe approximations of read/written observables 
to detect when two transactions are swappable,
recasting in our setting Bernstein's conditions~\cite{bernstein}
for the parallel execution of processes.
More specifically,
we require that the set of observables written by $\txT$
is disjoint from those written or read by $\txTi$, and vice versa.
When this happens, we say that the
two transactions are \emph{strongly swappable}.

\begin{defi}[{Strong swappability}]
  \label{def:pswap}
  We say that two transactions $\txT \neq \txTi$ are \emph{strongly swappable},
  in symbols $\txT \pswap{}{} \txTi$, when
  there exist $W,W',R,R' \subseteq \Obs$ such that
  $\wapprox{W}{\txT}$,
  $\wapprox{W'}{\txTi}$,
  $\rapprox{R}{\txT}$,
  $\rapprox{R'}{\txTi}$,
  and:
  \[
    \big( R \cup W \big) \cap W'
    \; = \;
    \emptyset
    \; = \;
    \big( R' \cup W' \big) \cap W
  \]
\end{defi}

The following theorem ensures the soundness of our approximation:
if two transactions are strongly swappable, then they are also swappable.
Since its proof depends on notions that have yet to be defined, 
we postpone it at the end of the section.
The converse implication does not hold neither in Bitcoin nor in Ethereum,
as shown by Examples~\ref{cex:pswap-implies-swap:btc} and~\ref{cex:pswap-implies-swap:eth}.

\begin{restatable}{thm}{thpswapimpliesswap}
  \label{th:pswap-implies-swap}
  $\txT \pswap{}{} \txTi \implies \txT \swap \txTi$.
\end{restatable}

Theorem~\ref{th:pswap:mazurkiewicz} states that
the Mazurkiewicz equivalence $\equivStSeq[\pswap{}{}]$
is stricter than $\equivStSeq[\swap]$.
Together with Theorem~\ref{th:swap:mazurkiewicz},
if $\bcB$ is transformed into $\bcBi$
by exchanging adjacent strongly swappable transactions,
then $\bcB$ and $\bcBi$ are observationally equivalent.

\begin{restatable}{thm}{thpswapmazurkiewicz}
  \label{th:pswap:mazurkiewicz}
  $\mathord{\equivStSeq[\pswap{}{}]} \subseteq \mathord{\equivStSeq[\swap]}$.
\end{restatable}

\paragraph*{Parameterised strong swappability}

Note that if the contract language is Turing-equivalent, 
then finding approximations which satisfy 
the disjointness condition in Definition~\ref{def:pswap} is not computable, 
and so the relation $\pswap{}{}$ is undecidable.
This is because strong swappability abstracts from the actual static analysis 
used to compute the safe approximations:
it just assumes that these approximations exist.
Definition~\ref{def:pswapWR} below parameterises strong swappability over 
a static analysis, which we render as a function 
from transactions to sets of observables.
Formally, $\wset{}$ is a \emph{static analysis of written observables}
when $\wapprox{\wset{\txT}}{\txT}$, for all $\txT$;
similarly, $\rset{}$ is a \emph{static analysis of read observables}
when $\rapprox{\rset{\txT}}{\txT}$, for all $\txT$.

\begin{defi}[{Parameterised strong swappability}]
  \label{def:pswapWR}
  Let $\wset{}$ and $\rset{}$ be static analyses of written/read observables.
  We say that two transactions $\txT \neq \txTi$ are
  \emph{strongly swappable \wrt $\wset{}$ and $\rset{}$},
  in symbols $\txT \pswapWR \txTi$, if:
  \[
  \big( \rset{\txT} \cup \wset{\txT} \big) \cap \wset{\txTi}
  \; = \;
  \emptyset
  \; = \;
  \big( \rset{\txTi} \cup \wset{\txTi} \big) \cap \wset{\txT}
  \]
\end{defi}

By the definition, it directly follows that 
$\txT \pswapWR \txTi$ implies that $\txT \pswap{}{} \txTi$.
Further, if $\wset{}$ and $\rset{}$ are computable, 
then $\pswapWR$ is decidable.
Later on, we will show that the relations $\pswapWR$ and $\pswap{}{}$ 
are equivalent in Bitcoin (Theorem~\ref{th:pswap-implies-pswapWR:btc}).

\begin{lem}
  \label{lem:commutativity-writes}
  \(
  \txT \pswapWR \txTi \implies \txT\txTi \sim_{\wset{\txT}} \txTi\txT
  \)
\end{lem}
\begin{proof}
  By Definition~\ref{def:pswapWR}, 
  $\wset{\txT} \cap \wset{\txTi} = \emptyset$ and
  $\rset{\txT} \cap \wset{\txTi} = \emptyset$.
  Since $\wapprox{\wset{\txTi}}{\txTi}$, 
  by Definition~\ref{def:safeapprox} we have
  $\txTi \sim_{\wset{\txT}} \bcEmpty$ and
  $\txTi \sim_{\rset{\txT}} \bcEmpty$.
  Since $\sim_{\wset{\txT}}$ is a congruence, $\txT\txTi \sim_{\wset{\txT}} \txT$.
  Since $\rapprox{\rset{\txT}}{\txT}$, $\txTi \sim_{\rset{\txT}} \bcEmpty$
  and $\txTi \sim_{\wset{\txT}} \bcEmpty$,
  by Definition~\ref{def:safeapprox} we have
  $\txTi \txT \sim_{\wset{\txT}} \txT$.
  By simmetry and transitivity of $\sim$ (Lemma~\ref{lem:sim:equiv}), 
  we conclude $\txT\txTi \sim_{\wset{\txT}} \txTi\txT$.
\end{proof}

The following lemma states that the relation $\pswapWR$ 
is a sound approximation of swappability.

\begin{lem}
  \label{lem:pswapWR-implies-swap}
  $\txT \pswapWR \txTi \implies \txT \swap \txTi$
\end{lem}
\begin{proof}
  By applying Lemma~\ref{lem:commutativity-writes} twice, we obtain
  $\txT\txTi \sim_{\wset{\txT}} \txTi\txT$ and 
  $\txTi\txT \sim_{\wset{\txTi}} \txT\txTi$.
  Let $\QmvA = \QmvU \setminus (\wset{\txT} \cup \wset{\txTi})$.
  Since $\QmvA \cap \wset{\txT} = \emptyset = \QmvA \cap \wset{\txTi}$,
  by applying Definition~\ref{def:safeapprox} twice we obtain
  $\bcEmpty \sim_{\QmvA} \txT$ and $\bcEmpty \sim_{\QmvA} \txTi$.
  Since $\sim_{\QmvA}$ is a congruence, $\txT\txTi \sim_{\QmvA} \txTi\txT$.
  Summing up:
  \[
  \txT\txTi \sim_{\QmvA \,\cup\, (\wset{\txT} \cup \wset{\txTi})} \txTi\txT
  \]
  from which we obtain the thesis, since 
  $\QmvA \cup (\wset{\txT} \cup \wset{\txTi}) = \QmvU$ and
  $\sim_{\QmvU} \, = \, \sim$.
\end{proof}

Note that if $\txT \pswap{}{} \txTi$, then there exist
$\wset{}$ and $\rset{}$ such that $\txT \pswapWR \txTi$.
Then, from Lemma~\ref{lem:pswapWR-implies-swap} it follows
that $\txT$ and $\txTi$ are swappable.
This proves Theorem~\ref{th:pswap-implies-swap},
from which in turns we obtain Theorem~\ref{th:pswap:mazurkiewicz}.
Putting it all together, we have proved the inclusions:

\begin{restatable}{thm}{thpswapWRmazurkiewicz}
  \label{th:pswapWR:mazurkiewicz}
  $\equivStSeq[\pswapWR] \;\; \subseteq \;\; \equivStSeq[\pswap{}{}] \;\; \subseteq \;\; \equivStSeq[\swap]$.
\end{restatable}

  \subsection{Swapping transactions in Bitcoin}
\label{sec:txswap:btc}

By instantiating our general blockchain model to Bitcoin,
we can refine some of the swappability results presented before.
In particular, in Bitcoin we can easily construct safe approximations
of the observables read/written by a transaction, by just considering
their inputs and outputs (Lemma~\ref{lem:safeapprox:btc}).
Further, while strong swappability is stricter than swappability
(Example~\ref{cex:pswap-implies-swap:btc}),
strong and parameterized strong swappability coincide in Bitcoin
(Theorem~\ref{th:pswap-implies-pswapWR:btc}).

The following lemma provides the least safe approximations
of the observables read and written by consistent transactions.
For inconsistent transactions, these approximations are just
the empty set (Lemma~\ref{lem:inconsistent-emptyset}).
Intuitively, the observables written by $\txT$ can be approximated as
$\txT.\txIn{} \cup \txT.\txOut{}$, because $\txT$ spends all the 
transaction outputs in $\txT.\txIn{}$, 
and creates the transaction outputs in $\txT.\txOut{}$.
Instead, the read observables can be approximated as $\txT.\txIn{}$,
since by Definition~\ref{def:blockchain-btc:valid},
executing $\txT$ from two states which agree on $\txT.\txIn{}$
leads to two states which only differ on the observables
for which they differed before.

\begin{lem}
  \label{lem:safeapprox:btc}
  Let $\txT$ be a consistent Bitcoin transaction, and let:
  \[
  W \; = \; \txT.\txIn{} \cup \txT.\txOut{} 
  \qquad\qquad
  R \; = \; \txT.\txIn{}
  \]
  Then, $W$ (resp.\ $R$) is the least safe approximation
  of written (resp.\ read) observables.
\end{lem}
\begin{proof}
  We first show that $W$ and $R$ are safe approximations
  of written / read observables,
  and then that they are the least ones.
  \begin{itemize}
    
  \item \emph{$W$ is a safe approximation of written observables.}
    Let $\QmvB$ be such that $\QmvB \cap (\txT.\txIn{} \cup \txT.\txOut{}) = \emptyset$.
    For all blockchain states $\bcSt$, 
    we have that $\semBc{\emptyseq}{\bcSt} = \bcSt$, and:
    \[
    \semBc{\txT}{\bcSt} = \bcSti 
    \qquad \text{where }
    U_{\bcSti} = \begin{cases}
      (U_{\bcSt} \setminus \txT.\txIn{}) \cup \txT.\txOut{} & \text{if $\bcSt \valid \txT$} \\
      U_{\bcSt} & \text{otherwise}
    \end{cases}
    \]
    Since $Q$ and $\txT.\txIn{} \cup \txT.\txOut{}$ are disjoint, 
    we have that $\bcSt \sim_{\QmvB} \bcSti$.
    Therefore, $\txT \sim_{\QmvB} \bcEmpty$.
    By Definition~\ref{def:safeapprox}, it follows that 
    $\wapprox{\txT.\txIn{} \cup \txT.\txOut{}}{\txT}$.

  \item \emph{$R$ is a safe approximation of read observables.}
    Assume that $\bcB[0] \sim_{\txT.\txIn{}} \bcB[1]$.
    Then, $\txT$ is valid in $\bcB[0]$ iff it is valid in $\bcB[1]$.
    Let $\QmvB$ be such that $\bcB[0] \sim_{\QmvB} \bcB[1]$,
    let $\bcSt[0] = \semBc{\bcB[0]}{}$ and $\bcSt[1] = \semBc{\bcB[1]}{}$.
    For $i \in \setenum{0,1}$, we have that:
    \[
    \semBc{\txT}{\bcSt[i]} = \bcSti[i] 
    \qquad \text{where }
    U_{\bcSti[i]} = \begin{cases}
      (U_{\bcSt[i]} \setminus \txT.\txIn{}) \cup \txT.\txOut{} & \text{if $\bcSt[i] \valid \txT$} \\
      U_{\bcSt[i]} & \text{otherwise}
    \end{cases}
    \]
    Since $\bcSt[0] \sim_{\QmvB} \bcSt[1]$, it follows that
    $\bcSti[0] \sim_{\QmvB} \bcSti[1]$, 
    and therefore $\bcB[0]\txT \sim_{\QmvB} \bcB[1]\txT$.
    By Definition~\ref{def:safeapprox}, it follows that   
    $\rapprox{\txT.\txIn{}}{\txT}$.

  \item \emph{$W$ is the least safe approximation of written observables.}
    By contradiction, 
    let $W' \subsetneq W$
    be such that $\wapprox{W'}{\txT}$,
    and let $R' \subsetneq R$
    be such that $\rapprox{R'}{\txT}$.  
    Let $\qmvA \in W \setminus W'$.
    Since $\setenum{\qmvA} \cap W' = \emptyset$ and 
    $\wapprox{W'}{\txT}$,
    then $\txT \sim_{\setenum{\qmvA}} \emptyseq$.
    Since $\txT$ is consistent, there exists $\bcSt$ such that
    $\bcSt \valid \txT$.
    Then, 
    $\semBc{\txT}{\bcSt} = \bcSti$, where
    $U_{\bcSti} = (U_{\bcSt} \setminus \txT.\txIn{}) \cup \txT.\txOut{}$.
    Since $\qmvA \in \txT.\txIn{} \cup \txT.\txOut{}$,
    it follows that $\bcSti\qmvA \neq \bcSt\qmvA$, and so
    $\txT \not\sim_{\setenum{\qmvA}} \emptyseq$ --- contradiction.
    Therefore, $W$ is the least approximation of the observables
    written by $\txT$.
    
  \item \emph{$R$ is the least safe approximation of read observables.}
    Let $\QmvB = \txT.\txOut{}$.
    Since $\txT$ is consistent, 
    there exists a reachable $\bcSti$ such that $\bcSti \valid \txT$.
    Since $\bcSti$ is reachable, there exists $\bcBi$ such that
    $\semBc{\bcBi}{} = \bcSti$.
    Let $\bcB = \bcBi \txTi$, where $\txTi$ spends $R \setminus R'$.
    Then, $\bcB \sim_{R'} \bcBi$ and $\bcB \sim_{\QmvB} \bcBi$.
    Since $\rapprox{R'}{\txT}$, it must be
    $\bcB\txT \sim_{\QmvB} \bcBi\txT$.
    Since in $\bcB$ some of the inputs needed by $\txT$ have been spent,
    we have that $\txT$ is not valid in $\bcB$, and so 
    $\semBc{\bcB\txT}{} = \semBc{\bcB}{}$.
    On the other hand, since $\txT$ is valid in $\bcBi$, then
    $\semBc{\bcBi\txT}{} = \bcStii$, where 
    $U_{\bcStii} = (U_{\bcSti} \setminus \txT.\txIn{}) \cup \txT.\txOut{}$.
    Since $\QmvB = \txT.\txOut{}$ belongs to $\semBc{\bcBi\txT}{}$ 
    but not to $\semBc{\bcB\txT}{}$, 
    it follows that $\bcB\txT \not\sim_{\QmvB} \bcBi\txT$ --- contradiction.
    Therefore, $R$ is the least approximation of the observables
    read by $\txT$.
    \qedhere

  \end{itemize}
\end{proof}

\begin{lem}
  \label{lem:inconsistent-emptyset}
  $\txT$ is inconsistent if and only if 
  $\wapprox{\emptyset}{\txT}$ and $\rapprox{\emptyset}{\txT}$.
\end{lem}
\begin{proof}
  For the ``only if'' part, assume that $\txT$ is inconsistent.
  For $\wapprox{}{}$, for all $\bcSt$ we have that 
  $\semBc{\txT}{\bcSt} = \bcSt = \semBc{\emptyseq}{\bcSt}$.
  Therefore, for all $\QmvB$ e have $\txT \sim_{\QmvB} \emptyseq$,
  from which it follow that $\wapprox{\emptyset}{\txT}$.
  For all $\rapprox{}{}$, for all $\bcB$, $\bcBi$ and $\QmvB$ we have that
  if $\bcB \sim_{\QmvB} \bcBi$ then $\bcB\txT \sim_{\QmvB} \bcBi\txT$.
  Therefore, $\rapprox{\emptyset}{\txT}$.
  For the ``if'' part, assume that $\wapprox{\emptyset}{\txT}$.
  Then, for all $\QmvB$ it must be $\txT \sim_{\QmvB} \emptyseq$, 
  \ie $\txT \sim \emptyseq$.
  By definition of $\sim$ this implies that, for all $\bcSt$,
  $\semBc{\txT}{\bcSt} = \semBc{\emptyseq}{\bcSt} = \bcSt$.
  Therefore, $\txT$ is not valid in any blockchain state $\bcSt$, 
  and so $\txT$ is inconsistent.
\end{proof}

By exploiting the results above, we can provide an alternative
sufficient condition for (strong) swappability.
If $\txT$ is valid in some state 
where it is also possible to append another transaction $\txTi$
before $\txT$ (\ie, $\txTi \txT$ is valid in that state), 
then $\txT$ and $\txTi$ are strongly swappable.
This is a peculiar property of UTXO-based blockchains like Bitcoin:
in Example~\ref{ex:contextual-txswap:eth} we show that this is not 
the case for Ethereum.

\begin{lem}
  \label{lem:contextual-txswap:btc}
  In Bitcoin,
  if there exists $\bcSt$ such that 
  $\bcSt \valid \txT$ and $\bcSt \valid \txTi \txT$,
  then $\txT \pswap{}{} \txTi$.
\end{lem}
\begin{proof}
  Let $\bcSt$ be such that $\bcSt \valid \txT$ and $\bcSt \valid \txTi \txT$.
  Then, by condition~\eqref{consistent-update:unspent} 
  of Definition~\ref{def:blockchain-btc:valid}:
  \begin{equation} 
    \label{eq:contextual-txswap:btc}
    \txTi.\txIn{} \cap \txT.\txOut{} 
    \;\; = \;\;
    \txT.\txIn{} \cap \txTi.\txOut{} 
    \;\; = \;\;
    \txT.\txIn{} \cap \txTi.\txIn{} 
    \;\; = \;\;
    \txT.\txOut{} \cap \txTi.\txOut{}
    \;\; = \;\;
    \emptyset  
  \end{equation}
  By Lemma~\ref{lem:safeapprox:btc},
  $\txT.\txIn{} \cup \txT.\txOut{}$ and $\txT.\txIn{}$
  are safe approximations of written/read observables.
  By~\eqref{eq:contextual-txswap:btc}, these approximations satisfy 
  the condition of Definition~\ref{def:pswap}, 
  and so $\txT \pswap{}{} \txTi$.
\end{proof}

The following example shows that the converse of
Theorem~\ref{th:pswap-implies-swap} does not hold in Bitcoin, \ie
there exist transactions which are swappable but not strongly swappable.

\begin{exa}[{Swappable transactions, but not strongly}]
  \label{cex:pswap-implies-swap:btc}
  Consider the transactions in Figure~\ref{fig:pswap-implies-swap:btc},
  where the scripts and currency values are immaterial
  (we just assume that condition~\eqref{consistent-update:output} 
  of Definition~\ref{def:blockchain-btc:valid} is satisfied for each
  matching input/output pair).
  We show that $\txT[1]$ and $\txT[3]$ are swappable.
  Let $\bcSt$ be a blockchain state.
  If $\txT[1]$ is not valid in $\bcSt$, 
  then $\semBc{\txT[1]\txT[3]}{\bcSt} = \semBc{\txT[1]\txT[3]}{\bcSt}$
  holds trivially, since also $\txT[3]$ is not valid.
  Otherwise, if $\bcSt \valid \txT[1]$:
  \begin{align*}
    \semBc{\txT[1]\txT[3]}{\bcSt} 
    & = \semBc{\txT[3]}{\bcSti}
    && \text{where $U_{\bcSti} = (U_{\bcSt} \setminus \txT[1].\txIn{}) 
       \cup \setenum{(\txT[1],1),\, (\txT[1],2)}$}
    \\
    & = \bcSti
    && \text{since $\bcSti \nvalid \txT[3]$, as $(\txT[2],1) \not\in U_{\bcSti}$}
    \\
    \semBc{\txT[3]\txT[1]}{\bcSt} 
    & = \semBc{\txT[1]}{\bcSt}
    && \text{since $\bcSt \nvalid \txT[3]$, as $(\txT[1],1) \not\in U_{\bcSt}$}
    \\
    & = \bcSti
    && \text{where $U_{\bcSti} = (U_{\bcSt} \setminus \txT[1].\txIn{}) 
       \cup \setenum{(\txT[1],1),\, (\txT[1],2)}$}
  \end{align*}
  Therefore, $\txT[1]$ and $\txT[3]$ are swappable.
  We now show that they are not \emph{strongly} swappable.
  Assume that $\txT[1]$ is consistent. 
  Then, by Lemma~\ref{lem:safeapprox:btc} it follows that
  $W_1 = \txT[1].\txIn{} \cup \txT[1].\txOut{}$ is the least safe approximation
  of the observables written by $\txT[1]$.
  Let $\bcSt$ be such that $\txT[1]$ is valid in~$\bcSt$.
  Then, $\txT[3]$ is valid in $\semBc{\txT[1]\txT[2]}{\bcSt}$,
  and so by Lemma~\ref{lem:safeapprox:btc} is also follows that
  $W_3 = \txT[3].\txIn{} \cup \txT[3].\txOut{}$ 
  is the least safe approximation of the observables written by $\txT[3]$.
  Since $(\txT[1],1) \in W_1 \cap W_3$,
  then $\txT[1]$ and $\txT[3]$ are not strongly swappable.
\end{exa}

\newcommand{\txFigBtcA}{%
  \begin{tabular}[t]{|l|l|}
    \hline
    \multicolumn{2}{|c|}{$\txT[1]$} \\
    \hline
    \multicolumn{2}{|l|}{$\txIn[1]{\cdots}$} \\
    \hline
    $\txOut[1]{\cdots}$ & $\txOut[2]{\cdots}$ \\
    \hline
  \end{tabular}
}

\newcommand{\txFigBtcB}{%
  \begin{tabular}[t]{|l|l|}
    \hline
    \multicolumn{2}{|c|}{$\txT[2]$} \\
    \hline
    \multicolumn{2}{|l|}{\txIn[1]{$(\txT[1],2)$}} \\
    \hline
    \multicolumn{2}{|l|}{\txOut[1]{$\cdots$}} \\
    \hline
  \end{tabular}
}

\newcommand{\txFigBtcC}{%
  \begin{tabular}[t]{|l|l|}
    \hline
    \multicolumn{2}{|c|}{$\txT[3]$} \\
    \hline
    \txIn[1]{$(\txT[1],1)$} & \txIn[2]{$(\txT[2],1)$} \\
    \hline
    \multicolumn{2}{|l|}{\txOut[1]{$\cdots$}}
    \\
    \hline
  \end{tabular}
}

\begin{figure}[t]
  \centering\footnotesize
  \begin{tikzpicture}
    \node at (-1,2) {\txFigBtcA};
    \node at (8,2) {\txFigBtcB};
    \node at (4,0.2) {\txFigBtcC};
    \draw [->] (-2,1.4) -- (-2,0.2) -- (1.8,0.2);
    \draw [->] (0.9,1.6) -- (4,1.6) -- (4,2) -- (6.9,2);
    \draw [->] (7.8,1.4) -- (7.8,0.2) -- (6.2,0.2);
  \end{tikzpicture}
  \caption{Transactions $\txT[1]$ and $\txT[3]$ are swappable but not strongly swappable.}
  \label{fig:pswap-implies-swap:btc}
\end{figure}

Finally, we prove that 
strong and parameterized strong swappability coincide in Bitcoin.

\begin{thm}
  \label{th:pswap-implies-pswapWR:btc}
  Let $\txT$, $\txTi$ be consistent Bitcoin transactions.
  If \mbox{$\txT \pswap{}{} \txTi$}, then $\txT \pswapWR \txTi$,
  using the static analyses $\wset{\txT} = \txT.\txIn{} \cup \txT.\txOut{}$ and
  $\rset{\txT} = \txT.\txIn{}$.
\end{thm}
\begin{proof}
  Since $\txT \pswap{}{} \txTi$, then there exist safe approximations
  $\wapprox{W}{\txT}$,
  $\wapprox{W'}{\txTi}$,
  $\rapprox{R}{\txT}$, and
  $\rapprox{R'}{\txTi}$
  such that
  \(
    \big( R \cup W \big) \cap W'
    =
    \emptyset
    =
    \big( R' \cup W' \big) \cap W
  \).
  Since $\txT$ and $\txTi$ are consistent,
  then by Lemma~\ref{lem:safeapprox:btc}
  the static analyses $\wset{\txT} = \txT.\txIn{} \cup \txT.\txOut{}$ and
  $\rset{\txT} = \txT.\txIn{}$
  give their \emph{least} safe approximation of written/read observables.
  Then, $\wset{\txT} \subseteq W$, $\wset{\txTi} \subseteq W'$,
  $\rset{\txT} \subseteq R$, and $\rset{\txTi} \subseteq R'$.
  Then, the disjointess condition required by Definition~\ref{def:pswapWR}
  holds for the static analyses, and so $\txT \pswapWR \txTi$.
\end{proof}

  \subsection{Swapping transactions in Ethereum}
\label{sec:txswap:eth}

We now illustrate our notions of swappability in Ethereum
through a series of examples.
We postpone to~\Cref{sec:validation} a discussion on how to 
approximate the observables read/written by Ethereum contracts,
so to compute the parameterized strong swappability relation.

\begin{exa}[{Swappability}]
  Recall the contract $\cmvA$ and the blockchain 
  $\bcB = \txT[0] \txT[1] \txT[2]$ 
  from Example~\ref{ex:th-sem}.
  By Definition~\ref{def:swap}, we have that: 
  \begin{itemize}

  \item $\txT[0] \swap \txT[2]$ (see \Cref{fig:swap:nocongruence}, top left).
    Indeed, regardless of whether $\txT[0]$ is appended to the blockchain 
    before or after $\txT[2]$, after their execution we obtain the same state:  $\bcSti\setenum{\bind{\cmvA.x}{1}}$, when $\bcSt \pmvA.\varBalance \geq 1$, and $\bcSt\setenum{\bind{\cmvA.x}{1}}$ otherwise. 

  \item $\txT[1] \nswap \txT[2]$ (see \Cref{fig:swap:nocongruence}, top right).
    Let $\bcSt$ be such that $\bcSt \cmvA. x \neq 0$ and $\bcSt \pmvA.\varBalance = 1$. 
    If we append $\txT[1]$ before $\txT[2]$ we obtain the state $\bcSti = \bcSt - \pmvA:1 + \cmvA :1$ and in this state $\txT[2]$ is idempotent.  
    Instead, the result of executing $\txT[2]$ before $\txT[1]$ is the state $\bcStii = \bcSt - \pmvA:1 + \pmvB : 1$ and in this state $\txT[1]$ is idempotent.
    Therefore, $\txT[1] \nswap \txT[2]$.

  \item $\txT[0] \nswap \txT[1]$ (see \Cref{fig:swap:nocongruence}, bottom).
    Depending on how we append $\txT[0]$ and $\txT[1]$ we obtain two different states. 
    Let $\bcSt$ be such that $\bcSt \cmvA.\valX = 0$ and $\bcSt\pmvA.\varBalance \geq 1$.
    If we append $\txT[1]$ before $\txT[0]$ we obtain the state $\bcSti\setenum{\bind{\cmvA.x}{1}}$.
    Instead, if we append $\txT[0]$ and then $\txT[1]$ we obtain the state $\bcSt\setenum{\bind{\cmvA.x}{1}} - \pmvA:1 + \cmvA : 1$. 

  \end{itemize}
\end{exa}

\begin{figure*}[t]
  \centering
  \hspace{-10pt}
  \begin{subfigure}[b]{0.3\textwidth}
    \adjustbox{scale=0.8,center}{
      \begin{tikzcd}[column sep=normal, row sep=huge]
        & \bcSti = \bcSt - \pmvA:1 + \pmvB:1  \ar[rd, "{\txT[0]}" sloped, end anchor=north west, start anchor=south east] & \\
        \bcSt
        \ar[r, "{\txT[0]}"]
        \ar[ur, start anchor=north east, end anchor=south west,"{\bcSt\pmvA \varBalance \geq 1}"' pos=0, sloped, "{\txT[2]}" near end]
        \ar[dr, start anchor=south east, end anchor=north west,"{\txT[2]}", sloped, "\bcSt\pmvA \varBalance < 1"' near end]
        & \bcSt\setenum{\bind{\cmvA.x}{1}} \ar[r, "\bcSt\pmvA \varBalance \geq 1"', sloped, "{\txT[2]}"] \ar[rd, "\bcSt\pmvA \varBalance < 1"' near end, sloped, "{\txT[2]}"]
        & \bcSti\setenum{\bind{\cmvA.x}{1}}  \\
        & \bcSt \ar[r, "{\txT[0]}"]
        & \bcSt\setenum{\bind{\cmvA.x}{1}}
      \end{tikzcd}
    }
    % \vspace{-35pt}
    \subcaption{Proof of $\txT[0] \swap \txT[2]$.}
    \label{fig:proof:t0:t2}
  \end{subfigure}
  \hspace{70pt}
  \begin{subfigure}[b]{0.3\textwidth}
    \adjustbox{scale=0.8,center}{
      \begin{tikzcd}[column sep=small, row sep=small]
        & \bcSti = \bcSt - \pmvA:1 + \cmvA : 1 \ar[loop right,"{\txT[2]}", end anchor=east]  & \\
        \bcSt \ar[ur, "{\txT[1]}", end anchor=south west] \ar[dr, "{\txT[2]}"', end anchor=north west]& \begin{scriptsize}
          \begin{array}{l}
            \bcSt\cmvA \valX \neq 0 \\
            % \bcSt\pmvB\varBalance = 0 \\
            \bcSt\pmvA\varBalance = 1
          \end{array}
        \end{scriptsize}   \\
        & \bcStii = \bcSt - \pmvA:1 + \pmvB : 1 \ar[loop right,"{\txT[1]}", end anchor=east]
      \end{tikzcd}
    }
    % \vspace{-35pt}
    \subcaption{Proof of $\txT[1] \nswap \txT[2]$.}
    \label{fig:proof:t1:t2}
  \end{subfigure}
  
  \begin{subfigure}[b]{0.33\textwidth}
    \adjustbox{scale=0.8,center}{   % ,raise=15pt only for single column!!
      \begin{tikzcd}[column sep=small, row sep=normal, negated/.style={
          decoration={markings,
            mark= at position 0.5 with {
              \node[transform shape] (tempnode) {$\slash$};
            }
          },
          postaction={decorate}
        }
        ]
        & \bcSti = \bcSt - \pmvA:1 + \pmvB:1 \ar[r, "{\txT[0]}"] & \bcSti\setenum{\bind{\cmvA.x}{1}} \ar[dd, equal, negated] \\
        \bcSt
        \ar[ru, end anchor=south west, "{\txT[1]}", sloped]  \ar[rd, end anchor=north west, "{\txT[0]}"', sloped]
        &
        \begin{scriptsize}
          \begin{array}{l}
            \bcSt\cmvA \valX = 0 \\
            % \bcSt\pmvB\varBalance = 0 \\
            \bcSt\pmvA\varBalance \geq 1
          \end{array}
        \end{scriptsize}
        &
        \\
        &  \bcSt\setenum{\bind{\cmvA.x}{1}} \ar[r, "{\txT[1]}"] & \bcSt\setenum{\bind{\cmvA.x}{1}} - \pmvA:1 + \cmvA:1
      \end{tikzcd}
    }
    \subcaption{Proof of $\txT[0] \nswap \txT[1]$.}
  \end{subfigure}
  \caption{Proofs for $\txT[0] \swap \txT[2]$ and $\txT[0] \nswap \txT[1]$. A transition $\txT$ from $\bcSt$ can be taken only if the guard below the arrow is satisfied in $\bcSt$.}
  \label{fig:swap:nocongruence}
\end{figure*}

\begin{exa}[{Strong swappability}]
  Let $\cmvA$ be the contract of Example~\ref{ex:th-sem}, and let
  $\contrFun{\funF[3]}{}{\cmdSkip}$ be a function of a contract $\cmv{D}$.  
  Then, consider the following transactions:
  \begin{align*}
    & \txT[3] = \ethtx{1}{\pmvA}{\cmv{D}}{\funF[3]}{}
    && \txT[4] = \ethtx{1}{\pmvB}{\cmvA}{\funF[2]}{\cmv{F}}
  \end{align*}
  where $\pmvA$, $\pmvB$, and $\cmv{F}$ are account addresses.
  Intuitively, $\txT[3] \pswap{}{} \txT[4]$ because they are operating on observables of different addresses.
  Formally, consider the following safe approximations of the written/read observables of $\txT[3]$ and $\txT[4]$:
  \begin{align*}
    & \wapprox{W_3 = \setenum{\pmvA.\varBalance, \cmv{D}.\varBalance}}{\txT[3]}
    && \rapprox{R_3 = \setenum{\pmvA.\varBalance}}{\txT[3]}
    \\
    & \wapprox{W_4 = \setenum{\pmvB.\varBalance, \cmv{F}.\varBalance}}{\txT[4]}
    && \rapprox{R_4 = \setenum{\pmvB.\varBalance}}{\txT[4]}
  \end{align*}
  Since $(W_3 \cup R_3) \cap W_4 = \emptyset = (W_4 \cup R_4) \cap W_3$, the two transactions are strongly swappable.

  Now, consider the following transaction that calls the function $\funF[2]$ with the address $\pmvA$:
  \[
  \txT[5] = \ethtx{1}{\pmvB}{\cmvA}{\funF[2]}{\pmvA}
  \]
  This transaction transfers $1$ currency unit from $\pmvB$ to $\pmvA$.
  Intuitively, since $\txT[5]$ touches $\pmvA.\varBalance$, 
  then it should not be swappable with $\txT[3]$ and $\txT[4]$.
  Formally, consider the following safe approximations $W_5$ and $R_5$:
  \begin{align*}
    & \wapprox{W_5 = \setenum{\pmvB.\varBalance, \pmvA.\varBalance}}{\txT[5]}
      \qquad
      \rapprox{R_5 = \setenum{\pmvB.\varBalance}}{\txT[5]}
  \end{align*}
  Since $W_3 \cap W_5 \ne \emptyset \ne W_4 \cap W_5$,
  then $\neg (\txT[3] \pswap{}{} \txT[5])$ and
  $\neg(\txT[4] \pswap{}{} \txT[5])$.
\end{exa}

The following example shows that the converse of 
Theorem~\ref{th:pswap-implies-swap} does not hold, 
\ie there may exist transactions that are swappable but not strongly swappable. 
This is because of static analyses could produce false negatives.

\begin{exa}[{Swappable transactions, but not strongly}]
  \label{cex:pswap-implies-swap:eth}
  Consider the following functions of a contract $\cmvA[1]$, 
  and the following transactions sent by users $\pmvA$ and $\pmvB$:
  \begin{align*}
    & \contrFun{\funH[1]}{}{\cmdIfTE{\varSender = \pmvA \;\codeand\; \expLookup{\valK[1]} = 0}{\cmdAss{\valK[1]}{1}}{\cmdThrow}} 
    && \txT[1] = \ethtx{1}{\pmvA}{\cmvA[1]}{\funH[1]}{}
    \\
    & \contrFun{\funH[2]}{}{\cmdIfTE{\varSender = \pmvB \;\codeand\; \expLookup{\valK[2]} = 0}{\cmdAss{\valK[2]}{1}}{\cmdThrow}}
    && \txT[2] = \ethtx{1}{\pmvB}{\cmvA[1]}{\funH[2]}{}
  \end{align*}
  We have that $\txT[1]$ and $\txT[2]$ are swappable. 
  To see why, consider the following two cases:
  \begin{enumerate}
  \item a state $\bcSt$ where $\bcSt \pmvA.\varBalance > 1$,
    $\bcSt \pmvB.\varBalance > 1$,  $\bcSt \cmvA[1].\varBalance = n$, 
    $\bcSt \cmvA[1].\valK[1] = 0$ and $\bcSt \cmvA[1].\valK[2] = 0$.
    In $\bcSt$ it holds that:
    \[
    \semTx{\txT[1] \txT[2]}{\bcSt} = \bcSt\setenum{\bind{\cmvA[1].\valK[1]}{1},
      \bind{\cmvA[1].\valK[2]}{1},\bind{\cmvA[1].\varBalance}{n + 2}} = \semTx{\txT[2] \txT[1]}{\bcSt}
    \]
  \item a state $\bcSt$ such that $\bcSt \pmvA.\varBalance < 1$, 
    or $\bcSt \pmvB.\varBalance < 1$, or $\bcSt \cmvA[1].\valK[1] \neq 0$, 
    or $\bcSt \cmvA[1].\valK[2] \neq 0$.
    Since it is not possible that the guards of $\funH[1]$ and $\funH[2]$ 
    are both true, one of $\txT[1]$ 
    or $\txT[2]$ raises an exception, leaving the state unaffected.
    Then, also in this case we have that
    \[
    \semTx{\txT[1] \txT[2]}{\bcSt} = \semTx{\txT[2] \txT[1]}{\bcSt}
    \] 
  \end{enumerate}
  
  However, $\txT[1]$ and $\txT[2]$ are \emph{not} strongly swappable.
  Intuitively, this is because there exist reachable states $\bcSt,\bcSti$ such that $\bcSt \cmvA[1]\valK[1] = 0 = \bcSti \cmvA[1]\valK[2]$.
  Formally, consider the following sets 
  \begin{align*}
    W_1 = \setenum{\expGet{\pmvA}{\varBalance}, \expGet{\cmvA[1]}{\varBalance},
    \expGet{\cmvA[1]}{\valK[1]}}
    &&
       W_2 = \setenum{\expGet{\pmvB}{\varBalance},
       \expGet{\cmvA[1]}{\varBalance}, \expGet{\cmvA[1]}{\valK[2]}}
  \end{align*}
  which are the least safe over-approximations of the written observables 
  by $\txT[1]$ 
  % ($\wapprox{W_1}{\txT[1]}$) 
  and by $\txT[2]$, 
  % ($\wapprox{W_2}{\txT[2]}$), 
  respectively.
  This means that every safe approximation of $\txT[1]$ must include the observables of $W_1$, and  similarly for the set $W_2$. 
  Since $W_1 \cap W_2 \neq \emptyset$, 
  then $\txT[1] \pswap{}{} \txT[2]$ does not hold.
\end{exa}

The following example shows that Lemma~\ref{lem:contextual-txswap:btc},
which is specific to Bitcoin and UTXO-based blockchains, 
does not hold on Ethereum.

\begin{exa}
  \label{ex:contextual-txswap:eth}
  Recall the functions $\contrFun{\funF[0]}{}{\cmdAss{\valX}{1}}$ 
  and $\contrFun{\funF[1]}{}{\cmdIfT{\valX=0}{\cmdSend{1}{\pmvB}}}$ 
  from Example~\ref{ex:th-sem},
  and consider the following transactions:
  \[
  \txT[1] = \ethtx{1}{\pmvA}{\cmvA}{\funF[1]}{}
  \qquad
  \txT[5] = \ethtx{1}{\pmvA}{\cmvA}{\funF[0]}{}
  \]
  Let $\bcSt$ be a state such that $\bcSt \, \cmvA.x = 0$
  and $\bcSt \, \pmvA.\varBalance \geq 2$. 
  We have that $\bcSt \valid \txT[1]$, and so 
  $\semBc{\txT[1]}{\bcSt} = \bcSt - \pmvA : 1 + \pmvB: 1$.
  Further, $\bcSt \valid \txT[5]\txT[1]$ and 
  $\bcSt \valid \txT[1]\txT[5]$.
  Then:
  \[
  \bcSt[5,1] 
  = 
  \semBc{\txT[5]\txT[1]}{\bcSt} = \bcSt\setenum{\bind{\cmvA.x}{1}} - \pmvA : 2 + \cmvA: 1
  \qquad
  \bcSt[1,5]
  = 
  \semBc{\txT[1]\txT[5]}{\bcSt} 
  = 
  \bcSt\setenum{\bind{\cmvA.x}{1}} - \pmvA : 2 + \pmvB: 1
  \]
  Hence, $\txT[1]$ and $\txT[5]$ are not swappable, because
  $\bcSt[1,5]$ and $\bcSt[5,1]$ differ in the balances of $\cmvA$ and $\pmvB$.
\end{exa}

\section{True concurrency for blockchains}
\label{sec:txpar}

Given a swappability relation $\relR$, 
we transform a sequence of transactions $\bcB$ into an 
\emph{occurrence net} $\PNet{\relR}{\bcB}$,
which describes the partial order induced by $\relR$. 
Our main result is that any concurrent execution of the transactions in $\bcB$
which respects this partial order is equivalent 
to the serial execution of $\bcB$ (Theorem~\ref{th:bc-to-pnet}).

\paragraph{Occurrence nets}

We start by recapping the notion of Petri net~\cite{Reisig85book}.
A \emph{Petri net} is a tuple
$\netN = (\Places, \Transitions, \Arcs, \markM[0])$,
where $\Places$ is a set of \emph{places},
$\Transitions$ is a set of \emph{transitions}
(with $\Places \cap \Transitions = \emptyset$),
and
$\Arcs : (\Places \times \Transitions) \cup (\Transitions \times \Places) \rightarrow \Nat$
is a \emph{weight function}.
The state of a net is a \emph{marking}, \ie a multiset
$\markM: \Places \rightarrow \Nat$ defining how many \emph{tokens}
are contained in each place;
we denote with $\markM[0]$ the initial marking.
The behaviour of a Petri net is specified as 
a transition relation between markings:
intuitively, a transition $\trT$ is enabled at $\markM$ 
when each place $\placeP$ 
has at least $\Arcs(\placeP,\trT)$ tokens in $\markM$.
When an enabled transition $\trT$ is fired, it consumes 
$\Arcs(\placeP,\trT)$ tokens from each $\placeP$,
and produces $\Arcs(\trT,\placePi)$ tokens in each $\placePi$.
Formally, given $x \in \Places \cup \Transitions$, 
we define the \emph{preset} $\pre{x}$ and the \emph{postset} $\post{x}$ 
as multisets:
% $\pre{x}: \Places \cup \Transitions \rightarrow \Nat$ is defined as 
$\pre{x}(y) = \Arcs(y,x)$, and 
$\post{x}(y) = \Arcs(x,y)$.
A transition $\trT$ is \emph{enabled} at $\markM$
when $\pre{\trT} \subseteq \markM$.
The transition relation between markings is defined as
$\markM \trans{\trT} \markM - \pre{\trT} + \post{\trT}$, 
where $\trT$ is enabled.
We say that $\trT[1] \cdots \trT[n]$ is a 
\emph{firing sequence from $\markM$ to $\markMi$} when
$\markM \trans{\trT[1]} \cdots \trans{\trT[n]} \markMi$,
and in this case we say that
$\markMi$ is \emph{reachable from} $\markM$.
We say that $\markMi$ is \emph{reachable} when it is reachable from $\markM[0]$.

An \emph{occurrence net}~\cite{Best87tcs} is a Petri net such that:
\begin{inlinelist}
\item \label{item:petri:onet:p}
$\card{\post{\plP}} \leq 1$ for all $\plP$;
\item \label{item:petri:onet:mzero}
$\card{\pre{\plP}} = 1$ if $\plP \not\in \markM[0]$, and
$\card{\pre{\plP}} = 0$ if $\plP\in \markM[0]$;
\item \label{item:petri:onet:relation}
$\Arcs$ is a relation, \ie $\Arcs(x,y) \leq 1$ for all $x,y$;
\item \label{item:petri:onet:po}
$\Arcs^*$ is a acyclic, \ie
$\forall x,y \in \Places \cup \Transitions : (x,y) \in \Arcs^* \land (y,x) \in \Arcs^* \implies x=y$
(where $\Arcs^*$ is the reflexive and transitive closure of $\Arcs$).
\end{inlinelist}

\paragraph{From blockchains to occurrence nets}

We describe in Figure~\ref{def:bc-to-pnet} how to transform 
a blockchain $\bcB = \txT[1] \cdots \txT[n]$ 
into a Petri net $\PNet{\relR}{\bcB}$, 
where $\relR$ is an arbitrary relation between transactions.
Although any relation $\relR$ ensures that $\PNet{\relR}{\bcB}$ 
is an occurrence net (Lemma~\ref{lem:bc-to-pnet:onet}),
our main results hold when $\relR$ is a strong swappability relation.
The transformation works as follows:
the $i$-th transaction in $\bcB$ is rendered as a transition $(\txT[i],i)$ in $\PNet{\relR}{\bcB}$, and
transactions related by $\relR$ are transformed into concurrent transitions.
Technically, this concurrency is specified as 
a relation $<$ between transitions, such that
$(\txT[i],i) < (\txT[j],j)$ whenever $i < j$, 
but $\txT[i]$ and $\txT[j]$ are not related by $\relR$.
The places, the weight function, and the initial marking of $\PNet{\relR}{\bcB}$ 
are chosen to ensure that the firing ot transitions respects the relation $<$.

\begin{figure}[t]
  \begin{align*}
    & \Transitions 
    = \setcomp{(\txT[i],i)}{1 \leq i \leq n}
    \qquad\qquad \Places 
    \; = \setcomp{(*,\trT)}{\trT \in \Transitions} \cup 
      \setcomp{(\trT,*)}{\trT \in \Transitions} \cup \setcomp{(\trT,\trTi)}{\trT < \trTi}
    \\
    & \hspace{188pt} \text{where }
      (\txT,i) < (\txTi,j) \eqdef (i < j) \,\land\, \neg(\txT \,\relR\, \txTi)
    \\[5pt]
    & \Arcs(x,y) 
    = \begin{cases}
      1 & \text{if $y=\trT$ and \big($x=(*,\trT)$ or $x=(\trTi,\trT)$\big)} \\
      1 & \text{if $x=\trT$ and \big($y=(\trT,*)$ or $y =(\trT,\trTi)$\big)} \\
      0 & \text{otherwise}
    \end{cases}
    \hspace{40pt}
    \markM[0](\plP)
    = \begin{cases}
      1 & \text{if $\plP = (*,\trT)$} \\
      0 & \text{otherwise}
    \end{cases}
  \end{align*}
  \caption{Construction of a Petri net from a blockchain $\bcB = \txT[1] \cdots \txT[n]$.}
  \label{def:bc-to-pnet}
\end{figure}

\begin{restatable}{lem}{lembctopnetonet}
\label{lem:bc-to-pnet:onet}
  $\PNet{\relR}{\bcB}$ is an occurrence net, for all $\relR$ and $\bcB$.
\end{restatable}

\paragraph{Step firing sequences}

Theorem~\ref{th:bc-to-pnet} establishes a correspondence between 
concurrent and serial execution of transactions.
Since the semantics of serial executions is given in terms of blockchain states $\bcSt$, 
to formalise this correspondence we use the same semantics domain
also for concurrent executions.
This is obtained in two steps.
First, we define concurrent executions of $\bcB$  
as the \emph{step firing sequences} (\ie finite sequences of \emph{sets} of transitions) 
of $\PNet{\pswap{}{}}{\bcB}$.
Then, we give a semantics to step firing sequences, in terms of blockchain states.

We denote finite sets of transitions, called \emph{steps},
as $\trSU,\trSUi,\hdots$.
Their preset and postset are defined as
\(
\textstyle
\pre{\trSU} = \sum_{\plP \in \trSU} \pre{\plP} 
\)
and
\(
\textstyle
\post{\trSU} = \sum_{\plP \in \trSU} \post{\plP}
\),
respectively.
We say that $\trSU$ is \emph{enabled at $\markM$} 
when $\pre{\trSU} \leq \markM$, and in this case
firing $\trSU$ results in the move 
$\markM \trans{\trSU} \markM - \pre{\trSU} + \post{\trSU}$.
Let $\vec{\trSU} = \trSU[1] \cdots \trSU[n]$ be a finite sequence of steps.
We say that $\vec{\trSU}$ is a 
\emph{step firing sequence from $\markM$ to $\markMi$}
if $\markM \trans{\trSU[1]} \cdots \trans{\trSU[n]} \markMi$, 
and in this case we write
$\markM \xrightarrow{\vec{\trSU}} \markMi$.

\paragraph{Concurrent execution of transactions}

To execute transactions in parallel,
the idea is to execute them in \emph{isolation},
and then merge their changes, whenever they are disjoint.
The state updates $\mSubst$ resulting from the execution of a transaction
are formalised as in Section~\ref{sec:transactions}.

An \emph{update collector} is a function $\WR{}$ that,
given a state $\bcSt$ and a transaction $\txT$,
gives a state update $\mSubst = \WR{\bcSt,\txT}$ which maps (at least)
the updated observables to their new values.
In practice, update collectors can be obtained by instrumenting
the run-time environment of blockchains, 
to record the state updates resulting from the execution of transactions.
We formalise update collectors in Definition~\ref{def:wr}
by abstracting from the implementation details of such an instrumentation:

\begin{defi}[{Update collector}]
  \label{def:wr}
  We say that a function $\WR{}$ is an \emph{update collector} when
  $\semTx{\txT}{\bcSt} = \bcSt (\WR{\bcSt,\txT})$,
  for all $\bcSt$ and $\txT$.
\end{defi}

There exists a natural ordering of update collectors,
which extends the ordering between state updates
(\ie, set inclusion, when interpreting them as sets of substitutions):
namely, $\WR{} \sqsubseteq \WRi{}$ holds when
$\forall \bcSt, \txT : \WR{\bcSt,\txT} \subseteq \WRi{\bcSt,\txT}$.
The following lemma characterizes the least 
update collector \wrt this ordering.

\begin{lem}[{Least update collector}]
  \label{lem:wr:minimal}
  Let \mbox{$\WRmin{\bcSt,\txT} = \semTx{\txT}{\bcSt} - \bcSt$},
  where we define $\bcSti - \bcSt$ as
  $\bigcup_{{\bcSti \qmvA \neq \bcSt \qmvA}} \setenum{\bind{\qmvA}{\bcSti \qmvA}}$.
  Then, $\WRmin{}$ is the least update collector.
\end{lem}

The merge of two state updates is the union of the corresponding substitutions;
to avoid collisions, we make the merge undefined 
when the domains of the two updates overlap. 

\begin{defi}[{Merge of state updates}]
  \label{def:merge}
  Let $\mSubst[0]$, $\mSubst[1]$ be state updates.
  When $\dom{\mSubst[0]} \cap \dom{\mSubst[1]} = \emptyset$,
  we define $\mSubst[0] \mrg \mSubst[1]$ as follows:
  \[
  (\mSubst[0] \mrg \mSubst[1]) \qmvA =
  \begin{cases}
    \mSubst[0] \qmvA & \text{if $\qmvA \in \dom{\mSubst[0]}$} \\
    \mSubst[1] \qmvA & \text{if $\qmvA \in \dom{\mSubst[1]}$} \\
    \bot & \text{otherwise}
  \end{cases}
  \]
\end{defi}

The merge operator enjoys the commutative monoidal laws, and can therefore be
extended to (finite) sets of state updates.

We now associate step firing sequences with state updates.
The semantics of a step $\trSU = \setenum{(\txT[1],1),\ldots,(\txT[n],n)}$
in $\bcSt$ is obtained by applying to $\bcSt$ 
the merge of the updates $\WR{\bcSt,\txT[i]}$, 
for all $i \in 1..n$ --- whenever the merge is defined.
The semantics of a step firing sequence is then obtained by
folding the semantics of its steps.

\begin{defi}[{Semantics of step firing sequences}]
  \label{def:sem:step:seq}  
  We define the semantics of step firing sequences,
  given $\WR{}$ and $\bcSt$, as:
  \[
  \msem{\WR{}}{\bcSt}{\bcEmpty} 
  \; = \;
  \bcSt
  \hspace{30pt}
  \msem{\WR{}}{\bcSt}{\trSU \vec{\trSU}} 
  \; = \;
  \msem{\WR{}}{\bcSti}{\vec{\trSU}} 
  \quad \text{where } \bcSti = 
  \msem{\WR{}}{\bcSt}{\trSU}
  =
  \bcSt \bigoplus_{(\txT,i) \in \trSU} \WR{\bcSt,\txT}
  \]
\end{defi}

\paragraph{Concurrent execution of blockchains}

Theorem~\ref{th:bc-to-pnet} below relates serial executions of transactions 
to concurrent ones (which are rendered as step firing sequences).
Item~\ref{th:bc-to-pnet:confluence} establishes a confluence property:
if two step firing sequences lead to the same marking,
then they also lead to the same blockchain state.
Item~\ref{th:bc-to-pnet:bc} ensures that the blockchain,
interpreted as a sequence of transitions, is a step firing sequence,
and it is \emph{maximal} 
(\ie, there is a bijection between the transactions in the blockchain 
and the transitions of the corresponding net).
Finally, item~\ref{th:bc-to-pnet:maximal}
ensures that executing maximal step firing sequences 
is equivalent to executing serially the entire blockchain.

\begin{restatable}{thm}{thbctopnet}
  \label{th:bc-to-pnet}
  Let $\bcB = \txT[1] \cdots \txT[n]$. 
  Then, in $\PNet{\pswapWR}{\bcB}$:
  \begin{enumerate}[(a)]
    \vspace{-5pt}
  \item \label{th:bc-to-pnet:confluence}
    if $\markM[0] \xrightarrow{\vec{\trSU}} \markM$ and
    $\markM[0] \xrightarrow{\vec{\trSUi}} \markM$,
    then
    $\msem{\WRmin{}}{\bcSt}{\vec{\trSU}} = \msem{\WRmin{}}{\bcSt}{\vec{\trSUi}}$,
    for all reachable $\bcSt$;
    
  \item \label{th:bc-to-pnet:bc}
    $\setenum{(\txT[1],1)} \cdots \setenum{(\txT[n],n)}$
    is a maximal step firing sequence;

  \item \label{th:bc-to-pnet:maximal}
    for all maximal step firing sequences $\vec{\trSU}$,
    for all reachable $\bcSt$,
    $\msem{\WRmin{}}{\bcSt}{\vec{\trSU}} = \semBc{\bcB}{\bcSt}$.

  \end{enumerate}
\end{restatable}

  \begin{figure}[t]
  \centering
  \begin{tikzpicture}[>=stealth',scale=0.7] 
    \tikzstyle{place}=[circle,thick,draw=black!75,minimum size=5mm]
    \tikzstyle{transition}=[rectangle,thick,draw=black!75,minimum size=5mm]
    \tikzstyle{edge}=[->,thick,draw=black!75]
    \tikzstyle{edgered}=[->,thick,draw=red!75]
    \tikzstyle{edgeblu}=[->,thick,draw=blue!75]
    \node[place] (p*1) at (0,3) [label = above:{$(*,\trT[\funF])$}] {};
    \node[place] (p*2) at (11,3) [label = above:{$(*,\trT[\funH])$}] {};
    \node[place] (p*3) at (6,4.5) [label = right:{$(*,\trT[\funG])$}] {};
    \node[place] (p1*) at (2,4.5) [label = right:{$(\trT[\funF],*)$}] {};
    \node[place] (p2*) at (15,3) [label = above:{$(\trT[\funH],*)$}] {};
    \node[place] (p3*) at (8,3) [label = above:{$(\trT[\funG],*)$}] {};
    \node[place] (p13) at (4,3) [label = above:{$(\trT[\funF],\trT[\funG])$}] {};
    \node[place,token = 1] (tok1) at (0,3) {};
    \node[place,token = 1] (tok2) at (11,3) {};
    \node[place,token = 1] (tok3) at (6,4.5) {};
    \node[transition] (t1) at (2,3) {$\trT[\funF]$}
    edge[pre] node {} (p*1)
    edge[post] node {} (p1*)
    edge[post] node {} (p13);
    \node[transition] (t2) at (13,3) {$\trT[\funH]$}
    edge[pre] node {} (p*2)
    edge[post] node {} (p2*);
    \node[transition] (t3) at (6,3) {$\trT[\funG]$}
    edge[pre] node {} (p*3)
    edge[pre] node {} (p13)
    edge[post] node {} (p3*);
  \end{tikzpicture}
  \caption{Occurrence net for Example~\ref{ex:petri:1}.}
  \label{fig:petri:1}
\end{figure}

\begin{exa}[{Occurrence Net construction in Ethereum}]
  \label{ex:petri:1}
  Consider the following Ethereum transactions and functions of a contract $\cmvA$:
  \begin{align*}
    & \txT[\funF] = \ethtx{0}{\pmvA}{\cmvA}{\funF}{}
    &&  \contrFunSig{\funF}{} \, \{
       \cmdIfTE{\expLookup{\valX} = 0}{\cmdAss{\valY}{1}}{\cmdThrow} \}
    \\
    & \txT[\funG] = \ethtx{0}{\pmvA}{\cmvA}{\funG}{}
    && \contrFunSig{\funG}{} \, \{
       \cmdIfTE{\expLookup{\valY} = 0}{\cmdAss{\valX}{1}}{\cmdThrow} \}
    \\
    & \txT[\funH] = \ethtx{0}{\pmvA}{\cmvA}{\funH}{}
    && \contrFunSig{\funH}{} \, \{
       \cmdAss{\valZ}{1} \}
  \end{align*}
  Let the following sets be safe approximations of the corresponding transactions 
  \[
  \QmvA[\funF]^{w} = \QmvA[\funG]^{r} = \setenum{\cmvA.\valY},
  \QmvA[\funF]^{r} = \QmvA[\funG]^{w} = \setenum{\cmvA.\valX},
  \QmvA[\funH]^{w} = \setenum{\cmvA.\valZ},
  \QmvA[\funH]^{r} = \emptyset
  \] 
  where the subscript denotes the transaction and the superscript denotes if the set approximates the read or written keys, \eg, $\QmvA[\funF]^{w}$ safely approximates the keys written by $\txT[\funF]$, whereas $\QmvA[\funG]^{r}$ safely approximates the keys read by $\txT[\funG]$.

  By Definition~\ref{def:pswap} we have that $\txT[\funF] \pswap{}{} \txT[\funH]$ and
  $\txT[\funG] \pswap{}{} \txT[\funH]$, but $\neg(\txT[\funF] \pswap{}{} \txT[\funG])$.
  By instantiating the construction of Figure~\ref{def:bc-to-pnet} using the relation $\pswap{}{}$, we obtain the occurrence net $\PNet{\pswap{}{}}{\txT[\funF]\txT[\funH]\txT[\funG]}$ of Figure~\ref{fig:petri:1}, where $\trT[\funF] = (\txT[\funF],1)$, 
  $\trT[\funH] = (\txT[\funH],2)$, and
  $\trT[\funG] = (\txT[\funG],3)$.
  From this occurrence net is easy to see that transition $\trT[\funG]$ can only be fired after $\trT[\funF]$,
  while $\trT[\funH]$ can be fired independently from 
  $\trT[\funF]$ and $\trT[\funG]$.
  This is coherent with the fact that $\txT[\funH]$ is swappable with
  both $\txT[\funF]$ and $\txT[\funG]$, while $\txT[\funF]$ and $\txT[\funG]$
  are not swappable.  

  Recall that $\WRmin{}$ is the least update collector, \ie a function that given a state $\bcSt$ returns the minimal update $\pi$ mapping qualified keys to their new values. 
  To run in parallel the transactions, we execute them in isolation and then we merge their effect, by merging their state updates. 
  For example, given a state $\bcSt$ such that 
  $\bcSt \cmvA.x = \bcSt \cmvA.y = 0$ the minimal updates for 
  $\trT[\funF], \trT[\funG]$, and $\trT[\funH]$ are:
  \[
  \WRmin{\bcSt,\txT[\funF]} = \setenum{\bind{\cmvA.y}{1}}\qquad
  \WRmin{\bcSt,\txT[\funG]} = \setenum{\bind{\cmvA.x}{1}}\qquad
  \WRmin{\bcSt,\txT[\funH]} = \setenum{\bind{\cmvA.z}{1}}
  \]
  By Definition~\ref{def:sem:step:seq} the parallel execution  of $\trT[\funF], \trT[\funG]$, and $\trT[\funH]$ in $\bcSt$ results in the following states
  \begin{align*}
    & \msem{\WRmin{}}{\bcSt}{\setenum{\trT[\funF],\trT[\funH]}} 
    % = \bcSt(\WRmin{\bcSt,\txT[\funF]} \mrg \WRmin{\bcSt,\txT[\funF]})
      = \bcSt(\setenum{\bind{\cmvA.y}{1}} \mrg \setenum{\bind{\cmvA.z}{1}}) 
      = \bcSt\setenum{\bind{\cmvA.y}{1},\bind{\cmvA.z}{1}}
    \\
    & \msem{\WRmin{}}{\bcSt}{\setenum{\trT[\funG],\trT[\funH]}} 
      = \bcSt(\setenum{\bind{\cmvA.x}{1}} \mrg \setenum{\bind{\cmvA.z}{1}}) 
      = \bcSt\setenum{\bind{\cmvA.x}{1},\bind{\cmvA.z}{1}}
    \\
    & \msem{\WRmin{}}{\bcSt}{\setenum{\trT[\funF],\trT[\funG]}} 
      = (\bcSt\setenum{\bind{\cmvA.y}{1}} \mrg \setenum{\bind{\cmvA.x}{1}})
      = \bcSt\setenum{\bind{\cmvA.y}{1},\bind{\cmvA.x}{1}}
  \end{align*} 
  Note that, for all $\bcSt$ the serial execution of $\txT[\funF]$ and $\txT[\funH]$ 
  (in both orders) is equal to their concurrent execution
  (similarly  for $\txT[\funG]$ and $\txT[\funH]$):
  \[
  \begin{array}{c}
    \semBc{\txT[\funF]\txT[\funH]}{\bcSt} 
    \; = \; \semBc{\txT[\funH]\txT[\funF]}{\bcSt}
    \; = \; \bcSt\setenum{\bind{\cmvA.y}{1},\bind{\cmvA.z}{1}} 
    \; = \; \msem{\WRmin{}}{\bcSt}{\setenum{\trT[\funF],\trT[\funH]}} 
    \\
    \semBc{\txT[\funG]\txT[\funH]}{\bcSt} 
    \; = \; \semBc{\txT[\funH]\txT[\funG]}{\bcSt}
    \; = \; \bcSt\setenum{\bind{\cmvA.x}{1},\bind{\cmvA.z}{1}} 
    \; = \;
    \msem{\WRmin{}}{\bcSt}{\setenum{\trT[\funG],\trT[\funH]}}
  \end{array}
  \]
  Instead, for all $\bcSt$ such that $\bcSt \cmvA \valX = \bcSt \cmvA \valY = 0$ the concurrent executions of $\txT[\funF]$ and $\txT[\funG]$ may 
  differ from serial ones:
  \[
  \semBc{\txT[\funF]\txT[\funG]}{\bcSt} = 
  \bcSt\setenum{\bind{\cmvA.y}{1}}
  \qquad
  \semBc{\txT[\funG]\txT[\funF]}{\bcSt} = 
  \bcSt\setenum{\bind{\cmvA.x}{1}}
  \qquad
  \msem{\WRmin{}}{\bcSt}{\setenum{\trT[\funF],\trT[\funG]}} =
  \bcSt\setenum{\bind{\cmvA.y}{1},\bind{\cmvA.x}{1}}
  \]
  This is due the fact that $\trT[\funF]$ and $\trT[\funG]$ are \emph{not} concurrent in the occurrence net of Figure~\ref{fig:petri:1}.
  
  Now let $\vec{\trSU} = \setenum{\trT[\funF],\trT[\funH]}\setenum{\trT[\funG]}$ be a maximal step firing sequence of $\PNet{\pswap{}{}}{\bcB}$.
  Since $\trT[\funF]$ and $\trT[\funH]$ are concurrent by item~\ref{th:bc-to-pnet:maximal} of Theorem~\ref{th:bc-to-pnet} we can conclude that the semantics of $\vec{\trSU}$ in the state $\bcSt$ is equivalent to the serial one of $\bcB = \txT[\funF]\txT[\funH]\txT[\funG]$:
  \[
  \semBc{\bcB}{\bcSt} 
  \; = \; 
  \bcSt \setenum{\bind{\cmvA.y}{1}} \setenum{\bind{\cmvA.z}{1}}
  \; = \; 
  \msem{\WRmin{}}{\bcSt}{\vec{\trSU}}
  \]
  It is worth noticing that any other maximal step firing sequence of  $\PNet{\pswap{}{}}{\bcB}$ results in the same state.
  For example, consider $\vec{\trSUi} = \setenum{\trT[\funF]}\setenum{\trT[\funG],\trT[\funH]}$, 
  where the places $(\trT[\funF],*)$, $(\trT[\funG],*)$ and $(\trT[\funH],*)$
  contain one token each, while the other places have no tokens.
  Since $\vec{\trSU}$ and $\vec{\trSUi}$ lead to the same marking,
  by item~\ref{th:bc-to-pnet:confluence} of \Cref{th:bc-to-pnet}
  we conclude that 
  \[
  \msem{\WRmin{}}{\bcSt}{\vec{\trSU}} = \msem{\WRmin{}}{\bcSt}{\vec{\trSUi}}
  \]
  Now, consider $\vec{\trSUii} = \setenum{\trT[\funH]}\setenum{\trT[\funF],\trT[\funG]}$.
  % U'' = {th}{tf,tg}
  Although $\vec{\trSUii}$ is maximal, it is not a step firing sequence,
  since the second step is not enabled, therefore, no items of~\Cref{th:bc-to-pnet} apply to $\vec{\trSUii}$.
  This is coherent with the fact that $\vec{\trSUii}$ does not represent any sequential execution of $\bcB$.
\end{exa}

\section{Experimental validation}
\label{sec:validation}

In this~\namecref{sec:validation} we discuss 
how to exploit our theoretical results in practice
to improve the performance of blockchain nodes.
We start by sketching the algorithm used by miners and validators
to construct blocks.
Miners should perform the following steps:
\begin{enumerate} % [leftmargin=.15in]
\item \label{item:txpar:miners:gather}
  gather from the network a set of transactions,
  and put them in an arbitrary linear order $\bcB$, which is the mined block;
\item \label{item:txpar:miners:pswap}
  compute the relation $\pswapWR$ on $\bcB$, using a 
  static analysis of read/written observables;
\item \label{item:txpar:miners:pnet}
  construct the occurrence net $\PNet{\pswapWR}{\bcB}$;
\item \label{item:txpar:miners:exec}
  execute the transactions in $\bcB$ concurrently 
  according to the occurrence net,
  exploiting the available parallelism.
\end{enumerate}

The protocol followed by validators is almost identical to that of miners:
the main difference is that step~\ref{item:txpar:miners:gather} is skipped, 
and at step~\ref{item:txpar:miners:pswap},
the relation $\pswapWR$ is computed starting from the block $\bcB$
to be validated.
Note that the static analysis used by a validator could be different 
from the analysis used by the node which mined $\bcB$, and therefore
the occurrence net could be different from that used by the miner.
However, this is not a problem:
from item~\ref{th:bc-to-pnet:maximal} of \Cref{th:bc-to-pnet} it follows
that executing $\bcB$ on 
any occurrence nets built on any static analysis of read/written variables
leads to the same state.
In this way, blocks do not need to carry the occurrence net as metadata: 
this makes our approach is compatible with any blockchain platform,
without requiring a soft-fork.

For the case of Bitcoin, we argue that implementing this algorithm is 
straightforward: 
indeed, Lemma~\ref{lem:safeapprox:btc}
allows to compute the strong swappability relation
directly from the transactions inputs and outputs.
For Ethereum the problem is more complex, since the algorithm relies
on a static analysis of the observables read/written by transactions.
Therefore, in the rest of this~\namecref{sec:validation} we  
evaluate the feasibility of our approach on Ethereum. 
To this purpose, 
we implement a prototype analyser of Ethereum bytecode, 
and we evaluate its precision on a relevant contract.
We then compare the time of sequential executions of blocks
against their parallel executions 
(which includes the time for the static analysis).
Despite the limitations of the static analysis tool
(that we discuss at the end of the section),
we find that our technique improves the execution time in our experiment.

\paragraph{Analysing Ethereum bytecode}

In general, precise static analyses at the level of the 
Ethereum bytecode are difficult to achieve,
since the language has features like dynamic dispatching and pointer aliasing
which are notoriously a source of imprecision for static analysis.
As far as we know, none of the analysis tools for Ethereum contracts
exports an over-approximation of read/written keys 
which is usable to the purpose of this paper.
The only tool we know of that outputs such an over-approximation
is ES-ETH~\cite{Marcia19eseth}, but it has several limitations
which make its output too coarse to be usable in practice.
So, to perform an empirical validation of our approach
we develop a new prototypical tool~\cite{CLDB}. 
Our tool takes as input the EVM bytecode of a contract
and a sequence of transactions, 
and gives as output the occurrence net,
using the construction in~\Cref{sec:txpar}.
The tool implements as a standalone library~\cite{EthCA} 
a static analysis that over-approximates the read and written keys 
for each function of a given smart contract.
 
Before presenting the design underlying our static analyzer, 
we briefly recall the EVM memory model and how the bytecode generated by Solidity compiler is organized (see~\cite{EVMdoc,ethereumyellowpaper} 
for further details).
The execution of a smart contract involves three kinds of memory: 
\begin{inlinelist}
    \item the \emph{world state}, 
      \ie a mapping from addresses to account information
      (\eg, balances, functions, etc.);
    \item the \emph{contract storage}, mapping keys to values;
    \item the \emph{working memory}, \ie a stack which stores
      function parameters, local variables and temporary values 
      created during the function execution.
\end{inlinelist}
The EVM machine features instructions to load and store values 
from these memories, \eg, \code{SSTORE} and \code{SLOAD} operate 
on the world state.

The Solidity compiler splits the generated bytecode in two sections:
the \emph{constructor code} and the \emph{runtime code}. 
The constructor code is executed upon contract creation, 
and typically returns the runtime code to be deployed on the blockchain.
The runtime code is executed upon a function call. 
This code first initializes the contract storage and the stack, 
and then transfers the control to the body of the function called 
in the transaction.

Our static analysis symbolically executes both the constructor and runtime
code.
Since EVM bytecode has no explicit notion of function declaration, 
we analyze the constructor code and the first part of the runtime code 
to detect which functions are declared in the contract and 
where their code is located.
Once we identify the functions, we analyze their code separately. 
For each function we compute three sets: 
the sets of keys that are read/written by the function, 
and the set of calls made to external contracts.
To construct these sets we exploit a symbolic semantics 
of EVM instructions, that operates on abstract versions of the stack 
and memory storage. 
Intuitively, the analysis of each instruction 
results in an abstract value, 
specifying the operation performed and the affected keys.

\begin{figure}[t]
    \centering\footnotesize
    \begin{tikzpicture}[>=stealth',scale=0.6] 
        \tikzstyle{place}=[circle,thick,draw=black!75,minimum size=5mm]
        \tikzstyle{transition}=[rectangle,thick,draw=black!75,minimum size=5mm,node distance=70mm]
        \tikzstyle{edge}=[->,thick,draw=black!75]
        \tikzstyle{edgered}=[->,thick,draw=red!75]
        \tikzstyle{edgeblu}=[->,thick,draw=blue!75]

        \node[place] (new*p)  at (0, 8) [label = above:{$(*, \trT[n])$}]{};
        \node[place, token = 1] (new*pt) at (0,8) {};
        \node[transition] (new) at (0,6) {$\trT[n]$};
        \node[place] (newp*)  at (0, 4) [label = below:{$(\trT[n],*)$}]{};
        
        \node[transition] (join0) at (-6,4) {$\trT[j_0]$};
        \node[place] (constrjoin0) at (-6,6) [label = left:{$(\trT[n], \trT[j_0])$}] {};
        \node[place] (join0*p) at (-8,4) [label = left:{$(*, \trT[j_0])$}] {};
        \node[place] (join0p*) at (-4,4) [label = right:{$(\trT[j_0], *)$}] {};
        \node[place, token = 1] (join0*pt) at (-8,4) {};

        \node[transition] (join1) at (6,4) {$\trT[j_1]$};
        \node[place] (constrjoin1) at (6,6) [label = right:{$(\trT[n], \trT[j_1])$}] {};
        \node[place] (join1*p) at (8,4) [label = right:{$(*, \trT[j_1])$}] {};
        \node[place, token = 1] (join1*pt) at (8,4) {};
        \node[place] (join1p*) at (4,4) [label = left:{$(\trT[j_1], *)$}] {};

        \node[transition] (commit0) at (-6,0) {$\trT[c_0]$};
        \node[place] (joincommit0) at (-6,2) [label = left:{$(\trT[j_0], \trT[c_0])$}] {};
        \node[place] (commit0*p) at (-8,0) [label = left:{$(*, \trT[c_0])$}] {};
        \node[place] (commit0p*) at (-4,0) [label = right:{$(\trT[c_0], *)$}] {};
        \node[place, token = 1] (commit0*pt) at (-8,0) {}; 
        
        \node[transition] (commit1) at (6,0) {$\trT[c_1]$};
        \node[place] (joincommit1) at (6,2) [label = right:{$(\trT[j_1], \trT[c_1])$}] {};
        \node[place] (commit1*p) at (8,0) [label = right:{$(*, \trT[c_1])$}] {};
        \node[place] (commit1p*) at (4,0) [label = left:{$(\trT[c_1], *)$}] {};
        \node[place, token = 1] (commit1*pt) at (8,0) {};

        \node[transition] (reveal0) at (-6,-4) {$\trT[r_0]$};
        \node[place] (commitreveal0) at (-6,-2) [label = left:{$(\trT[c_0], \trT[r_0])$}] {};
        \node[place] (reveal0*p) at (-8,-4) [label = left:{$(*, \trT[r_0])$}] {};
        \node[place] (reveal0p*) at (-4,-4) [label = right:{$(\trT[r_0], *)$}] {};
        \node[place, token = 1] (reveal0*pt) at (-8,-4) {};

        \node[transition] (reveal1) at (6,-4) {$\trT[r_1]$};
        \node[place] (commitreveal1) at (6,-2) [label = right:{$(\trT[c_1], \trT[r_1])$}] {};
        \node[place] (reveal1*p) at (8,-4) [label = right:{$(*, \trT[r_1])$}] {};
        \node[place] (reveal1p*) at (4,-4) [label = left:{$(\trT[r_1], *)$}] {};
        \node[place, token = 1] (reveal1*pt) at (8,-4) {};       
        
        \node[place] (reveal1win) at (6,-6) [label = right:{$(\trT[r_1], \trT[w])$}] {};
        \node[place] (reveal0win) at (-6,-6) [label = left:{$(\trT[r_0], \trT[w])$}] {};
        \node[place] (win*p)  at (0, -4) [label = above:{$(*, \trT[w])$}]{};
        \node[place, token = 1] (win*pt) at (0,-4) {};
        \node[transition] (win) at (0,-6) {$\trT[w]$};
        \node[place] (winp*)  at (0, -8) [label = below:{$(\trT[w],*)$}]{};
        \node[place]  (commit0reveal1) at (2.5,-2.5) [label = right:{$(\trT[c_0], \trT[r_1])$}] {};
        \node[place]  (commit1reveal0) at (-2.5,-2.5) [label = left:{$(\trT[c_1], \trT[r_0])$}] {};
%        
%        \node[transition] (constr) at (6,7) {$\trT[n]$};
        
%        	
%        \node[transition] (commit0) at (3, 3) {$\trT[c_0]$};
%        \node[transition] (commit1) at (9, 3) {$\trT[c_1]$};
%        \node[transition] (reveal0) at (3, 1) {$\trT[r_0]$};
%        \node[transition] (reveal1) at (9, 1) {$\trT[r_1]$};
%        \node[transition] (win) at (6, -1) {$\trT[w]$};
%        
%        \draw[->]   (constr*p) -- (constr);
%        \draw[->]   (constr) -- (constrjoin0);
%        \draw[->]   (constr) -- (constrjoin1);
%         \draw[->]   (constrjoin1) -- (join1);
         \draw[->]   (new*p) -- (new);
         \draw[->]   (new) -- (newp*);
         \draw[->]   (new) -- (constrjoin0);
         \draw[->]   (new) -- (constrjoin1);
         
         \draw[->]   (join0*p) -- (join0);
         \draw[->]   (constrjoin0) -- (join0);
         \draw[->]   (join0) -- (join0p*);
         \draw[->]   (join0) -- (joincommit0);
         \draw[->]   (commit0*p) -- (commit0);
         \draw[->]   (joincommit0) -- (commit0);
         \draw[->]   (commit0) -- (commit0p*);
         \draw[->]   (commit0) -- (commitreveal0);
         \draw[->]   (commit0) -- (commit0reveal1);
         \draw[->]   (reveal0*p) -- (reveal0);
         \draw[->]   (commitreveal0) -- (reveal0);
         \draw[->]   (commit1reveal0) -- (reveal0);
         \draw[->]   (reveal0) -- (reveal0p*);
         \draw[->]   (reveal0) -- (reveal0win);
         \draw[->]   (reveal0win) -- (win);
         
         \draw[->]   (join1*p) -- (join1);
         \draw[->]   (constrjoin1) -- (join1);
         \draw[->]   (join1) -- (join1p*);
         \draw[->]   (commit1*p) -- (commit1);
         \draw[->]   (join1) -- (joincommit1);
         \draw[->]   (joincommit1) -- (commit1);
         \draw[->]   (commit1) -- (commit1p*);
         \draw[->]   (commit1) -- (commitreveal1);
         \draw[->]   (commit1) -- (commit1reveal0);
         \draw[->]   (commit0reveal1) -- (reveal1);
         \draw[->]   (reveal1*p) -- (reveal1);
         \draw[->]   (commitreveal1) -- (reveal1);
         \draw[->]   (reveal1) -- (reveal1p*);
         \draw[->]   (reveal1) -- (reveal1win);
         \draw[->]   (reveal1win) -- (win);
         
         \draw[->]    (win*p) -- (win);
         \draw[->]    (win) -- (winp*);
         
%        \draw[->]   (constr) -- (constrp*);
        
%        \draw[->]   (join1*p) -- (join1);
        
        %\draw[->]   (constrjoin1) -- (join1);
        %\draw[->]   (join0)  -- (commit0);
        %\draw[->]   (commit0)  -- (reveal0);
        %\draw[->]   (reveal0)  -- (win);
        
        %\draw[->]   (constr) -- (join1);
        %\draw[->]   (join1)  -- (commit1);
        %\draw[->]   (commit1)  -- (reveal1);
        %\draw[->]   (reveal1)  -- (win);
    \end{tikzpicture}
    \caption{Occurrence net for the Lottery contract.}
    \label{fig:petri:lottery}
\end{figure}

\begin{table}[t]
  \centering    
  \begin{tabular}{@{}c||ccc@{}}
    \hline
    \textbf{Sequential execution} & \textbf{Net construction} & \textbf{Parallel execution} & \textbf{Total time} \\ [0.5ex] 
    \hline
    $41.38$ ms & $4.4$ ms & $25.02$ ms & $29.42$ ms\\ 
    \hline
  \end{tabular}
  \vspace{10pt}
  \caption{Average times for executing \code{Lottery} sequentially and in parallel. 
    The total time is the sum of times for analyzing the contract bytecode, computing the occurrence net, and of running the transactions in parallel. }
  \label{tbl:experiments}
\end{table}

\paragraph{Experiments}

We experimentally validate our approach
by estimating the potential speed up achieved by running transactions in parallel.
To this purpose we consider a contract
which implements a two-players lottery
(see Listing~\ref{ex:eth:lottery} in the Appendix for its Solidity code).
Intuitively, a user who wants 
to participate in the lottery performs the following steps:
\begin{enumerate}

\item \code{join} the game by sending a certain amount of cryptocurrency, representing the bid; 

\item \code{commit} to a secret string by sending its hash,
which is stored on the contract state; 

\item once both players have completed the commit phase, 
they can \code{reveal} their secrets, independently from each other;

\item once both players have revealed, anyone can call the \code{win} function
to transfer the bets to the winner, who is determined according to the parity of the length of players' secrets.

\end{enumerate}

Once the contract has been initialized (with transaction $\trT[n]$),
a complete execution of the lottery then requires 7 transactions:
$\trT[j_0]$, $\trT[c_0]$, $\trT[r_0]$,
representing the \code{join}, \code{commit} and \code{reveal} 
of the first player, 
$\trT[j_1]$, $\trT[c_1]$, $\trT[r_1]$ for the second player,
and $\trT[w]$ for invoking the \code{win} function.
\Cref{fig:petri:lottery} displays the occurrence net computed by our tool
from a single complete execution of the lottery.
The occurrence net shows that
players can \code{join}, \code{commit} and \code{reveal}
independently from each other.
However, the \code{commit} transactions can be fired 
only after both \code{join} have been fired,
while the \code{reveal} transactions can be fired only after both \code{commit}.
Further, the \code{win} transaction can be fired only after both players 
have revealed their secrets.

To estimate the possible speed up obtained by running the transactions in parallel, we play the whole lottery 10 times, 
generating a total amount of 70 transactions (besides the contract creation).
We first run these transactions sequentially, and measure the execution time of each transaction.
Then, we use our tool to find a parallel schedule, 
and compute the time spent if the transactions were run in parallel.
   
We carry out our experiments on a laptop machine with Intel Core i5-3320M CPU @ 2.60GHz and 4Gb of RAM.%
\footnote{Our scripts and data are available online at \url{https://github.com/lillo/lmcs-analysis-validation}}
We use \code{geth}\footnote{\url{https://geth.ethereum.org/}} 
to setup a development chain, 
and Truffle\footnote{\url{https://www.trufflesuite.com/}}
to deploy a local instance of the \code{Lottery} contract on this chain.

We first compute the sequential execution time by summing up the time spent for running each transaction, 
as reported by the logs of \code{geth}. 
The first column of \Cref{tbl:experiments} displays the time of sequential execution,  averaged over 10 measurements.

Then, we analyze the sequence of transactions using our tool, 
obtaining the occurrence net.
The second column of \Cref{tbl:experiments} displays the average time spent 
by the tool to analyze the transactions and to build the occurrence net 
(again, the measurements are repeated for 10 times).
From the occurrence net, we estimate the average time required 
by the most expensive parallel schedule.
This schedule is computed as the longest and most expensive path (in terms of time) of the occurrence net.
The time required to execute this schedule is in the third column of \Cref{tbl:experiments}.
Note that estimating the cost of the parallel execution in this way
implies that we are assuming to have a sufficient number of threads 
to execute the transactions 
(in the \code{Lottery} experiment, two threads are enough),
and that once a transaction is assigned to a thread it is executed with 
no latency or queuing time.
Finally, the fourth column of \Cref{tbl:experiments} displays 
the total time required to analyze the transactions
and to run them in parallel.

Although the experiment is carried with simplifying assumptions
and on a single contract,
the results of \Cref{tbl:experiments} are a first empirical evidence 
of the practical applicability of our approach, 
and that parallelizing the execution of transaction 
may lead to performance improvements in Ethereum nodes.
We discuss below some current limitations and possible improvements 
of our experimental validation.

\paragraph{Limitations and possible improvements}

The current version of our static analysis tool of Ethereum bytecode
has been developed under some simplifying assumptions.
First, the tool can only analyse contracts whose bytecode respects 
the following conditions,
which are always satisfied for bytecode obtained by the Solidity compiler:
\begin{inlinelist}
\item the constructor code always returns the runtime code;
\item the runtime code does not access the world state in response to a call with an
  invalid function signature;
\item when a transaction calls a valid function, 
  the runtime code always transfers the control to the body of the function. 
\end{inlinelist}
While the tool could be adapted to updates of the Solidy compiler, 
pieces of bytecode not generated by the compiler may easily violate these
conditions, and it seems implausible to obtain a precise analysis
without making any assumption on the structure of bytecode.
However, this should not be an issue in practice, 
since the vast majority of transaction currently occurring in Ethereum blocks
looks like to call contracts with a verified Solidity source%
\footnote{\label{most-contracts-from-solidity}
Although we are not aware of any research explicitly quantifying 
the fraction of Ethereum transactions directed to Solidity contracts,
some empirical evidence of this conjecture can be obtained
by inspecting blocks and their transactions in \url{https://etherscan.io/},
which displays the Solidity code of target contracts.
According to~\cite{Oliva20ese}, $\sim$72\% of all transactions 
sent to contracts target contracts with verified source code.}.

A second simplification used in our tool is that the 
over-approximation of the keys read/written by a transaction
does not exploit the transaction fields 
(besides the called contract and function).
Thus, different calls to the same function but with different 
actual parameters result in the same over-approximation.
Although this simplifies the implementation, 
it may decrease the precision of the analysis, 
because the values of the function parameters are left abstract.
Consequently, the occurrence net constructed by the tool 
contains more dependencies than strictly needed.
For instance, the tool would not detect the swappable transactions 
in the ERC-721 example described in Section~\ref{sec:erc721},
since there the transaction fields are essential to obtain a precise 
over-approximation.
A possible improvement could be to refine the analysis tool
using all the transaction fields. 

Finally, the measurements we performed in our experiment are too coarse-grained
to allow a precise estimation of the speed up achieved by running 
the transactions in parallel. 
For example, we did not consider the overhead required to maintain the threads 
and to dispatch the transactions when executing the schedule given 
by the occurrence net.
To precisely measure this overhead, one would need to  
integrate our approach with an Ethereum node, 
and use it to compute the achieved speed up. 
Although preliminary, the results of our experiment 
shown in \Cref{tbl:experiments} are positive enough
to make us believe that a speed up will be confirmed 
also when taking into account these overheads.

\section{Conclusions}
\label{sec:conclusions}

We have proposed a theory of transaction parallelism for blockchains,
aimed at improving the performance of blockchain nodes.
We have started by introducing a general model of blockchain platforms, 
and we have shown how to instantiate it
to Bitcoin and Ethereum, the two most widespread blockchains.
We have defined two transactions to be \emph{swappable}
when inverting their order does not affect the blockchain state.
Since swappability is undecidable in general,
we have introduced a static approximation, called strong swappability,
based on a static analysis of the observables 
read/written by transactions.
We have rendered concurrent executions of a sequence of transactions as
\emph{step firing sequences} in the associated occurrence net.
Our main technical result, Theorem~\ref{th:bc-to-pnet}, shows that
these concurrent executions are semantically equivalent to the
sequential one.
An initial experimental assessment of our approach in Ethereum
shows that there are margins to make it applicable in practice.

We remark that our work does not address the problem of 
selecting and ordering transactions to maximize the gain of the miner,
\ie it does not proposes strategies to construct blocks of transactions
(step~\ref{item:txpar:miners:gather} in the miner algorithm
described in Section~\ref{sec:validation}).
Rather, our theory studies how to exploit the available parallelism 
to execute a block of transactions, assuming that the block is given
(which is always the case for validators).
Miners can follow different strategies to construct blocks,
driven by the economic incentives provided by the blockchain platform.
In Bitcoin, miner incentives are given by block rewards and
by the fees paid by users for each transaction included in a block.
In Ethereum, besides these incentives, 
miners can extract value directly from smart contracts 
by suitably ordering users' transactions and inserting their own.
This form of \emph{miner extractable value} has become prominent
with the emergence of DeFi contracts like 
decentralized exchanges~\cite{Daian19flash,Qin21quantifying,Zhou21discovery}.
Once a miner has formed a block of transaction according to its strategy,
our theory tells how to speed up its execution by parallelizing transactions.

In Ethereum, malevolent users could attempt 
a denial-of-service attack by bloating the blockchain with 
transactions directed to contracts which are hard to statically analyse.
This would make a na\"ive miner spend a lot of time executing 
the static analysis on these adversarial transactions.
This kind of attacks can be mitigated by miner strategies which
put a strict upper bound to the execution time of the analysis.
Note that, since most transactions in Ethereum are directed to 
a small number of well-known contracts,
like \eg ERC tokens, DeFi contracts, \etc~\cite{Oliva20ese},
to achieve an effective speed up it would be enough to 
parallelize the transactions sent to these contracts,
and execute the transactions sent to unknown contracts
without any concurrency.

Aiming at minimality, our model does not include the \emph{gas mechanism},
which is used in Ethereum to pay miners for executing contracts.
The sender of a transaction deposits into it some crypto-currency,
to be paid to the miner which appends the transaction to the blockchain.
Each instruction executed by the miner consumes part of this deposit;
when the deposit reaches zero, the miner stops executing the transaction.
At this point, all the effects of the transaction
(except the payment to the miner) are rolled back.
Our transaction model could be easily extended with a gas mechanism,
by associating a cost to statements and recording the gas consumption in the environment.
Remarkably, adding gas does not invalidate approximations of read/written keys 
which are correct while neglecting gas.
However, a gas-aware analysis may be more precise of a gas-oblivious one:
for instance, in the statement
$\cmdIfTE{\valK}{\cmdCall{\funF[long]()}{}{}; \cmdAss{\valX}{1}}{\cmdAss{\valY}{1}}$
(where $\funF[long]$ is a function which exceeds the available gas)
a gas-aware analysis would be able to detect that $\valX$ is not written.

\paragraph{Acknowledgements} 
Massimo Bartoletti is partially supported 
by Aut.\ Reg.\ Sardinia project 
\textit{``Sardcoin''}.
Letterio Galletta is partially supported by MIUR project PRIN 2017FTXR7S
\textit{``Methods and Tools for Trustworthy Smart Systems''}.
Maurizio Murgia is partially supported by MIUR PON \textit{``Distributed Ledgers for Secure Open Communities''}.

\bibliographystyle{alpha}
\bibliography{main}

\newpage
\appendix
\section{Proofs for~\Cref{sec:txswap}} \label{sec:proofs:txswap}

% \lemsimequiv*
\begin{proofof}{Lemma}{lem:sim:equiv}
  Items~\ref{lem:sim:equiv:equiv} and~\ref{lem:sim:equiv:subseteq} are trivial.
  The inclusion $\sim_{\QmvU} \,\subseteq\, \sim$ is trivial,
  and $\sim\, \subseteq\, \sim_{\QmvU}$
  follows from item~\ref{lem:sim:equiv:subseteq}.
\end{proofof}

% \lemsimappend*
\begin{proofof}{Lemma}{lem:sim:append}
  Direct from the fact that semantics of transactions is a function,
  and it only depends on the blockchain states after the execution of
  $\bcB$ and $\bcBi$, which are equal starting from any 
  blockchain state $\bcSt$, since $\bcB \sim \bcBi$.
\end{proofof}

\thswapmazurkiewicz*
\begin{proof}
  By definition, $\equivStSeq[\swap]$ is the least equivalence 
  relation closed under the rules:
  \[
  \begin{array}{c}
    \irule{}{\bcEmpty \equivStSeq[\swap] \bcEmpty} 
    \nrule{[$\equivStSeq$0]}
    \quad
    \irule{}{\txT \equivStSeq[\swap] \txT} 
    \nrule{[$\equivStSeq$1]}
    \quad
    \irule{\txT \swap \txTi}{\txT \txTi \equivStSeq[\swap] \txTi \txT} 
    \nrule{[$\equivStSeq$2]}
    \quad
    \irule{\bcB[0] \equivStSeq[\swap] \bcBi[0] \quad 
    \bcB[1] \equivStSeq[\swap] \bcBi[1]}
    {\bcB[0] \bcB[1] \equivStSeq[\swap] \bcBi[0] \bcBi[1]} 
    \nrule{[$\equivStSeq$3]}
  \end{array}
  \] 
  Let $\bcB \equivStSeq[\swap] \bcBi$. 
  We have to show $\bcB \sim \bcBi$.
  We proceed by induction on the rules above. 
  For rules \nrule{[$\equivStSeq$0]} and \nrule{[$\equivStSeq$1]},
  the thesis follows by reflexivity, 
  since $\sim$ is an equivalence relation (Lemma~\ref{lem:sim:equiv}).
  For rule \nrule{[$\equivStSeq$2]}, the thesis follows immediately 
  by Definition~\ref{def:swap}.
  For rule \nrule{[$\equivStSeq$3]}, first note that $\bcB = \bcB[0] \bcB[1]$ and
  $\bcBi = \bcBi[0] \bcBi[1]$. By the induction hypothesis it follows that:
  \[\bcB[0] \sim \bcBi[0] \quad \text{and} \quad \bcB[1] \sim \bcBi[1]\]
  Therefore, by two applications of Lemma~\ref{lem:sim:append}:
  \[
  \bcB = \bcB[0] \bcB[1] \sim \bcB[0] \bcBi[1] \sim \bcBi[0] \bcBi[1] = \bcBi
  \tag*{\qedhere}
  \]
\end{proof}

\begin{proofof}{Lemma}{lem:safeapprox}
  \Cref{lem:safeapprox:subseteq}.
  For the case $\bullet = w$, let $\wapprox{\QmvA}{\txT}$ and 
  $\QmvA \subseteq \QmvAi$. 
  Let $\QmvB$ be such that $\QmvB \cap \QmvAi = \emptyset$.
  We have to show that $\txT \sim_{\QmvB} \bcEmpty$.
  Since $\QmvA \subseteq \QmvAi$,
  it must be $\QmvB \cap \QmvA = \emptyset$. 
  Then, since $\wapprox{\QmvA}{\txT}$,
  it must be $\txT \sim_{\QmvB} \bcEmpty$, as required.
  For the case $\bullet = r$, let $\rapprox{\QmvA}{\txT}$ and 
  $\QmvA \subseteq \QmvAi$. 
  We have to show that, for all $\bcB[1],\bcB[2]$,
  if $\bcB[1] \sim_{\QmvAi} \bcB[2]$ and $\bcB[1] \sim_{\QmvB} \bcB[2]$,
  then $\bcB[1]\txT \sim_{\QmvB} \bcB[2]\txT$. But this follows immediately
  by the fact that $\QmvA \subseteq \QmvAi$ and $\rapprox{\QmvA}{\txT}$.
  
  \medskip\noindent
  \Cref{lem:safeapprox:cap}. 
  Let $\QmvC$ be such that $\QmvC \cap (\QmvA \cap \QmvB) = \emptyset$.
  % Note that $\QmvC \cap (\QmvA \cap \QmvB) = (\QmvC \cap \QmvB) \cap \QmvA$.
  Since 
  $\safeapprox[w]{\QmvA}{\txT}$ and $(\QmvC \setminus \QmvA) \cap \QmvA =
  \emptyset$, then:
  \[
  \txT \sim_{\QmvC \setminus \QmvA} \bcEmpty
  \]
  Similarly, since $\safeapprox[w]{\QmvB}{\txT}$ and 
  $(\QmvC \setminus \QmvB) \cap \QmvB = \emptyset$, we have that:
  \[
  \txT \sim_{\QmvC \setminus \QmvB} \bcEmpty
  \]
  By assumption $\QmvC \cap (\QmvA \cap \QmvB) = \emptyset$,
  then $(\QmvC \setminus \QmvA) \cup (\QmvC \setminus \QmvB) = \QmvC$.
  By Definition~\ref{def:sim}, we conclude:
  \[
  \txT 
  \; \sim_{\QmvC} \;
  \txT 
  \; \sim_{(\QmvC \setminus \QmvA) \cup (\QmvC \setminus \QmvB)} \;
  \bcEmpty
  \tag*{\qedhere}
  \]
\end{proofof}

\begin{lem}
  \label{lem:equivRel-implies-equiv}
  \(
  \bcB \equivStSeq[\pswap{}{}] \bcBi
  \;\;\implies\;\;
  \bcB \sim \bcBi
  \)
\end{lem}
\begin{proof}
  Direct by Theorems~\ref{th:swap:mazurkiewicz} and~\ref{th:pswap:mazurkiewicz}.
\end{proof}

\section{Proofs for~\Cref{sec:txpar}} \label{sec:proofs:txpar}

\begin{lem}
  \label{lem:merge-comm-monoid}
  $\mrg$ is commutative and associative, 
  with $\lambda \qmvA. \bot$ as neutral element.
\end{lem}
\begin{proof}
  Trivial.
\end{proof}

\begin{lem}
  \label{lem:mrg-seq}
  If $\mSubst[1] \mrg \mSubst[2] = \mSubst$, then 
  $\mSubst = \mSubst[1]\mSubst[2]$.
\end{lem}
\begin{proof}
  Since $\mSubst[1] \mrg \mSubst[2]$ is defined, it must be
  $\keys{\mSubst[1]} \cap \keys{\mSubst[2]} = \emptyset$.
  Let $\qmvA$ be an observable. We have two cases:
  \begin{itemize}
  \item $\qmvA \in \keys{\mSubst}$. Since $\keys{\mSubst} = \keys{\mSubst[1]} 
    \cup \keys{\mSubst[2]}$, we have two subcases:
    \begin{itemize}
    \item $\qmvA \in \keys{\mSubst[1]}$. Then, $\mSubst \qmvA = \mSubst[1]\qmvA$.
      By disjointness, $\qmvA \not\in \keys{\mSubst[2]}$, and hence
      $\mSubst[1]\mSubst[2]\qmvA = \mSubst[1]\qmvA$.
    \item $\qmvA \in \keys{\mSubst[2]}$. Then, $\mSubst \qmvA = \mSubst[2]\qmvA = 
      \mSubst[1]\mSubst[2]\qmvA$.
    \end{itemize}
  \item $\qmvA \not\in \keys{\mSubst}$. 
    Then, $\qmvA \not\in \keys{\mSubst[1]}$,
    $\qmvA \not\in \keys{\mSubst[2]}$, 
    and so $\mSubst \qmvA = \bot = \mSubst[1]\mSubst[2]\qmvA$.
    \qedhere
  \end{itemize}
\end{proof}

\begin{lem}
  If $\bcB[1] \seqn \txsA[1]$ and $\bcB[2] \seqn \txsA[2]$, then 
  $\bcB[1]\bcB[2] \seqn (\txsA[1] \cup \txsA[2])$.
\end{lem}
\begin{proof}
  By induction on $\card{\bcB[2]}$. For the base case, it must be 
  $\bcB[2] = \bcEmpty$ and hence $\txsA[2] = \emptyset$. Then,
  $\bcB[1]\bcB[2] = \bcB[1]$ and $\txsA[1] \cup \txsA[2] = \txsA[1]$.
  Therefore, the thesis coincides with the first hypothesis.
  For the induction case, it must be $\bcB[2] = \bcBi[2]\txT$, with
  $\card{\bcBi[2]} = n$. Furthermore, it must be
  $\txsA[2] = \setenum{\txT} \cup \txsAi[2]$, for some $\txsAi[2]$ such that
  $\bcBi[2] \seqn \txsAi[2]$. By the induction hypothesis:
  \[
  \bcB[1]\bcBi[2] \seqn (\txsA[1] \cup \txsAi[2])
  \]
  Then:
  \[
  \bcB[1]\bcBi[2]\txT = \bcB[1]\bcB[2] \seqn 
  (\setenum{\txT} \cup \txsA[1] \cup \txsAi[2]) = \txsA[1] \cup \txsA[2]
  \tag*{\qedhere}
  \]
\end{proof}

\begin{lem}
  \label{lem:strong-swap-chain-W-R}
  Let $\bcB$ and $\txT$ be such that 
  $\bcB = \bcB[1]\txTi\bcB[2] \implies \txT \pswapWR \txTi$. Then, for all
  $\bcBi$, $\semBc{\bcBi}{} \sim_{\rset{\txT}}\semBc{\bcB}{\semBc{\bcBi}{}}$
  and $\semBc{\bcBi}{} \sim_{\wset{\txT}}\semBc{\bcB}{\semBc{\bcBi}{}}$.
\end{lem}
\begin{proof}
  A simple induction on $\card{\bcB}$, using Definition \ref{def:safeapprox} 
  for the induction case.
\end{proof}

We now formalize when a blockchain $\bcB$ is a
serialization of a multiset of transactions $\txsA$.

\begin{defi}[{Serialization of multisets of transactions}]
  \label{def:mset-serial}
  We define the relation $\seqn$ between blockchains and multisets 
  of transactions as follows:
  \[
  \irule{}
  {\bcEmpty \seqn \emptymset}
  \qquad\qquad
  \irule{\bcB \seqn \txsA}
  {\bcB\txT \seqn (\msetenum{\txT} + \txsA)}
  \]
\end{defi}

\begin{lem}
  \label{lem:strong-swap-seq}
  If $\txT \pswapWR \txTi$ for all $\txTi \in \txsA$ and $\bcB \seqn \txsA$
  then, $\bcB = \bcB[1]\txTi\bcB[2] \implies \txT \pswapWR \txTi$.
\end{lem}
\begin{proof}
  By a simple induction on $\card{\txsA}$ we can conclude that, whenever $\bcB$ is of the form $\bcB[1]\txTi\bcB[2]$ for some $\bcB[1],\bcB[2]$ and 
  $\txTi$, we have that $\txTi \in \txsA$. 
  The thesis then follows immediately.
\end{proof}

\begin{lem}
  \label{lem:pi-swap}
  If $\rapprox{\qmvA}{\txT}$, $\bcB[1] \sim_{\qmvA} \bcB[2]$ and 
  $\bcB[1] \sim_{\qmvB} \bcB[2]$, then $\semBc{\bcB[1]\txT}{} \sim_{\qmvB} 
  \semBc{\bcB[1]}{}\WR{\semBc{\bcB[2]}{},\txT}$.
\end{lem}
\begin{proof}
  Let $\mSubst[1] = \WR{\semBc{\bcB[1]}{},\txT}$ and 
  $\mSubst[2] = \WR{\semBc{\bcB[2]}{},\txT}$. By Definition \ref{def:wr}, 
  $\semBc{\bcB[1]\txT}{} = \semBc{\bcB[1]}{}\mSubst[1]$. Let $\qmvA \in \QmvB$.
  We have two cases:
  \begin{itemize}
  \item $\qmvA \in \keys{\mSubst[2]}$. 
    \[
    \semBc{\bcB[1]}{}\mSubst[2]\qmvA = \mSubst[2]\qmvA = 
    \semBc{\bcB[2]}{}\mSubst[2]\qmvA = \semBc{\bcB[2]\txT}{}\qmvA = 
    \semBc{\bcB[1]\txT}{}\qmvA
    \]
  \item $\qmvA \not\in \keys{\mSubst[2]}$.
    \[
    \semBc{\bcB[1]}{}\mSubst[2]\qmvA = \semBc{\bcB[1]}{}\qmvA = 
    \semBc{\bcB[2]}{}\qmvA = \semBc{\bcB[2]}{}\mSubst[2]\qmvA = 
    \semBc{\bcB[1]\txT}{}\qmvA
    \tag*{\qedhere}
    \]
  \end{itemize}
\end{proof}

\begin{defi}
  Let $\WR{}$ be a state updater, and
  let $\wset{}$ be such that $\forall \txT: \wapprox{\wset{\txT}}{\txT}$.
  We say that $\WR{}$ and $\wset{}$ are \emph{compatible} when
  $\forall \bcSt,\txT : \keys{\WR{\bcSt,\txT}} \subseteq \wset{\txT}$.
\end{defi}

We extend the semantics of transactions
to finite \emph{multisets} of transactions.
Hereafter, we denote with $\emptymset$ the empty multiset,
with $\msetenum{\txT[1],\ldots,\txT[n]}$ the multiset containing
$\txT[1],\ldots,\txT[n]$, and with
$A + B$ the sum between multisets,
\ie $(A+B)(x) = A(x) + B(x)$ for all $x$.

\begin{defi}[{Semantics of multisets of transactions}]
  \label{def:multiset-sem}
  We denote the semantics of a multiset of transactions $\txsA$,
  in a state $\bcSt$ and an update collector $\WR{}$,
  as $\msem{\WR{}}{\bcSt}{\txsA}$,
  where the partial function
  $\msem{\WR{}}{\bcSt}{\cdot}$ is defined as:
  \( \;
  \msem{\WR{}}{\bcSt}{\txsA}
  \, = \,
  \bcSt \bigoplus_{\txT \in \txsA} \WR{\bcSt,\txT}
  \).
\end{defi}

\noindent
Hereafter, we say that a multiset $\txsA$ is strongly swappable 
\wrt a relation $\relR \subseteq \pswap{}{}$ when:
\[
\forall \txT \in \txsA, \forall \txTi \in \txsA - \msetenum{\txT} : \;
\txT \,\relR\, \txTi
\]

\begin{lem}
  \label{lem:step-seq-union-aux}
  Let $\txsA$ be strongly swappable \wrt $\pswapWR$, 
  let $\bcB \seqn \txsA$, 
  and let $\WR{}$ be compatible with $\wset{}$. 
  Then, for all $\bcB[0]$: $\msem{\WR{}}{\semBc{\bcB[0]}{}}{\txsA} = 
  \semBc{\bcB}{\semBc{\bcB[0]}{}}$.
\end{lem}
\begin{proof}
  By induction on $\card{\bcB}$. For the base case, it must be 
  $\bcB = \bcEmpty$ and $\txsA = \emptyset$, and hence 
  $\msem{\WR{}}{\semBc{\bcB[0]}{}}{\emptyset} = \semBc{\bcB[0]}{} = 
  \semBc{\bcEmpty}{\semBc{\bcB[0]}{}}$.
  For the induction case, it must be $\bcB = \bcBi\txT$, with 
  $\card{\bcBi} = n$. 
  Clearly, 
  $\txsA = \msetenum{\txT} + \txsAi$ for some $\txsAi$ such that
  $\bcBi \seqn \txsAi$. 
  Let $\WR{\semBc{\bcB[0]}{},\txT} = \mSubst[\txT]$.
  By the induction hypothesis:
  \begin{equation}
    \label{eq:step-seq-union-3}
    \msem{\WR{}}{\semBc{\bcB[0]}{}}{\txsAi} = \semBc{\bcBi}{\semBc{\bcB[0]}{}}
  \end{equation}
  Notice that:
  \begin{equation}
    \label{eq:step-seq-union-2}
    \msem{\WR{}}{\semBc{\bcB[0]}{}}{\txsAi} = \semBc{\bcB[0]}{}\mSubsti
  \end{equation} 
  where $\mSubsti \bigoplus_{\txTi \in \txsAi} \WR{\semBc{\bcB[0]}{},\txTi})$.
  Let $\WR{\semBc{\bcB[0]}{},\txT} = \mSubst[\txT]$.
  Since $\txsA$ is strongly swappable \wrt $\pswapWR$ and $\WR{}$
  is compatible with $\wset{}$, it must be $\keys{\mSubsti}\cap 
  \keys{\mSubst[\txT]} = \emptyset$, and hence $(\mSubsti\mrg \mSubst[\txT])$ 
  is defined. 
  Then, it must be:
  \begin{align}
    \nonumber
    \msem{\WR{}}{\semBc{\bcB[0]}{}}{\txsA} 
    & = \semBc{\bcB[0]}{}(\mSubsti\mrg \mSubst[\txT]) 
    \\
    \nonumber
    & = \semBc{\bcB[0]}{}\mSubsti \mSubst[\txT] 
    && \text{By Lemma \ref{lem:mrg-seq}}
    \\
    \nonumber
    & = \msem{\WR{}}{\semBc{\bcB[0]}{}}{\txsAi}\mSubst[\txT] 
    && \text{By \Cref{eq:step-seq-union-2}}
    \\
    \label{eq:step-seq-union-1}
    & =  \semBc{\bcBi}{\semBc{\bcB[0]}{}}\mSubst[\txT] 
    && \text{By \Cref{eq:step-seq-union-3}}
  \end{align}
  We have that:
  \[
  \semBc{\bcB}{\semBc{\bcB[0]}{}} = \semBc{\bcBi\txT}{\semBc{\bcB[0]}{}} =
  \semBc{\bcBi}{\semBc{\bcB[0]}{}}\mSubsti[\txT]
  \]
  where $\mSubsti[\txT] = \WR{\semBc{\bcBi[0]}{\semBc{\bcB[0]}{}},\txT}$.
  Since $\keys{\mSubst[\txT]} \subseteq \wset{\txT}$ and 
  $\keys{\mSubst[\txT]} \subseteq \rset{\txT}$, it follows
  immediately that $\semBc{\bcBi}{\semBc{\bcB[0]}{}}\mSubst[\txT] \sim_{\qmvA}
  \semBc{\bcB}{\semBc{\bcB[0]}{}}$ for all $\qmvA \not\in \wset{\txT}$.
  It remains to show that $\semBc{\bcBi}{\semBc{\bcB[0]}{}}\mSubst[\txT] 
  \sim_{\wset{\txT}} \semBc{\bcB}{\semBc{\bcB[0]}{}}$.
  First notice that, by Lemmas \ref{lem:strong-swap-seq} and \ref{lem:strong-swap-chain-W-R}:
  \[
  \semBc{\bcB[0]}{} \sim_{\rset{\txT}} \semBc{\bcBi}{\semBc{\bcB[0]}{}} \qquad
  \semBc{\bcB[0]}{}\sim_{\wset{\txT}} \semBc{\bcBi}{\semBc{\bcB[0]}{}}
  \]
  Then, by Lemma \ref{lem:pi-swap}:
  \[
  \semBc{\bcBi}{\semBc{\bcB[0]}{}}\mSubst[\txT] 
  \sim_{\wset{\txT}} \semBc{\bcB}{\semBc{\bcB[0]}{}}
  \]
  And hence:
  \begin{equation}
    \label{eq:step-seq-union-4}
    \semBc{\bcBi}{\semBc{\bcB[0]}{}}\mSubst[\txT] = \semBc{\bcB}{\semBc{\bcB[0]}{}}
  \end{equation}
  The thesis $\msem{\WR{}}{\semBc{\bcB[0]}{}}{\txsA} = 
  \semBc{\bcB}{\semBc{\bcB[0]}{}}$ then follows by 
  \Cref{eq:step-seq-union-4,eq:step-seq-union-1}.
\end{proof}

The following~\namecref{th:step-seq-union} 
ensures that the parallel execution of strongly swappable transactions 
is equivalent to any sequential execution of them.

\begin{restatable}{thm}{thstepsequnion}
  \label{th:step-seq-union}
  Let $\txsA$ be strongly swappable \wrt $\pswapWR$, and
  let $\bcB \seqn \txsA$. 
  Then, for all $\bcSt$: 
  \[
  \msem{\WRmin{}}{\bcSt}{\txsA} 
  \; = \;
  \semBc{\bcB}{\bcSt}
  \]
\end{restatable}
\begin{proof}
  Direct by Lemma~\ref{lem:step-seq-union-aux} and 
  by the fact that every $\wset{}$ is compatible with $\WRmin{}{}$.
\end{proof}

A \emph{parellelized blockchain} $\bcsA$ 
is a finite sequence of multisets of transactions;
we denote with $\bcEmpty$ the empty sequence.
We extend the semantics of multisets (\Cref{def:multiset-sem}) to 
parallelized blockchains as follows.

\begin{defi}[{Semantics of parallelized blockchains}]
  \label{def:txpar:pbc-semantics}
  The semantics of parallelized blockchains is defined as follows:
  \[
  \msem{\WR{}}{\bcSt}{\bcEmpty} = \bcSt
  \qquad
  \msem{\WR{}}{\bcSt}{\txsA \bcsA} =
  \msem{\WR{}}{\scriptsize \msem{\WR{}}{\bcSt}{\txsA}}{\bcsA}
  \]
  We write $\msem{\WR{}}{}{\bcsA}$ for $\msem{\WR{}}{\bcStInit}{\bcsA}$,
  where $\bcStInit$ is the initial state.
\end{defi}

We also extend the serialization relation $\seqn$ (Definition \ref{def:mset-serial}) to
parallelized blockchains.

\begin{defi}[{Serialization of parallelized blockchains}]
  We define the relation $\seqn$ between blockchains and 
  parallelized blockchains as follows:
  \[
  \irule{}{\bcEmpty \seqn \bcEmpty} 
  \qquad\qquad
  \irule
  {\bcB[1] \seqn \txsA \quad \bcB[2] \seqn \bcsA}
  {\bcB[1]\bcB[2] \seqn \txsA\bcsA}
  \]
\end{defi}

The following~\namecref{th:seq-union}
states that our technique to parallelize the transactions in a blockchain
preserves its semantics.

\begin{restatable}{thm}{thsequnion}
  \label{th:seq-union}
  Let each multiset in $\bcsA$ be strongly swappable \wrt $\pswapWR$, and 
  let $\bcB \seqn \bcsA$. 
  Then, for all $\bcSt$: 
  \[
  \msem{\WRmin{}}{\bcSt}{\bcsA} 
  \; = \;
  \semBc{\bcB}{\bcSt}
  \]
\end{restatable}
\begin{proof}
  By induction on the rule used for deriving $\bcB \seqn \bcsA$.
  \begin{itemize}

  \item Rule: $\irule{}{\bcEmpty \seqn \bcEmpty}$.
    \\[5pt]
    The thesis follows trivially, since
    \(
    \semBc{\bcEmpty}{\bcSt}= \bcSt = \msem{\WRmin{}}{\bcSt}{\bcEmpty} 
    \).

  \item Rule: $\irule{\bcB[1] \seqn \txsA \quad \bcB[2] \seqn \bcsA}
    {\bcB[1]\bcB[2] \seqn \txsA\bcsA}$.
    \\[5pt] 
    By Theorem \ref{th:step-seq-union}, for some reachable $\bcSti$ 
    it must be
    \(
    \semBc{\bcB[1]}{\bcSt}= \bcSti = \msem{\WRmin{}}{\bcSt}{\txsA} 
    \).
    By the induction hypothesis,
    \(
    \semBc{\bcB[2]}{\bcSti} = \msem{\WRmin{}}{\bcSti}{\bcsA} 
    \).
    The thesis then follows by:
    \[
    \semBc{\bcB[2]}{\bcSti} = \semBc{\bcB[1]\bcB[2]}{\bcSt} 
    \qquad
    \msem{\WRmin{}}{\bcSti}{\bcsA} = \msem{\WRmin{}}{\bcSt}{\txsA\bcsA}
    \qedhere
    \]
  \end{itemize}
\end{proof}

\begin{lem}
  \label{lem:le-star-po}
  Let 
  $\PNet{\relR}{\txT[1] \cdots \txT[n]} = (\Places, \Transitions, \Arcs, \markM[0])$.
  Then $(\Transitions,<^*)$ is a partial order.
\end{lem}
\begin{proof}
  Transitivity and reflexivity hold by definition. 
  For antisymmetricity, assume that 
  $(\txT[i],i) <^* (\txT[j],j)$ and $(\txT[j],j) <^* (\txT[i],i)$.
  Then, it is easy to verify that $i \leq j$ and $j \leq i$, and so $i = j$.
  Since $\txT[i]$ and $\txT[j]$ are uniquely determined by $i$ and $j$, we have
  that $\txT[i] = \txT[j]$. 
  Therefore, $(\txT[i],i) = (\txT[j],j)$, as required.
\end{proof}

\lembctopnetonet*
\begin{proof}
  By Definition \ref{def:bc-to-pnet}, the first three conditions
  of the definition of occurrence net are easy to verify.
  To prove that $\Arcs^*$ is acyclic, we proceed by contradiction. 
  Assume that there is a sequence 
  $\vec x = x_0,x_1,\hdots x_m$ such that $(x_{i},x_{i+1}) \in \Arcs$ 
  for all $0 \leq i < m$, and $x_0 = x_m$ with $m > 0$. 
  Notice that the above sequence alternates between 
  transitions and places, and so, since $m > 0$, 
  at least one place and one transition occur in $\vec x$. 
  Further, a place between two transitions 
  $\trT \neq \trTi$ can exist only if $\trT < \trTi$. 
  Therefore, if $\trT,\trTi$ occur in $\vec x$, 
  it must be $\trT <^* \trTi$ and $\trTi <^* \trT$.
  So, if $\vec x$ contains at least two transitions, by Lemma \ref{lem:le-star-po}, 
  we have a contradiction. If only one transition $\trT = (\txT,i)$ occurs in $\vec x$, 
  then there is a place of the form $(\trT,\trT)$ occuring in $\vec x$. 
  Therefore, $\trT < \trT$, which implies $i < i$ --- contradiction.
\end{proof}

\begin{lem}
  \label{lem:ON-disjoint-pre-trans}
  Let $\netN = (\Places, \Transitions, \Arcs, \markM[0])$ be an occurrence net.
  For all $\trT,\trTi \in \Transitions$, if $\trT \neq \trTi$ then
  $\pre{\trT} \cap \pre{\trTi} = \emptyset$.
\end{lem}
\begin{proof}
  By contradiction, assume that $\placeP \in \pre{\trT} \cap \pre{\trTi}$
  with $\trT \neq \trTi$.
  Then, $\setenum{\trT,\trTi} \subseteq \post{\placeP}$, and hence
  $\card{\post{\placeP}} \geq 2$ ---
  contradiction with constraint~\ref{item:petri:onet:p}
  of the definition of occurrence nets.
\end{proof}

\begin{lem}
  \label{lem:occurrenceNet-LTS}
  Let $\markM$ be a reachable marking of an occurrence net $\netN$.
  Then:
  \begin{enumerate}
  \item     \label{lem:occurrenceNet-LTS:item-determinism}
    If $\markM \trans{\trT} \markMi$ and $\markM \trans{\trT} \markMii$,
    then $\markMi = \markMii$ (determinism).
  \item \label{lem:occurrenceNet-LTS:item-diamond}
    If $\markM \trans{\trT} \markMi$, $\markM \trans{\trTi} \markMii$
    and $\trT \neq \trTi$, then there exists $\markMiii$ such that 
    $\markMi \trans{\trTi} \markMiii$ and $\markMii \trans{\trT} \markMiii$ 
    (diamond property).
  \item \label{lem:occurrenceNet-LTS:item-distinct}
    If $\markM \trans{\trT}\;\trans{}^*\,\trans{\trTi}$ 
    then $\trT \neq \trTi$
    (linearity).
  \item \label{lem:occurrenceNet-LTS:item-acyclic}
    If $\markM \trans{\vec{\trT}}{\!\!}^* \;\markM$ then $\card{\vec{\trT}} = 0$
    (acyclicity).
  \end{enumerate}
\end{lem}
\begin{proof}
  For item \ref{lem:occurrenceNet-LTS:item-determinism},
  by definition of the firing of transitions of Petri Nets it must be
  \(
  \markMi = \markM - \pre{\trT} + \post{\trT} = \markMii
  \).

  \smallskip\noindent
  For item \ref{lem:occurrenceNet-LTS:item-diamond},
  since $\markM \trans{\trT} \markMi$ and $\markM \trans{\trTi} \markMii$,
  it must be:
  \begin{align*}
    & \pre{\trT} \subseteq \markM
    && \markMi = \markM - \pre{\trT} + \post{\trT}
    \\
    & \pre{\trTi} \subseteq \markM
    && \markMii = \markM - \pre{\trTi} + \post{\trTi}
  \end{align*}
  By Lemma~\ref{lem:ON-disjoint-pre-trans},
  $\trTi$ is enabled at $\markMi$, and $\trT$ is enabled at $\markMii$.
  Then, by definition of firing:
  \[
  \markMi \trans{\trTi} \markMi - \pre{\trTi} + \post{\trTi}\qquad \text{and} 
  \qquad \markMii \trans{\trT}\markMii - \pre{\trT} + \post{\trT}
  \]
  Then:
  \begin{align*}
    \markMi - \pre{\trTi} + \post{\trTi} 
    & = (\markM - \pre{\trT} + \post{\trT}) - \pre{\trTi} + \post{\trTi} \\
    & = (\markM - \pre{\trTi} + \post{\trTi}) - \pre{\trT} + \post{\trT} 
    && \text{(as $\pre{\trTi} \subseteq \markM$)} \\
    & = \markMii  - \pre{\trT} + \post{\trT}
  \end{align*}
  Hence, the thesis follows by choosing 
  $\markMiii = \markMi - \pre{\trTi} + \post{\trTi}$.

  \smallskip\noindent
  Item \ref{lem:occurrenceNet-LTS:item-distinct} follows directly
  by induction on the length of the reduction $\trans{}^*$,
  exploiting the fact that $\Arcs^*$ is a partial order.

  \smallskip\noindent
  Item \ref{lem:occurrenceNet-LTS:item-acyclic} follows by the fact that
  $\Arcs^*$ is a partial order.
\end{proof}

\begin{lem}
  \label{lem:occurrenceNet-reach}
  Let $\netN = (\Places, \Transitions, \Arcs, \markM[0])$ be an occurrence net,
  and let $\markM$ be a reachable marking, such that, 
  for some $\trT$, $\markMi$, $\markMii$:
  \begin{center}
    \begin{tikzcd}[column sep=normal, row sep=normal, negated/.style={
        decoration={markings,
          mark= at position 0.5 with {
            \node[transform shape] (tempnode) {$\slash$};
          }
        },
        postaction={decorate}
      }]
      \markM
      \ar[r]
      \ar[dr, start anchor=south east, end anchor=north west,"{\trT}", sloped, ""' near end]
      & {\!\!}^{n + 1} \;\; \markMi \ar[r, sloped, negated, "\trT"]
      & {} \\
      & \;\; \markMii
      & {}
    \end{tikzcd}
  \end{center}
  Then, $\markMii \trans{}^n \markMi$.
\end{lem}
\begin{proof}
  By induction on $n$. 
  For the base case, it must be $\markM \trans{}^1 \markMi$,
  and hence $\markM \trans{\trTi} \markMi$ for some $\trTi$. Since $\markM 
  \trans{\trT} \markMii$, by the contrapositive of 
  item \ref{lem:occurrenceNet-LTS:item-diamond} of Lemma \ref{lem:occurrenceNet-LTS} 
  (diamond property) it follows $\trT = \trTi$. 
  So, by item \ref{lem:occurrenceNet-LTS:item-determinism} 
  of Lemma~\ref{lem:occurrenceNet-LTS}
  (determinism) we have that $\markMi = \markMii$. Clearly: 
  \[
  \markMii \trans{}^0 \markMi
  \]
  For the induction case, let $n = m + 1$, for some $m$.
  Then, for some $\trTi,\markMiii$: 
  \[
  \markM \trans{\trTi} \markMiii \trans{}^{m + 1} \markMi
  \] 
  If $\trT = \trTi$, then by 
  item~\ref{lem:occurrenceNet-LTS:item-determinism} 
  of Lemma \ref{lem:occurrenceNet-LTS} (determinism)
  it follows that $\markMiii = \markMii$,
  and so we have the thesis $\markMii \trans{}^{n} \markMi$.
  Otherwise, if $\trT \neq \trTi$,
  by item~\ref{lem:occurrenceNet-LTS:item-diamond}
  of Lemma~\ref{lem:occurrenceNet-LTS} (diamond property), 
  there must exists
  $\markM[1]$ such that:
  \[
  \markMii \trans{\trTi} \markM[1] \;\;\text{and}\;\; 
  \markMiii \trans{\trT} \markM[1]
  \] 
  We are in the following situation:
  \begin{center}
    \begin{tikzcd}[column sep=normal, row sep=normal, negated/.style={
        decoration={markings,
          mark= at position 0.5 with {
            \node[transform shape] (tempnode) {$\slash$};
          }
        },
        postaction={decorate}
      }]
      \markMiii
      \ar[r]
      \ar[dr, start anchor=south east, end anchor=north west,"{\trT}", sloped, ""' near end]
      & {\!\!}^{m + 1} \;\; \markMi \ar[r, sloped, negated, "\trT"]
      & {} \\
      & \;\; \markM[1]
      & {}
    \end{tikzcd}
  \end{center}
  Since $m + 1 = n$, by the induction hypothesis:
  \[
  \markM[1] \trans{}^m \markMi
  \]
  Therefore, we have the thesis:
  \[
  \markMii \trans{\trTi} \markM[1] \trans{}^m \markMi
  \qedhere
  \]
\end{proof}

\begin{lem}
  \label{lem:PNet-leq-implies-not-independent}
  Let $(\txT,i),(\txTi,j)$ be transitions of
  $\PNet{\relR}{\bcB}$.
  If $\markM$ is a reachable marking, then:
  \[
  (\txT,i) < (\txTi,j)
  \;\;\text{and}\;\;
  \markM \trans{(\txT,i)}
  \quad\implies\quad
  \markM \nottrans{(\txTi,j)}
  \]
\end{lem}
\begin{proof}
  By the construction in Figure~\ref{def:bc-to-pnet},
  since $(\txT,i) < (\txTi,j)$, then 
  $\placeP = ((\txT,i),(\txTi,j))$ is a place of the occurrence net,
  and 
  $\Arcs((\txT,i),\placeP) = 1$ and $\Arcs(\placeP,(\txTi,j)) = 1$.
\end{proof}

\begin{defi}[{Independency}]
  \label{def:independent}
  Let $\netN$ be an occurrence net.
  We say that two transitions $\trT$ and $\trTi$ are \emph{independent}, 
  in symbols $\trT \independent \trTi$, 
  if $\trT \neq \trTi$ and there exists a reachable marking $\markM$ such that:
  \[
  \markM \trans{\trT} \;\;\text{and}\;\; \markM \trans{\trTi}
  \]
  We define $\equivSeq$ 
  as the least congruence in the free monoid $\Transitions^*$ such that,
  for all $\trT, \trTi \in \Transitions$:
  \(\;
  \trT \, I \, \trTi \implies \trT \trTi \equivSeq \trTi \trT
  \).
\end{defi}

\begin{lem}
  \label{lem:PN-mstep-indepenent}
  Let $\netN$ be an occurrence net, with a reachable marking $\markM$. 
  If $\markM \trans{\trSU}$ then $\trT \independent \trTi$,
  for all $\trT \neq \trTi \in \trSU$.
\end{lem}
\begin{proof}
  Since $\markM \trans{\trSU}$, then $\markM \trans{\trT}$ for all
  $\trT \in \trSU$.
\end{proof}

\begin{lem}
  \label{lem:PN-trace-equiv}
  Let $\netN$ be an occurrence net,
  and let $\markM$ be a reachable marking. 
  If $\markM \trans{\vec{\trT[1]}}{\!}^*\, \markMi$ and 
  $\markM \trans{\vec{\trT[2]}}{\!}^* \, \markMi$, 
  then $\vec{\trT[1]} \equivSeq \vec{\trT[2]}$.
\end{lem}
\begin{proof}
  We proceed by induction on the length of the longest reduction among
  $\markM \trans{\vec{\trT[1]}}{\!}^* \markMi$ and 
  $\markM \trans{\vec{\trT[2]}}{\!}^* \markMi$.
  For the base case, the thesis is trivial as both $\vec{\trT[1]}$ and 
  $\vec{\trT[2]}$ are empty.
  For the induction case, assume that $\vec{\trT[1]}$ 
  is longer or equal to $\vec{\trT[2]}$ (the other case is symmetric). 
  Let $\vec{\trT[1]} = \trT[1]\vec{\trTi[1]}$.
  We first show that $\vec{\trT[2]}$ is not empty. 
  By contradiction, if $\vec{\trT[2]}$ is empty, then $\markM = \markMi$. 
  But then, by item \ref{lem:occurrenceNet-LTS:item-acyclic} of
  Lemma \ref{lem:occurrenceNet-LTS} (acyclicity) it follows that $\vec{\trT[1]}$ is
  empty as well: contradiction. 
  Therefore, $\vec{\trT[2]} = \trT[2]\vec{\trTi[2]}$ for some $\trT[2]$ and 
  $\vec{\trTi[2]}$. 
  Clearly, $\vec{\trT[1]}$ is longer
  than $\vec{\trTi[1]}$ and $\vec{\trTi[2]}$. 
  Let $\markM \trans{\trT[1]} \markM[1]$ and $\markM \trans{\trT[2]} \markM[2]$.
  We have two subcases.
  \begin{itemize}
  \item If $\trT[1] = \trT[2]$, 
    by determinism (Lemma \ref{lem:occurrenceNet-LTS}) 
    it follows that $\markM[1] = \markM[2]$.
    Let $\markMii = \markM[1]$.
    By the hypothesis of the~\namecref{lem:PN-trace-equiv}, we have
    $\markMii \trans{\vec{\trTi[1]}}{\!\!\!}^* \; \markMi$ and 
    $\markMii \trans{\vec{\trTi[2]}}{\!\!\!}^* \; \markMi$. 
    Then, by the induction hypothesis we have
    $\vec{\trTi[1]} \equivSeq \vec{\trTi[2]}$,
    and so the thesis $\trT[1] \vec{\trTi[1]} \equivSeq \trT[2] \vec{\trTi[2]}$
    follows since $\equivSeq$ is a congruence.
  \item If $\trT[1] \neq \trT[2]$, then by Definition~\ref{def:independent} it must be
    $\trT[1] \independent \trT[2]$. 
    By the diamond property (Lemma \ref{lem:occurrenceNet-LTS}), 
    there exists $\markMii$
    such that $\markM[1] \trans{\trT[2]} \markMii$ and 
    $\markM[2] \trans{\trT[1]} \markMii$. 
    By linearity
    (item~\ref{lem:occurrenceNet-LTS:item-distinct} of Lemma \ref{lem:occurrenceNet-LTS}), 
    $\markMi \nottrans{\;\;\trT[1]}$ and
    $\markMi \nottrans{\;\;\trT[2]}$. 
    By Lemma \ref{lem:occurrenceNet-reach}, applied on $\markM[2]$, 
    there exists $\vec{\trT}$ such that
    $\markMii \trans{\vec{\trT}}{\!\!}^* \; \markMi$ 
    and $\card{\vec{\trT}} + 1 = \card{\vec{\trTi[2]}}$. 
    So, we are in the following situation:
    
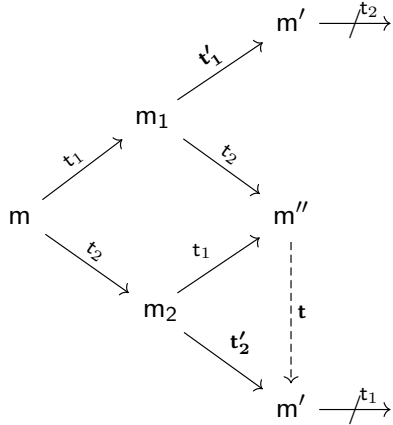
\begin{center}
      \begin{tikzcd}[column sep=normal, row sep=normal, negated/.style={
          decoration={markings,
            mark= at position 0.5 with {
              \node[transform shape] (tempnode) {$\slash$};
            }
          },
          postaction={decorate}
        }]
        & & \markMi 
        \ar[r, sloped, negated, "{\quad\trT[2]}"] & {} \\
        & \markM[1] 
        \ar[ru, sloped, "\vec{\trTi[1]}"] 
        \ar[rd, sloped, "{\trT[2]}"] & & \\
        \markM
        \ar[ru, sloped, "{\trT[1]}"]
        \ar[rd, sloped, "{\trT[2]}"]
        % \ar[dr, start anchor=south east, end anchor=north west,"{\trT}", sloped, ""' near end]
        & 
        & \markMii
        % \ar[uu, dashed, "{\trTii[1]}"] 
        \ar[dd, dashed, "\vec{\trT}"] 
        \\
        & \;\; \markM[2]
        \ar[ru, "{\trT[1]}"]
        \ar[rd, "\vec{\trTi[2]}"]
        \\
        & & \markMi 
        \ar[r, sloped, negated, "{\quad\trT[1]}"]
        & {}
      \end{tikzcd}
    \end{center}

    Therefore, we have that:
    \begin{align*}
      & \markM[1] \trans{\trT[2]} \trans{\vec{\trT}}{\!\!}^* \; \markMi
        \quad\text{and}\quad
        \markM[1] \trans{\vec{\trTi[1]}}{\!\!}^* \; \markMi
      \\
      & \markM[2] \trans{\trT[1]} \trans{\vec{\trT}}{\!\!}^* \; \markMi
        \quad\text{and}\quad
        \markM[2] \trans{\vec{\trTi[2]}}{\!\!}^* \; \markMi
    \end{align*}
    Notice that  
    $\card{\trT[2]\vec{\trT}} = \card{\trT[1]\vec{\trT}} = \card{\vec{\trT}} + 1 
    = \card{\vec{\trTi[2]}} \leq  \card{\vec{\trTi[1]}} < \card{\vec{\trT[1]}}$.
    Hence, by applying the induction hypothesis twice:
    \[
    \trT[2]\vec{\trT} \equivSeq \vec{\trTi[1]} 
    \quad \text{and} \quad
    \trT[1]\vec{\trT} \equivSeq \vec{\trTi[2]}
    \]
    Then, since $\equivSeq$ is a congruence:
    \[
    \trT[1]\trT[2]\vec{\trT} \equivSeq \trT[1]\vec{\trTi[1]} \;\;\text{and}\;\;
    \trT[2]\trT[1]\vec{\trT} \equivSeq \trT[2]\vec{\trTi[2]}
    \]
    Since $\trT[1] \independent \trT[2]$, 
    then $\trT[1]\trT[2]\vec{\trT} \equivSeq \trT[2]\trT[1]\vec{\trT}$. 
    By transitivity of $\equivSeq$:
    \[
    \vec{\trT[1]} 
    \; = \;
    \trT[1]\vec{\trTi[1]} 
    \; \equivSeq \;
    \trT[1] \trT[2] \vec{\trT} 
    \; \equivSeq \;
    \trT[2] \trT[1] \vec{\trT} 
    \; \equivSeq \;
    \trT[2]\vec{\trTi[2]} 
    \; = \;
    \vec{\trT[2]}
    \qedhere
    \]
  \end{itemize}
\end{proof}

\begin{defi}
  \label{def:petri:trOfTrSU}
  For all sequences of transitions $\vec{\trT}$, we define 
  the set $\trOfTrSU{\vec{\trT}}$ of the transitions 
  occurring in $\vec{\trT}$ as:
  \[
  \trOfTrSU{\vec{\trT}} = 
  \setcomp{\trT}{\exists \vec{\trT[1]},\vec{\trT[2]}: 
    \vec{\trT} = \vec{\trT[1]}\trT\vec{\trT[2]}}
  \]
  and we extend $\trOfTrSU{}$ to step firing sequences $\vec{\trSU}$ as follows:
  \[
  \trOfTrSU{\vec{\trSU}} = 
  \bigcup \setcomp{\trSU}{\exists \vec{\trSU[1]},\vec{\trSU[2]}: 
    \vec{\trSU} = \vec{\trSU[1]}\trSU\vec{\trSU[2]}}
  \]
\end{defi}

\begin{lem}
  \label{lem:equivSeq-implies-equal-sets}
  If $\vec{\trT} \equivSeq \vec{\trTi}$ then
  $\trOfTrSU{\vec{\trT}} = \trOfTrSU{\vec{\trTi}}$.
\end{lem}
\begin{proof}
  Trivial by Definition~\ref{def:independent}.
\end{proof}

\begin{lem}
  \label{lem:equal-Tr-implies-equal-marking}
  Let $\netN$ be an occurrence net,
  and let $\markM$ be a reachable marking.
  If $\markM \trans{\vec{\trSU[1]}} \markM[1]$, 
  $\markM \trans{\vec{\trSU[2]}} \markM[2]$ and 
  $\trOfTrSU{\vec{\trSU[1]}} = \trOfTrSU{\vec{\trSU[2]}}$, then
  $\markM[1] = \markM[2]$.
\end{lem}
\begin{proof}
  Since $\markM \trans{\vec{\trSU[1]}} \markM[1]$ and 
  $\markM \trans{\vec{\trSU[2]}} \markM[2]$, there exist
  sequentialisations $\vec{\trT[1]}$ of $\vec{\trSU[1]}$ 
  and $\vec{\trT[2]}$ of $\vec{\trSU[2]}$ such that
  $\markM \trans{\vec{\trT[1]}} \markM[1]$ and
  $\markM \trans{\vec{\trT[2]}} \markM[2]$.
  Since by hypothesis
  $\trOfTrSU{\vec{\trSU[1]}} = \trOfTrSU{\vec{\trSU[2]}}$, then
  $\trOfTrSU{\vec{\trT[1]}} = \trOfTrSU{\vec{\trT[2]}}$.
  We proceed by induction on the length of $\vec{\trT[1]}$.
  The base case is trivial, as $\markM = \markM[1] = \markM[2]$.
  For the inductive case, suppose $\vec{\trT[1]} = \trT[1]\vec{\trTi[1]}$,
  with $\card{\vec{\trTi[1]}} = n$. 
  By determinism, there exists a unique marking $\markMi[1]$
  such that $\markM \trans{\trT[1]} \markMi[1]$ (a single step).
  Since $\trOfTrSU{\vec{\trSU[1]}} = 
  \trOfTrSU{\vec{\trSU[2]}}$, it must be $\vec{\trT[2]} = \trT[2]\vec{\trTi[2]}$,
  with $\card{\vec{\trTi[2]}} = n$. 
  Let $\markMi[2]$ be the unique
  marking such that $\markM \trans{\trT[2]} \markMi[2]$ (a single step).

  \smallskip\noindent
  There are two subcases.
  \begin{itemize}
  \item If $\trT[1] = \trT[2]$, then $\markMi[1] = \markMi[2]$, and so
    the thesis follows directly by the induction hypothesis. 
  \item If $\trT[1] \neq \trT[2]$,
    by the diamond property 
    (item \ref{lem:occurrenceNet-LTS:item-determinism} of Lemma~\ref{lem:occurrenceNet-LTS}),
    there exists $\markMi$ such that $\markMi[1] \trans{\trT[2]} \markMi$ and
    $\markMi[2] \trans{\trT[1]} \markMi$. Since $\trT[2] \in
    \trOfTrSU{\vec{\trTi[1]}}$, by linearity
    (item \ref{lem:occurrenceNet-LTS:item-distinct} of Lemma~\ref{lem:occurrenceNet-LTS})
    it follows that
    $\markM[1] \nottrans{\;\;\trT[2]}$, and hence, 
    by applying Lemma \ref{lem:occurrenceNet-reach} on $\markMi[1]$
    we obtain $\markMi \trans{\vec{\trTii[1]}} \markM[1]$ 
    for some $\vec{\trTii[1]}$. Summing up, we have that:
    \[\markMi[1] \trans{\vec{\trTi[1]}} \markM[1] \qquad \text{and} \qquad
    \markMi[1] \trans{\trT[2]\vec{\trTii[1]}} \markM[1]\]
    Then, by Lemma \ref{lem:PN-trace-equiv}, 
    $\vec{\trTi[1]} \equivSeq \trT[2]\vec{\trTii[1]}$,
    and hence:
    \[
    \vec{\trT[1]} = \trT[1]\vec{\trTi[1]} \equivSeq 
    \trT[1]\trT[2]\vec{\trTii[1]}
    \]
    By Lemma \ref{lem:equivSeq-implies-equal-sets}:
    \[
    \trOfTrSU{\vec{\trT[1]}} = \trOfTrSU{\trT[1]\trT[2]\vec{\trTii[1]}}
    \]
    Similarly, we can conclude that $\markMi \trans{\vec{\trTii[2]}} \markM[2]$ 
    for some $\vec{\trTii[2]}$ and that:
    \[
    \trOfTrSU{\vec{\trT[2]}} = \trOfTrSU{\trT[2]\trT[1]\vec{\trTii[2]}}
    \]
    Since $\trOfTrSU{\vec{\trT[1]}} = \trOfTrSU{\vec{\trT[2]}}$,
    we can conclude:
    \[\trOfTrSU{\vec{\trTii[1]}} = \trOfTrSU{\vec{\trTii[2]}}\]
    Since $\card{\vec{\trTii[1]}} = n - 1 < n + 1 = \card{\vec{\trT[1]}}$,
    the thesis follows by the induction hypothesis.
    \qedhere
  \end{itemize}
\end{proof}

\begin{lem}
  \label{lem:PN-independent-implies-swap}
  Let $(\txT,i)$ and $(\txTi,j)$ be transitions of $\PNet{\relR}{\bcB}$.
  Then:
  \[
  (\txT,i) \independent (\txTi,j)
  \quad \implies \quad
  \txT \relR \txTi
  \]
\end{lem}
\begin{proof}
  By Definition~\ref{def:independent},
  $(\txT,i) \independent (\txTi,j)$ implies that
  $(\txT,i) \neq (\txTi,j)$ and there exists some reachable marking
  $\markM$ such that $\markM \trans{(\txT,i)}$ and $\markM \trans{(\txTi,j)}$.
  By contradiction, assume that $\neg (\txT \relR \txTi)$.
  Then, since $i < j$ or $j > i$, by Definition \ref{def:bc-to-pnet}
  we would have that $(\txT,i) < (\txTi,j)$ or $(\txTi,j) < (\txT,i)$.
  Then, by Lemma \ref{lem:PNet-leq-implies-not-independent} 
  we obtain a contradiction.
\end{proof}

\begin{defi}
  Let $\PNet{\relR}{\bcB} = (\Places, \Transitions, \Arcs, \markM[0])$.
  We define $\txOfTr{} : \Transitions \rightarrow \Tx$ as 
  \(
  \txOfTr{\txT,i} = \txT
  \).
  We then extend $\txOfTr{}$ to a function from steps to 
  multisets of transactions as follows:
  \[
  \txOfTr{\emptyset} = \emptymset 
  \qquad 
  \txOfTr{\trSU \cup \setenum{\trT}} = 
  \msetenum{\txOfTr{\trT}} + \txOfTr{\trSU}\;\;
  \]
  Finally, we extend $\txOfTr{}$ to finite sequences of steps as follows:
  \[
  \txOfTr{\bcEmpty} = \bcEmpty \qquad 
  \txOfTr{\trSU\vec{\trSU}} = \txOfTr{\trSU}\txOfTr{\vec{\trSU}}
  \]
\end{defi}

\begin{lem}
  \label{lem:petri-equiv-par-block}
  Let $\PNet{\relR}{\bcB} = (\Places, \Transitions, \Arcs, \markM[0])$,
  and let $\vec{\trSU}$ be a step firing sequence. 
  Then, for all $\WR{}$ and $\bcSt$:
  \[
  \msem{\WR{}}{\bcSt}{\vec{\trSU}} \; = \; \msem{\WR{}}{\bcSt}{\txOfTr{\vec{\trSU}}}
  \]
\end{lem}
\begin{proof}
  Straightforward by Definitions \ref{def:sem:step:seq} and \ref{def:txpar:pbc-semantics}.
\end{proof}

\begin{lem}
  \label{lem:equivSeq-implies-equivStSeq}
  If $\vec{\trT} \equivSeq \vec{\trTi}$ holds in $\PNet{\relR}{\bcB}$, then 
  $\txOfTr{\vec{\trT}} \equivStSeq[\relR] \txOfTr{\vec{\trTi}}$.
\end{lem}
\begin{proof}
  Define: 
  \[
  \equivSeq' 
  \; = \;
  \setcomp
  {(\vec{\trT},\vec{\trTi})}
  {\txOfTr{\vec{\trT}} \equivStSeq[\relR] \txOfTr{\vec{\trTi}}}
  \] 
  It suffices to show that $\equivSeq\; \subseteq\; \equivSeq'$.
  Notice that $\equivSeq'$ is a congruence satisfying:
  \[
  \txT \relR \txTi \;\;\implies\;\; 
  (\txT,i)(\txTi,j) \equivSeq' (\txTi,j)(\txT,i)
  \]
  But then, by Lemma \ref{lem:PN-independent-implies-swap}, 
  it follows that $\equivSeq'$ also satisfies:
  \begin{equation*}
    (\txT,i) \independent (\txTi,j) \;\;\implies\;\; 
    (\txT,i)(\txTi,j) \equivSeq' (\txTi,j)(\txT,i)
  \end{equation*}
  Since $\equivSeq$ is the smallest congruence satisfying
  this implication, we have
  $\equivSeq\; \subseteq\; \equivSeq'$.
\end{proof}

\thbctopnet*
\begin{proof}
  For item~\ref{th:bc-to-pnet:confluence},
  assume that $\markM[0] \trans{\vec{\trSU}} \markM$ and 
  $\markM[0] \trans{\vec{\trSUi}} \markM$. 
  A standard result from Petri nets theory ensures that
  there exists sequentializations $\vec{\trT}$ of $\vec{\trSU}$
  and $\vec{\trTi}$ of $\vec{\trSUi}$ such that:
  \[
  \markM[0] \trans{\vec{\trT}} \markM
  \quad \text{and} \quad
  \markM[0] \trans{\vec{\trTi}} \markM
  \]
  By Lemma \ref{lem:PN-trace-equiv},
  it must be $\vec{\trT} \equivSeq \vec{\trTi}$.
  Then, by Lemma \ref{lem:equivSeq-implies-equivStSeq}: 
  \[\txOfTr{\vec{\trT}} \equivStSeq[\pswapWR] \txOfTr{\vec{\trTi}}\]
  By Lemma \ref{lem:equivRel-implies-equiv}:
  \[
  \semBc{\txOfTr{\vec{\trT}}}{\bcSt} = \semBc{\txOfTr{\vec{\trTi}}}{\bcSt}
  \]
  By Lemmas \ref{lem:PN-mstep-indepenent} and \ref{lem:PN-independent-implies-swap},
  it follows that all multisets of transactions in $\txOfTr{\vec{\trSU}}$,
  as well as those in $\txOfTr{\vec{\trSUi}}$,
  are strongly swappable \wrt $\pswapWR$.
  Therefore, by Theorem \ref{th:seq-union}:
  \[
  \msem{\WRmin{}}{\bcSt}{\txOfTr{\vec{\trSU}}} = 
  \msem{\WRmin{}}{\bcSt}{\txOfTr{\vec{\trSUi}}}
  \]
  Then, by Lemma \ref{lem:petri-equiv-par-block}:
  \[
  \msem{\WRmin{}}{\bcSt}{\vec{\trSU}} = 
  \msem{\WRmin{}}{\bcSt}{\vec{\trSUi}}
  \]
  
  \smallskip\noindent
  For item~\ref{th:bc-to-pnet:bc}, note that a transition $(\txT[i],i)$ 
  is enabled if all transitions $(\txT[j],j)$ with $j < i$ have been fired.  
  So $\setenum{(\txT[1],1)} \cdots \setenum{(\txT[n],n)}$ is a step firing sequence. Moreover, $\setenum{(\txT[1],1)} \cdots \setenum{(\txT[n],n)}$
  contains all the transactions of $\PNet{\pswapWR}{\bcB}$, and so, by linearity (Lemma \ref{lem:occurrenceNet-LTS:item-distinct}) it is maximal.
  
  \smallskip\noindent
  For item~\ref{th:bc-to-pnet:maximal},
  let $\vec{\trSUi} = \setenum{(\txT[1],1)} \cdots \setenum{(\txT[n],n)}$. By item~\ref{th:bc-to-pnet:bc}, we have that $\vec{\trSUi}$ is a maximal
  step firing sequence. It is easy to see that $\bcB \seqn \txOfTr{\vec{\trSUi}}$. 
  By Theorem \ref{th:seq-union}: 
  \[\msem{\WRmin{}}{\bcSt}{\txOfTr{\vec{\trSUi}}} = \semBc{\bcB}{\bcSt}\]
  Since $\vec{\trSU}$ and $\vec{\trSUi}$ are both maximal, by 
  Lemma \ref{lem:equal-Tr-implies-equal-marking} and 
  by item~\ref{th:bc-to-pnet:confluence}, it follows that:  
  \[
  \msem{\WRmin{}}{\bcSt}{\vec{\trSUi}} = \msem{\WRmin{}}{\bcSt}{\vec{\trSU}}
  \]
  Since $\msem{\WRmin{}}{\bcSt}{\vec{\trSUi}} = 
  \msem{\WRmin{}}{\bcSt}{\txOfTr{\vec{\trSUi}}}$ 
  (by Lemma \ref{lem:petri-equiv-par-block}) we have that:
  \[
  \msem{\WRmin{}}{\bcSt}{\vec{\trSU}} = \semBc{\bcB}{\bcSt}
  \tag*{\qedhere}
  \]
\end{proof}

\begin{lstlisting}[label={ex:eth:token},caption={A simple ERC-721 token implementation}, language=solidity,numbersep=10pt,float]
pragma solidity >= 0.4.2;

contract Token {
    
  mapping(uint256 => address) owner;
  mapping(uint256 => address) approved;
  mapping(uint256 => bool) exists;
  mapping(address => uint256) balance;
  mapping (address => mapping (address => bool)) opApprovals;

  function ownerOf(uint256 tkId) external view returns (address) {
    require (exists[tkId]); 
    require (owner[tkId] != address(0));
    return owner[tkId];
  }
    
  function balanceOf(address addr) external view returns (uint256) {
    require (addr != address(0));
    return balance[addr];
  }

  function approve(address addr, uint256 tkId) external {
    require (exists[tkId]);
    require (owner[tkId] == msg.sender && addr != msg.sender);
    approved[tkId] = addr;
  }

  function setApprovalForAll(address op, bool isApproved) external {
    operatorApprovals[msg.sender][op] = isApproved;
  }
	
  function getApproved(uint256 tkId) external view returns (address) {
    require (exists[tkId]);
    return approved[tkId];
  }

  function isApprovedForAll(address addr, address op) external view returns (bool) {
    return operatorApprovals[addr][op];
  }
	
  function transferFrom(address from, address to, uint256 tkId) external {
    require (exists[tkId]);
    require (from == owner[tkId] && from != to);
    require (to != address(0));
    if (from == msg.sender || operatorApprovals[from][msg.sender] || approved[tkId] == msg.sender) {
      owner[tkId] = to;
      approved[tkId] = address(0);
      balance[from] -= 1;
      balance[to] += 1;
    }
  }
    
  function mint(address to, uint256 tkId) external {
    require (!exists[tkId]);
    require (to != address(0));       
    exists[tkId] = true;
    owner[tkId] = to;
    balance[to] += 1;
  }
}
\end{lstlisting}

\begin{lstlisting}[label={ex:eth:lottery},caption={A two-players lottery contract}, language=solidity,numbersep=10pt,float]
pragma solidity >=0.4.22 <0.6.0;

contract Lottery {
    address public owner;
    address payable player0; address payable player1; 
    address payable winner;
    bytes32 hash0; bytes32 hash1;
    string secret0; string secret1;

    constructor() public {
        owner = msg.sender;
    }

    function join0() payable public {
        require (player0==address(0) && msg.value > .01 ether);
        player0 = msg.sender;
    }

    function join1() payable public {
        require (player1==address(0) && msg.value > .01 ether);
        player1 = msg.sender;
    }

    function commit0(bytes32 h) public {
        require (msg.sender==player0 && hash0==0);
        hash0 = h;
    }

    function commit1(bytes32 h) public {
        require (msg.sender==player1 && hash1==0);
        hash1 = h;
    }

    function reveal0(string memory s) public {
        require (msg.sender==player0);
        require (hash0!=0 && hash1!=0 &&  hash0 != hash1);
        require(keccak256(abi.encodePacked(s))==hash0);
        secret0 = s;
    }

    function reveal1(string memory s) public {
        require (msg.sender==player1);
        require (hash0!=0 && hash1!=0 &&  hash0 != hash1);
        require(keccak256(abi.encodePacked(s))==hash1);
        secret1 = s;
    }

    function win() public {
        uint256 l0 = bytes(secret0).length;
        uint256 l1 = bytes(secret1).length;
        require (l0!=0 && l1!=0);
        if ((l0+l1) % 2 == 0) {
            winner = player0;
        }
        else {
            winner = player1;
        }
        winner.transfer(address(this).balance);

        // reset state for next round
        player0 = player1 = address(0);
        hash0 = hash1 = 0;
        secret0 = secret1 = "";
    }
}
\end{lstlisting}

\end{document}